\documentclass[journal,onecolumn,draftclsnofoot,]{IEEEtran}

\IEEEoverridecommandlockouts
\usepackage[utf8]{inputenc}
\usepackage[english]{babel}
\usepackage{epsfig}
\usepackage{amsfonts}
\usepackage{ifpdf}
\usepackage{multirow}
\usepackage{amssymb}
\usepackage{amsthm}
\usepackage{mathtools}
\usepackage{mathrsfs}
\usepackage{textgreek}
\usepackage{grffile}
\usepackage{xcolor}
\usepackage{bm}
\usepackage{cite}
\usepackage{setspace}
\usepackage{nccmath}

\usepackage[usestackEOL]{stackengine}

  \stackMath

\ifCLASSINFOpdf
\else
\fi
\usepackage{amsmath}
\usepackage{algorithmic}
\usepackage{algorithm}
\usepackage{graphicx}

\newtheorem{definition}{Definition}

\newtheorem{remark}{Remark}
\newtheorem{lemma}{Lemma}

\newcommand{\norm}[1]{\left\lVert#1\right\rVert}

\ifCLASSOPTIONcompsoc
  \usepackage[caption=false,font=normalsize,labelfont=sf,textfont=sf]{subfig}
\else
  \usepackage[caption=false,font=footnotesize]{subfig}
\fi

\usepackage{stfloats}
\begin{document}
%


\title{Joint $9$D Receiver Localization and Ephemeris Correction using LEO and $5$G Base Stations}
\author{Don-Roberts~Emenonye, Wasif J. Hussain,
        Harpreet~S.~Dhillon,
        and~R.~Michael~Buehrer
\thanks{D.-R. Emenonye, Wasif. J. Hussain, H. S. Dhillon, and R. M.  Buehrer are with Wireless@VT,  Bradley Department of Electrical and Computer Engineering, Virginia Tech,  Blacksburg,
VA, 24061, USA. Email: \{donroberts, wasif, hdhillon, rbuehrer\}@vt.edu. The support of the US National Science Foundation (Grants ECCS-2030215, CNS-1923807, and CNS-2107276) is gratefully acknowledged. Part of this work has been accepted to IEEE MILCOM 2024, Washington, DC, USA \cite{emenonye2023_MILCOM_conf_9D_localization,emenonye2023_MILCOM_conf_9D_localization_1}.
}
}

\maketitle
\IEEEpeerreviewmaketitle

\begin{abstract}
This paper leverages Fisher information to examine the interaction between low-Earth orbit (LEO) satellites and 5G base stations (BSs) in enabling 9D receiver localization and refining LEO ephemeris. First, we propose a channel model that incorporates all relevant links: LEO-receiver, LEO-BS, and BS-receiver. Then, we utilize the Fisher information matrix (FIM) to quantify the information available about the channel parameters in these links. By transforming these FIMs, we derive the FIM for 9D receiver localization parameters—comprising 3D position, 3D orientation, and 3D velocity—along with LEO position and velocity offsets. We present closed-form expressions for the FIM entries corresponding to these localization parameters. Our identifiability analysis based on the FIM reveals that:
i) With a single LEO, three BSs, and three time slots are required to estimate the 9D localization parameters and correct the LEO position and velocity.
ii) With two LEOs, the same configuration (three BSs and three time slots) suffices for both tasks.
iii) With three LEOs, three BSs and four time slots are needed to achieve the same goal. A key insight from the Cramér-Rao lower bound (CRLB) analysis is that, under the configuration of one LEO, three BSs, and three time slots, the estimated errors for receiver positioning, velocity, and orientation, as well as LEO position and velocity offsets, are $0.1 \text{ cm}$, $1 \text{ mm/s}$, $10^{-3} \text{ rad}$, $0.01 \text{ m}$, and $1 \text{ m/s}$, respectively. The receiver localization parameters are estimated after $1 \text{ s}$ while the LEO offset parameters are estimated after $20 \text{ s}.$ Additionally, our CRLB analysis indicates that operating frequency has minimal impact on receiver orientation estimation accuracy, and the number of receive antennas has a negligible effect on LEO velocity estimation accuracy.



\end{abstract}

\section{Introduction}

The challenge of providing quality global connectivity has led to the deployment of additional satellites in existing low-earth orbit (LEO) constellations and the creation of new constellations. This can be seen in older constellations such as Orbcomm, Iridium, and Globalstar, as well as in newly deployed constellations such as Boeing, SpaceMobile, OneWeb, Telesat, Kuiper, and Starlink. Since these constellations will be closer to the earth than medium earth orbit (MEO) satellites, they will encounter a shorter propagation time and lower propagation losses than the current Global Navigation Satellite System (GNSS); their utility for positioning is an area of natural interest. Moreover, these LEOs could be used as a backup during the inevitable scenarios where GNSS becomes unavailable, such as in deep urban canyons, under dense foliage, during unintentional interference, and intentional jamming.  Because of these reasons, there has been a surge in research on the utility of LEOs for localization.

Current research ranges from opportunistic approaches, where the signal structure is unknown or partially known, to dedicated approaches, where the structure is fully known. The current state of the art has failed to rigorously characterize the interplay between the information from the LEOs and terrestrial transceivers, such as $5$G base station (BSs) for both $9$D localization and LEO ephemeris correction.  Hence, in this paper, we utilize the Fisher information matrix (FIM) to rigorously present the available information in the LEO-receiver link, LEO-BS link, and BS-receiver and the utility of this information for joint $9$D receiver localization and LEO ephemeris correction, which leads to critical insights into the interplay between number of LEOs, BSs, operating frequency, number of transmission time slots, the combination of synchronized BSs and unsynchronized LEOs, and the number of receive antennas.

\subsection{Prior Art}
The following three research directions are of interest to this paper: i) LEO-based localization, ii) Localization using large antenna arrays, and iii) Opportunistic localization using $5$G BSs. We now summarize the relevant works in these directions.
\subsubsection{LEO-based localization}
Research on LEO-based localization varies from the dedicated signals \cite{Fundamentals_of_LEO_Based_Localization,Fundamental_Performance_Bounds_for_Carrier_Phase_Positioning_LEO_PNT,Broadband_LEO_Constellations_for_Navigation,Economical_Fused_LEO_GNSS,Empowering_the_Tracking_Performance_of_LEOBased_Positioning_by_Means_of_Meta_Signals,Performance_Analysis_of_a_Multi_Slope_Chirp_Spread_Spectrum,Integrated_Communications_and_Localization_for_Massive_MIMO_LEO_Satellite} to the opportunistic signals \cite{Psiaki2020NavigationUC,Kassas2019NewAgeSN,Navigation_With_Differential_Carrier_Phase_Measurements_From_Megaconstellation_LEO_Satellites,A_Hybrid_Analytical_Machine_Learning_Approach_for_LEO_Satellite_Orbit_Prediction,Doppler_effect_Downlink_Receivers_OFDM_Low_earth_orbit_satellites_Bandwidth_Synchronization_Doppler_positioning_low_Earth_orbit,Ad_Astra_STAN_With_Megaconstellation_LEO_Satellites,A_Hybrid_Analytical_Machine_Learning_Approach_for_LEO_Satellite_Orbit_Prediction_1,Assessing_Machine_Learning_for_LEO_Satellite_Orbit_Determination_in_Simultaneous_Tracking_and_Navigation,Cognitive_Navigation_With_Unknown_OFDM_signals_With_Application_Terrestrial_5G_Starlink,Observability_Analysis_of_Receiver_Localization_Pseudorange,Positioning_with_Starlink_LEO_Satellites_A_Blind_Doppler_Spectral_Approach,Receiver_Design_for_Doppler_Positioning_with_Leo_Satellites,Unveiling_Starlink_LEO_Satellite_OFDM_Like_Signal_Structure_Enabling_Precise_Positioning,dureppagari2023ntn,dureppagari2024leo}. On the opportunistic end, the signal structure (length, values, and periodicity) is completely unknown, while on the dedicated end, the signal structure is known. In \cite{Fundamentals_of_LEO_Based_Localization}, a FIM based rigorous investigation of the utility of LEOs for $9$D localization is presented where it is shown that obtaining delay and Doppler measurements from three satellites over three times slots using multiple receive antenna enables $9$D localization. In \cite{Fundamental_Performance_Bounds_for_Carrier_Phase_Positioning_LEO_PNT}, the signal structure is assumed to be known, and delay measurements are used to localize a receiver. The authors in \cite{Broadband_LEO_Constellations_for_Navigation} investigate using satellites deployed to provide broadband Internet connectivity to assist localization. The proposed framework in \cite{Broadband_LEO_Constellations_for_Navigation} uses delay measurements and describe the positioning errors as a function of the geometric dilution of precision (GDOP) to provide a benchmark. In \cite{Empowering_the_Tracking_Performance_of_LEOBased_Positioning_by_Means_of_Meta_Signals},  Doppler measurements obtained from Amazon Kuiper Satellites are used for receiver positioning. In \cite{Integrated_Communications_and_Localization_for_Massive_MIMO_LEO_Satellite}, LEOs met integrated sensing and localization, and the positioning information obtained from the LEOs is used to improve the transmission rate. The authors in \cite{Psiaki2020NavigationUC} utilize eight Doppler measurements to estimate the $3$D position, $3$D velocity, clock rate, and clock offset. In \cite{Kassas2019NewAgeSN}, an opportunistic experimental framework is developed to estimate position, clock, and correct LEO ephemeris. An unmanned
aerial vehicle (UAV) is tracked for two minutes using the received signals from two Orbcomm satellites in \cite{Navigation_With_Differential_Carrier_Phase_Measurements_From_Megaconstellation_LEO_Satellites}. In \cite{A_Hybrid_Analytical_Machine_Learning_Approach_for_LEO_Satellite_Orbit_Prediction}, a machine learning framework is developed for LEO ephemeris correction as well as receiver positioning utilizing Doppler measurements from two Orbcomm satellites. A detection algorithm is built, and six Starlink satellites are detected and used to localize a
ground receiver to an error of $20 \text{ m}.$ In \cite{Ad_Astra_STAN_With_Megaconstellation_LEO_Satellites}, a framework is developed that integrates IMUs, delay, and Doppler measurements. A receiver positioning error of $27.1 \text{ m}$ and $18.4 \text{ m}$ is achieved with two Orbcomm, one Iridium and three Starlink satellites, respectively. The authors in \cite{A_Hybrid_Analytical_Machine_Learning_Approach_for_LEO_Satellite_Orbit_Prediction_1} develop a machine-learning framework that localizes itself through Doppler measurements from a single satellite over multiple time slots. In \cite{Assessing_Machine_Learning_for_LEO_Satellite_Orbit_Determination_in_Simultaneous_Tracking_and_Navigation}, an opportunistic framework is developed to utilize Doppler for joint receiver positioning and LEO ephemeris correction. The authors develop a Doppler-based framework in \cite{Cognitive_Navigation_With_Unknown_OFDM_signals_With_Application_Terrestrial_5G_Starlink} that jointly use $5$G BSs and LEOs to localize a receiver. In \cite{Observability_Analysis_of_Receiver_Localization_Pseudorange}, a single Orbocmm satellite with a known orbit is used to localize a receiver. The authors in \cite{Positioning_with_Starlink_LEO_Satellites_A_Blind_Doppler_Spectral_Approach} characterize the received signal, propose a Doppler discriminator, and use the Dopplers to position a receiver with an error of $\text{4.3 m}.$ The received signal characterization in that work accounts for the extremely dynamic nature of the channel due to the speed of the satellites. In \cite{Receiver_Design_for_Doppler_Positioning_with_Leo_Satellites}, the LEO ephemeris is assumed to be perfectly known, and then Doppler measurements from two Orbcomm satellites are used to position a receiver to an error of $11 \text{ m}.$ The authors in \cite{Unveiling_Starlink_LEO_Satellite_OFDM_Like_Signal_Structure_Enabling_Precise_Positioning} investigate the signal structure of the satellites and discover that some use tones while others use Orthogonal Frequency Division Multiplexing (OFDM). Subsequently, the signals are used to position a receiver to an error of $6.5 \text{ m}$.
{\em Although there has been a surge in research using LEOs for localization, the interplay of LEO signals and signals from $5$G BSs has yet to be studied. Hence, in this work, we utilize the FIM to study the information in the LEO-receiver, LEO-BS, and BS-receiver links and their utility for joint $9$D localization and LEO ephemeris correction.}

\subsubsection{Localization using large antenna arrays}
The need for more bandwidth has mandated the usage of higher center frequencies. The corresponding small wavelengths, in turn, have enabled more elements/antennas on the arrays. The resulting large antenna arrays have been investigated for localization purposes and due to the sparsity of the channels in systems with large antenna arrays, the received signal can be parameterized by the angle of departure (AOD), angle of arrival (AOA), and time of arrival (TOA), and localization using this parameterization has been studied in \cite{garcia2017direct,8240645,8515231,8356190,guerra2018single,emenonye2023limits,fascista2021downlink,8755880,li2019massive,9082200}. In \cite{garcia2017direct}, distributed anchors are used to measure the TOAs and AOAs from a single agent. The TOAs restrict the agent's position to a convex set, while subsequently, the AOAs provide the position estimate. The authors in \cite{8240645} provide a seminal contribution to this area, and the FIM for the estimation of AOD, AOA, and TOA are provided. Subsequently, the FIM of channel parameters is transformed into the FIM of the location parameter, and the Cramer Rao lower bound (CRLB) is obtained along with an algorithm that achieves the bound at a high signal-to-noise ratio. The authors in \cite{8515231} utilize the bounds in \cite{8240645} to show that the presence of non-line of sight parameters does not decrease the information available for localization. The authors in \cite{8356190} extend the bounds in \cite{8240645,8515231} from the $2$D case to the $3$D case, after which the FIM is shown to have a definite structure that can be decomposed into the information from the transmitter and the information from the receiver. In \cite{guerra2018single}, the FIM is derived, but for the uplink, and the CRLB is shown to be unique in the limit of the number of BS antennas (this is because all receiver position leads to a unique CRLB). The authors in \cite{emenonye2023limits} show that while in the near-field, the transmitter/receiver $3$D orientation can be estimated, only the transmitter/receiver $2$D orientation can be estimated in the far-field. Also, that paper shows that beamforming is required to estimate the $2$D transmitter orientation in the far-field. The challenging problem of single antenna receiver positioning is tackled with a single observation \cite{fascista2021downlink} and with multiple observations \cite{8755880}. Simultaneous localization and mapping are tackled in \cite{li2019massive} while \cite{9082200} tackles localization under hardware impairments. A different parameterization is presented in \cite{5571889,9606768,7364259}. In \cite{5571889}, a Bayesian approach is developed, incorporating {\em a priori} information about the channel parameters and the locations of the transmitter and receiver. Also, in \cite{5571889}, the non-line of sight component is shown only to contribute when we have {\em a priori} information about them. The work in \cite{5571889} is extended to the cooperative case in\cite{9606768}. The cooperative case is extended to the collaborative case in massive networks \cite{7364259}. Localization with reconfigurable intelligent surfaces is presented in \cite{8264743,9729782,9781656
,9508872,9625826,9500663,9782100,9528041,emenonye2022fundamentals,emenonye2023_ICC_conf_workshop,emenonye2022ris,emenonye2023_ICC_conf, RIS_Aided_Kinematic,OTFS_Enabled_RIS,9774917}. These works can be grouped into i) continuous RIS\cite{8264743,9729782,9781656} and discrete RIS \cite{emenonye2022fundamentals,emenonye2022ris,emenonye2023_ICC_conf,emenonye2023_ICC_conf_workshop, RIS_Aided_Kinematic, OTFS_Enabled_RIS,9508872,9625826,9500663,9782100,9528041,9774917}, and ii) near-field\cite{emenonye2022ris,emenonye2023_ICC_conf, RIS_Aided_Kinematic, OTFS_Enabled_RIS,8264743,9729782,9781656,9508872,9500663,9625826,9774917} and far-field propagation \cite{emenonye2022fundamentals,emenonye2023_ICC_conf_workshop,9782100,9528041}. {\em Although extensive work has been done in these areas, we make the following contributions: i) we present the FIM for the case when a group of anchors is unsynchronized in both time and frequency with the agent while the other group of anchors is synchronized with the agent but unsynchronized with the other group of anchors, and ii) we present the FIM for joint $9$D localization and anchor position and velocity correction.}

\subsubsection{Opportunistic localization using $5$G BSs}
$5$G BSs transmit several signals that have good correlation properties. The transmission of the signals will either be i) always-on or ii) on-demand. Examples of always-on signals are primary and secondary synchronization signals and the physical broadcast channel
block, while examples of on-demand signals are demodulation reference signals, phase tracking reference signals, and sounding reference signals. In \cite{9573365}, a comprehensive receiver is developed that can detect BSs and their associated reference signals (which could be always on or on-demand). In addition, a sequential generalized likelihood ratio detector is used to detect the co-channel BSs. Subsequently, the Dopplers of the BSs are used to define a signal subspace, after which the reference signals are estimated. Finally, a UAV is tracked for over $400 \text{ m}$ with a positioning error of $4.15 \text{ m}.$ In \cite{9369049}, a framework is developed that uses always-on signals for localization. After removing the clock bias from the estimated range measurements, the ranging error standard deviation is calculated as $1.15 \text{ m}$. Finally, authors in \cite{kassas2021carpe}, the always-on signals are obtained, and the observable parameters are extracted using a software-defined radio. A framework utilizing the extended Kalman filter is used to provide the receiver position.

\subsection{Contribution}
\label{Joint_9D_Receiver_Localization_and_Ephemeris_Correction_with_LEO_and_5G_Base_Stations_Contribution}
This paper focuses on the $9$D localization of a receiver and LEO ephemeris correction using the signals in the LEO-receiver, LEO-BSs, and BSs-receiver links. With this setup, our main contributions are:
\subsubsection{Determining the available information about the channel parameters in the LEO-receiver, LEO-BSs, and BSs-receiver links}
We derive the FIM for the channel parameters in these links. To enable these derivations, we develop a channel model that captures: i) the frequency and time offset between a LEO and a receiver, and a LEO and the collection of BSs,
ii) the frequency and time offset between the collection of BSs and the receiver, and iii) the unknown offset in the LEO position and velocity due to outdated LEO ephemeris information.

\subsubsection{Determining the available information for $9$D localization and LEO ephemeris correction}
We transform the FIM for the channel parameters to the FIM for the location parameters ($9$D location parameters and the LEOs' position and velocity). We provide closed-form expressions of all the elements in the FIM for the location parameters. Next, we derive the information loss terms due to i) lack of time and frequency synchronization among the LEOs, ii) lack of time and frequency synchronization between the LEOs and receiver, iii) lack of time and frequency synchronization between the LEOs and BSs, and iv) lack of time and frequency synchronization between the BSs and the receiver. With this information loss term, we compute the equivalent FIM (EFIM), focusing on the parameters of interest. Closed-form expressions of the elements in the EFIM are presented.

\subsubsection{Determining the minimal infrastructure required for $9$D localization and LEO ephemeris correction}

With the EFIM, we perform an identifiability analysis by determining the combination of the number of LEOs, number of BSs, number of receive antennas, and number of transmission time slots that make the EFIM positive definite. Based on the identifiability analyses, we conclude the following: i) with a single LEO, at least three BSs and three time slots are required to estimate the 9D location parameters and correct the LEO’s position and velocity, ii) with two LEOs, a minimum of three BSs and three time slots are necessary to estimate the 9D location parameters and correct the LEO’s position and velocity, and iii) with three LEOs, at least three BSs and four time slots are needed to estimate the 9D location parameters and correct the LEO’s position and velocity. Next, we invert the EFIM and obtain the CRLB. We show that with a single LEO, three time slots, and three BSs, the receiver positioning error, velocity estimation error, orientation error, LEO position offset estimation error, and LEO velocity offset estimation error are $0.1 \text{ cm}$, $1 \text{ mm/s}$, $10^{-3} \text{ rad}$, $0.01 \text{ m}$, and $1 \text{ m/s}$, respectively. The receiver localization parameters are estimated after $1 \text{ s}$ while the LEO offset parameters are estimated after $20 \text{ s}.$ We also notice from the CRLB that the operating frequency and number of receive antennas have negligible impact on the estimation accuracy of the orientation of the receiver and the LEO velocity, respectively.

\textit{Notation:}
 The function $
 \bm{F}_{\bm{v}}(\bm{w} ; \bm{x}, \bm{y}) \triangleq \mathbb{E}_{\bm{v}}\left\{\left[\nabla_{ \bm{x}} \ln f(\bm{w})\right]\left[\nabla_{\bm{y}} \ln f(\bm{w})\right]^{\mathrm{T}}\right\}
 $. $ \bm{G}_{\bm{v}}(\bm{w} ; \bm{x}, \bm{y})$
 describes the loss of information in the FIM defined by $
 \bm{F}_{\bm{v}}(\bm{w} ; \bm{x}, \bm{y})$ due to uncertainty in the nuisance parameters. The inner product of $\bm{x}$ is $\norm{\bm{x}}^2$ and the outer product of $\bm{x}$ is $\norm{\bm{x}^{\mathrm{T}}}^2$. $\nabla_{x} y$ is the first derivative of $y$ with respect to $x$. 
\section{System Model}
We consider $N_B$ single antenna LEO satellites, $N_Q$ single antenna BSs, and a receiver with $N_U$ antennas. The $N_B$ LEOs communicate with the $N_Q$ BSs and the receiver over $N_K$ transmission time slots. Similarly, the $N_Q$ BSs communicate over $N_K$ transmission time slots with the receiver. There is a $\Delta_{t}$ spacing between the $N_K$ transmission slots. The LEOs are located at $\bm{p}_{b,k} \; \; b \in \{1,2,\cdots,N_{B}\} \text{ and} \; \; k \in \{1,2,\cdots,N_{K}\}$ while the BSs are located at $\bm{p}_{q,k} \;  q \in \{1,2,\cdots,N_{Q}\} \text{ and} \;  k \in \{1,2,\cdots,N_{K}\}$ as shown in Fig. \ref{System_model_1}. Finally, the receiver is located at $\bm{p}_{U,k} \; \; k \in \{1,2,\cdots,N_{K}\}.$ 
\begin{figure}
\centering
    \fbox{\includegraphics[clip, trim=9.5cm 7.6cm 10.5cm 2.6cm,scale=0.9]{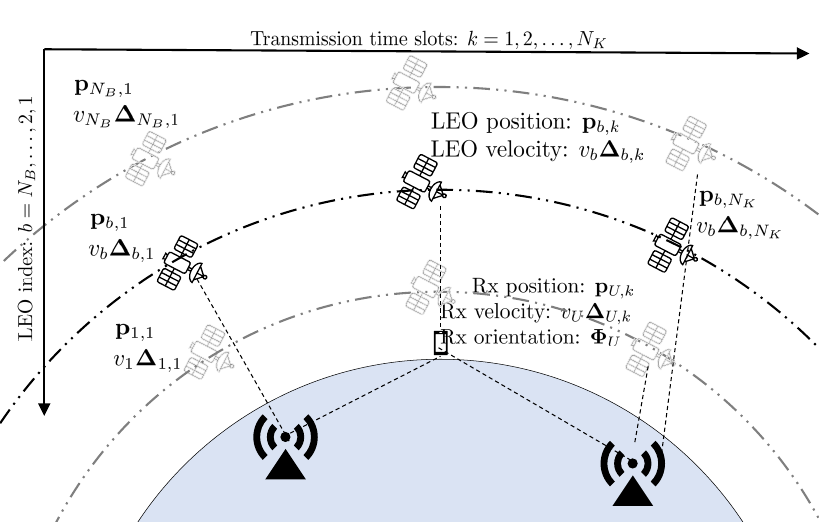}}
    \caption{Joint $9$D Receiver Localization and Ephemeris Correction with $N_B$ LEO and $N_Q$ $5$G Base Stations transmitting during $N_K$ transmission time slots to a receiver with $N_U$ antennas.}
    \label{System_model_1}
\end{figure}

The LEOs, BSs, and receiver positions are defined with respect to a global origin and a global reference axis. The position of the antennas on the receiver can be defined with respect to the receiver's centroid as $\bm{s}_{u} = \bm{Q}_{U}\Tilde{\bm{s}}_{u}$. It aligns with the global reference axis, with $\bm{Q}_{U} = \bm{Q}\left(\alpha_{U}, \psi_{U}, \varphi_{U}\right)$ serving as the 3D rotation matrix \cite{lavalle2006planning}. The orientation angles of the receiver are vectorized as $\bm{\Phi}_{U} = \left[\alpha_{U}, \psi_{U}, \varphi_{U}\right]^{\mathrm{T}}$. The point $\bm{s}_{u}$ can be described with respect to the global origin as $\bm{p}_{u,k} = \bm{p}_{U,k} + \bm{s}_{u}.$ The receiver's centroid can be described with respect to the position of the $b^{\text{th}}$ LEO satellite as $\bm{p}_{U,k} = \bm{p}_{b,k} + d_{bU,k}\bm{\Delta}_{bU,k}$, where  $d_{bU,k}$ is the distance from point $\bm{p}_{b,k}$ to point $\bm{p}_{U,k}$ and $\bm{\Delta}_{bU,k}$ is the appropriate unit direction vector $\bm{\Delta}_{bU,k} = [\cos \phi_{bU,k} \sin \theta_{bU,k}, \sin \phi_{bU,k} \sin \theta_{bU,k}, \cos \theta_{bU,k}]^{\mathrm{T}}$. The receiver's $u^{\text{th}}$ antenna can be described with respect to the position of the $b^{\text{th}}$ LEO satellite as $\bm{p}_{u,k} = \bm{p}_{b,k} + d_{bu,k}\bm{\Delta}_{bu,k}$ where  $d_{bu,k}$ is the distance from point $\bm{p}_{b,k}$ to point $\bm{p}_{u,k}$ and $\bm{\Delta}_{bu,k}$ is the appropriate unit direction vector $\bm{\Delta}_{bu,k} = [\cos \phi_{bu,k} \sin \theta_{bu,k}, \sin \phi_{bu,k} \sin \theta_{bu,k}, \cos \theta_{bu,k}]^{\mathrm{T}}$. The receiver's centroid can be described with respect to the position of the $q^{\text{th}}$ BS as $\bm{p}_{U,k} = \bm{p}_{q,k} + d_{qU,k}\bm{\Delta}_{qU,k}$, where  $d_{qU,k}$ is the distance from point $\bm{p}_{q,k}$ to point $\bm{p}_{U,k}$ and $\bm{\Delta}_{qU,k}$ is the appropriate unit direction vector $\bm{\Delta}_{qU,k} = [\cos \phi_{qU,k} \sin \theta_{qU,k}, \sin \phi_{qU,k} \sin \theta_{qU,k}, \cos \theta_{qU,k}]^{\mathrm{T}}$. The receiver's $u^{\text{th}}$ antenna can be described with respect to the position of the $q^{\text{th}}$ BS as $\bm{p}_{u,k} = \bm{p}_{q,k} + d_{qu,k}\bm{\Delta}_{qu,k}$, where  $d_{qu,k}$ is the distance from point $\bm{p}_{q,k}$ to point $\bm{p}_{u,k}$ and $\bm{\Delta}_{qu,k}$ is the appropriate unit direction vector $\bm{\Delta}_{qu,k} = [\cos \phi_{qu,k} \sin \theta_{qu,k}, \sin \phi_{qu,k} \sin \theta_{qu,k}, \cos \theta_{qu,k}]^{\mathrm{T}}$. The $q^{\text{th}}$ BS can be described with respect to the position of the $b^{\text{th}}$ LEO as $\bm{p}_{q,k} = \bm{p}_{b,k} + d_{bq,k}\bm{\Delta}_{bq,k}$, where  $d_{bq,k}$ is the distance from point $\bm{p}_{b,k}$ to point $\bm{p}_{q,k}$ and $\bm{\Delta}_{bq,k}$ is the appropriate unit direction vector $\bm{\Delta}_{bq,k} = [\cos \phi_{bq,k} \sin \theta_{bq,k}, \sin \phi_{bq,k} \sin \theta_{bq,k}, \cos \theta_{bq,k}]^{\mathrm{T}}$. 

\subsection{Transmit and Receive Processing}
There are three links of interest to consider: i) LEO-receiver link, ii) LEO-BS link, and iii) BS-receiver link. At time $t$, during $k^{\text{th}}$ transmission time slot, the $N_B$ LEOs communicate with the $N_Q$ BSs and receiver over $N_K$ transmission time slots using quadrature modulation. The $b^{\text{th}}$ LEO transmits the following symbol
\begin{equation}
    \begin{aligned}
        x_{b,k}[t] = s_{b,k}[t] \operatorname{exp}_{}{(j2 \pi f_c t )},
    \end{aligned}
\end{equation}
where $s_{b,k}[t]$ is the modulation symbol and $f_c = \frac{c}{\lambda}$ is the operating frequency. Here, $c$ is the speed of light, and $\lambda$ is the wavelength. In this work, in the LEO-receiver link, only the line of sight paths are considered, and the useful part of the signal received at time $t$, during $k^{\text{th}}$ transmission time slot on the $u^{\text{th}}$ receive antenna is

\begin{equation}
\begin{aligned}
\label{equ:receive_signal_1}
 \mu_{bu,k}[t] &=  \beta_{bu,k}  \sqrt{2} \Re\left\{s_{b,k}[t_{obu,k}]  \operatorname{exp}(j( 2 \pi f_{obU,k} t_{obu,k}))\right\}. \\
\end{aligned}
\end{equation}
Here, $\beta_{bu,k}$ is the channel gain at the $u^{\text{th}}$ receive antenna during the $k^{\text{th}}$ time slot. The effective time duration from the $b^{\text{th}}$ LEO to the $u^{\text{th}}$ receive antenna is $t_{obu,k} = t - \tau_{bu,k} + \delta_{bU}$ and the effective frequency observed at the receiver from the $b^{\text{th}}$ LEO is $f_{obU,k} = f_c (1 - \nu_{bU,k}) + \epsilon_{bU}$. During the $k^{\text{th}}$ time slot, the delay from the $b^{\text{th}}$ LEO to the $u^{\text{th}}$ receive antenna is
$$
\tau_{bu,k} \triangleq  \frac{\left\|\mathbf{p}_{u,k}- (\mathbf{p}_{b,k} + \check{\mathbf{p}}_{b,k})\right\|}{c}.
$$
The time offset and frequency offset of the $b^{\text{th}}$ LEO satellite with respect to the receiver is $\delta_{bU}$ and $\epsilon_{bU}$, respectively. The point, $\check{\mathbf{p}}_{b,k}$ describes the uncertainty associated with the position of the $b^{\text{th}}$ LEO during the $k^{\text{th}}$ time slot. The Doppler observed at the receiver with respect to the $b^{\text{th}}$ LEO satellite is
$$
\nu_{bU,k} = \bm{\Delta}_{bU,k}^{\mathrm{T}} \frac{(\bm{v}_{b,k} + \check{\bm{v}}_{b,k} - \bm{v}_{U,k})}{c}.
$$
Here, $\bm{v}_{b,k} = v_{b} \bm{\Delta}_{b,k}$ and $\bm{v}_{U,k}  = v_{U} \bm{\Delta}_{U,k}$ are the velocities of the $b^{\text{th}}$ LEO and receiver, respectively.   The speeds of the $b^{\text{th}}$ LEO satellite and receiver are $v_{b}$ and $v_{U}$, respectively. The associated directions are defined as $\bm{\Delta}_{b,k} = [\cos \phi_{b,k} \sin \theta_{b,k}, \sin \phi_{b,k} \sin \theta_{b,k}, \cos \theta_{b,k}]^{\mathrm{T}}$ and $\bm{\Delta}_{U,k} = [\cos \phi_{U,k} \sin \theta_{U,k}, \sin \phi_{U,k} \sin \theta_{U,k}, \cos \theta_{U,k}]^{\mathrm{T}}$, respectively. Here, $\check{\bm{v}}_{b,k}$ is the uncertainty related to the velocity of the $b^{\text{th}}$ LEO during the $k^{\text{th}}$ time slot.

In the LEO-BS link, only the line of sight paths are considered, and the useful part of the signal received at time $t$, during $k^{\text{th}}$ transmission time slot on the $q^{\text{th}}$ BS is

\begin{equation}
\begin{aligned}
\label{equ:receive_signal_2}
 \mu_{bq,k}[t] &=  \beta_{bq,k}  \sqrt{2} \Re\left\{s_{b,k}[t_{obq,k}]  \operatorname{exp}(j( 2 \pi f_{obq,k} t_{obq,k}))\right\}. \\
\end{aligned}
\end{equation}
Here, $\beta_{bq,k}$ is the channel gain at the $q^{\text{th}}$ BS during the $k^{\text{th}}$ time slot. The effective time duration from the $b^{\text{th}}$ LEO to the $q^{\text{th}}$ BS is $t_{obq,k} = t - \tau_{bq,k} + \delta_{bQ}$ and the effective frequency observed at the $q^{\text{th}}$ BS from the $b^{\text{th}}$ LEO is $f_{obq,k} = f_c (1 - \nu_{bq,k}) + \epsilon_{bQ}$. During the $k^{\text{th}}$ time slot, the delay from the $b^{\text{th}}$ LEO to the $q^{\text{th}}$ BS is
$$
\tau_{bq,k} \triangleq  \frac{\left\|\mathbf{p}_{q,k}- (\mathbf{p}_{b,k} + \check{\mathbf{p}}_{b,k})\right\|}{c}.
$$
The time offset and frequency offset of the $b^{\text{th}}$ LEO satellite with respect to the $q^{\text{th}}$ BS is $\delta_{bQ}$ and $\epsilon_{bQ}$, respectively. The Doppler observed at the $q^{\text{th}}$ BS with respect to the $b^{\text{th}}$ LEO satellite is
$$
\nu_{bq,k} = \bm{\Delta}_{bq,k}^{\mathrm{T}} \frac{(\bm{v}_{b,k} + \check{\bm{v}}_{b,k} )}{c}.
$$

In the BS-receiver link, only the line of sight paths are considered, and the useful part of the signal received at time $t$, during $k^{\text{th}}$ transmission time slot at the receiver is
\begin{equation}
\begin{aligned}
\label{equ:receive_signal_3}
 \mu_{qu,k}[t] &=  \beta_{qu,k}  \sqrt{2} \Re\left\{s_{q,k}[t_{oqu,k}]  \operatorname{exp}(j( 2 \pi f_{oqU,k} t_{oqu,k}))\right\}. \\
\end{aligned}
\end{equation}
Here, $\beta_{qu,k}$ is the channel gain at the $u^{\text{th}}$ antenna on the receiver from the $q^{\text{th}}$ BS during the $k^{\text{th}}$ time slot. The effective time duration from the $q^{\text{th}}$ BS to the $u^{\text{th}}$ antenna on the receiver is $t_{oqu,k} = t - \tau_{qu,k} + \delta_{QU}$ and the effective frequency observed at the $u^{\text{th}}$ receive antenna from the $q^{\text{th}}$ BS is $f_{oqU,k} = f_c (1 - \nu_{qU,k}) + \epsilon_{QU}$. During the $k^{\text{th}}$ time slot, the delay from the $q^{\text{th}}$ BS to the $u^{\text{th}}$ receive antenna is
$$
\tau_{qu,k} \triangleq  \frac{\left\|\mathbf{p}_{u,k}- \mathbf{p}_{q,k} \right\|}{c}.
$$
The time offset and frequency offset of BSs with respect to the receiver are $\delta_{QU}$ and $\epsilon_{QU}$, respectively. The Doppler observed at the receiver with respect to the $q^{\text{th}}$ BS is
$$
\nu_{qU,k} = \bm{\Delta}_{qU,k}^{\mathrm{T}} \frac{(0 - \bm{v}_{U,k} )}{c}.
$$ 
With this formulation, the received signal at the $u^{\text{th}}$ receive antenna during the $k^{th}$ time slot from the LEOs and BSs is
\begin{equation}
\begin{aligned}
\label{equ:receive_signalA}
y_{1u,k}[t] &= \sum_{q}^{N_Q} y_{qu,k}[t] = \sum_{q}^{N_Q} \mu_{qu,k}[t]+ {n}_{u,k}[t],\\
y_{2u,k}[t] &= \sum_{b}^{N_B} y_{bu,k}[t] = \sum_{b}^{N_B} \mu_{bu,k}[t] + {n}_{u,k}[t],
\end{aligned}
\end{equation}
where  ${n}_{u,k}[t] \sim \mathcal{C}\mathcal{N}(0,N_{01})$ is the Fourier transformed thermal noise local to the receiver's antenna array. The received signal at the $q^{\text{th}}$ BS during the $k^{th}$ time slot from the LEOs is
\begin{equation}
\begin{aligned}
\label{equ:receive_signalB}
y_{q,k}[t] &= \sum_{b}^{N_B} y_{bq,k}[t] = \sum_{b}^{N_B} \mu_{bq,k}[t] + {n}_{q,k}[t],
\end{aligned}
\end{equation}
where  ${n}_{q,k}[t] \sim \mathcal{C}\mathcal{N}(0,N_{02})$ is the Fourier transformed thermal noise local to the $q^{\text{th}}$ BS.

\begin{remark}
    The offset $\delta_{bU}$ captures the unknown ionospheric and tropospheric delay concerning the $b^{\text{th}}$ LEO satellite as well as the time offset in the LEO-receiver link. Similarly, the offset $\delta_{bQ}$ captures the unknown ionospheric and tropospheric delay concerning the $b^{\text{th}}$ LEO satellite as well as the time offset in the LEO-BS link.
\end{remark}

\begin{remark}
    The BSs are all synchronized in time and frequency. Hence, a single offset term describes the time offset between the receiver and all the BSs. Also, a single offset term describes the frequency offset between the receiver and all the BSs.
\end{remark}

The position of the $b^{\text{th}}$ LEO satellite and the $u^{\text{th}}$ receive antenna at the $k^{\text{th}}$ time slot is
$$
\begin{aligned}
    \mathbf{p}_{b,k} &= \mathbf{p}_{b,o} + \Tilde{\mathbf{p}}_{b,k}, \\
    \mathbf{p}_{u,k} &= \mathbf{p}_{u,o} + \Tilde{\mathbf{p}}_{U,k},
\end{aligned}
$$
where $\mathbf{p}_{b,o}$ and $\mathbf{p}_{u,o}$ serve as the reference points for the $b^{\text{th}}$ LEO satellite and the $u^{\text{th}}$ receive antenna, respectively. The distances covered by the $b^{\text{th}}$ LEO satellite and the $u^{\text{th}}$ receive antenna are $\Tilde{\mathbf{p}}{b,k}$ and $\Tilde{\mathbf{p}}{u,k}$, respectively. These traveled distances are defined as
$$
\begin{aligned}
\Tilde{\mathbf{p}}_{b,k} &= (k) \Delta_{t} v_{b} \mathbf{\Delta}_{b,k}, \\
    \Tilde{\mathbf{p}}_{U,k} &= (k) \Delta_{t} v_{U} \bm{\Delta}_{U,k}. \\
    \end{aligned}
$$

\subsection{Received Signal Properties}
The properties of the signal received across all $N_K$ antennas from the $N_B$ LEOs are described with the aid of the: i) Fourier transform of the baseband signal (spectral density) that is transmitted by the $b^{\text{th}}$ LEO satellite at time $t$ during the $k^{\text{th}}$ time slot,
$$
S_{b,k}[f] \triangleq \frac{1}{\sqrt{2 \pi}} \int_{-\infty}^{\infty} s_{b,k}[t] \operatorname{exp}_{}{(-j2 \pi f t )} \; \;    {\rm d} t,
$$
and
ii) Fourier transform of the baseband signal (spectral density) that is transmitted by the $q^{\text{th}}$ BS at time $t$ during the $k^{\text{th}}$ time slot,
$$
S_{q,k}[f] \triangleq \frac{1}{\sqrt{2 \pi}} \int_{-\infty}^{\infty} s_{q,k}[t] \operatorname{exp}_{}{(-j2 \pi f t )} \; \;    {\rm d} t.
$$
Some useful properties of the received signals are summarized below.
\subsubsection{Effective Baseband Bandwidth}
This relates to the variance of all the occupied frequencies. From the system definition, we have two effect baseband bandwidths
$$
\alpha_{1b,k} \triangleq\left(\frac{\int_{-\infty}^{\infty} f^2\left|S_{b,k}[f]\right|^2 d f}{\int_{-\infty}^{\infty}\left|S_{b,k}[f]\right|^2 d f}\right)^{\frac{1}{2}},
$$
and
$$
\alpha_{1q,k} \triangleq\left(\frac{\int_{-\infty}^{\infty} f^2\left|S_{q,k}[f]\right|^2 d f}{\int_{-\infty}^{\infty}\left|S_{q,k}[f]\right|^2 d f}\right)^{\frac{1}{2}}.
$$

\subsubsection{Baseband-Carrier Correlation (BCC)}
This property helps to provide a compact representation of the mathematical description of the available information in the received signals $$
\alpha_{2b,k} \triangleq\frac{\int_{-\infty}^{\infty} f\left|S_{b,k}[f]\right|^2 d f}{\left(\int_{-\infty}^{\infty} f^2\left|S_{b,k}[f]\right|^2 d f \right)^{\frac{1}{2}} \left(\int_{-\infty}^{\infty}\left|S_{b,k}[f]\right|^2 d f\right)^{\frac{1}{2}}},
$$
and
$$\alpha_{2q,k} \triangleq\frac{\int_{-\infty}^{\infty} f\left|S_{q,k}[f]\right|^2 d f}{\left(\int_{-\infty}^{\infty} f^2\left|S_{q,k}[f]\right|^2 d f \right)^{\frac{1}{2}} \left(\int_{-\infty}^{\infty}\left|S_{q,k}[f]\right|^2 d f\right)^{\frac{1}{2}}}.
$$

\subsubsection{Root Mean Squared Time Duration}
The root mean squared time duration from the $b^{\text{th}}$ LEO satellite to the $u^{\text{th}}$ receive antenna during the $k^{\text{th}}$ time slot is 
$$
\alpha_{obu,k} \triangleq\left(\frac{\int_{-\infty}^{\infty}  2 t_{obu,k}^{2} \left|s(t_{obu,k})\right|^2  \; dt_{obu,k}}{\int_{-\infty}^{\infty}  \left|s(t_{obu,k})\right|^2  \; dt_{obu,k}}\right)^{\frac{1}{2}}.
$$
The root mean squared time duration from the $b^{\text{th}}$ LEO satellite to the $q^{\text{th}}$ BS during the $k^{\text{th}}$ time slot is 
$$
\alpha_{obq,k} \triangleq\left(\frac{\int_{-\infty}^{\infty}  2 t_{obq,k}^{2} \left|s(t_{obq,k})\right|^2  \; dt_{obq,k}}{\int_{-\infty}^{\infty}  \left|s(t_{obq,k})\right|^2  \; dt_{obq,k}}\right)^{\frac{1}{2}}.
$$
The root mean squared time duration from the $q^{\text{th}}$ BS to the $u^{\text{th}}$ receive antenna during the $k^{\text{th}}$ time slot is 
$$
\alpha_{oqu,k} \triangleq\left(\frac{\int_{-\infty}^{\infty}  2 t_{oqu,k}^{2} \left|s(t_{oqu,k})\right|^2  \; dt_{oqu,k}}{\int_{-\infty}^{\infty}  \left|s(t_{oqu,k})\right|^2  \; dt_{oqu,k}}\right)^{\frac{1}{2}}.
$$

\subsubsection{Received Signal-to-Noise Ratio}
The SNR measures the power ratio of the signal across its frequencies to the noise spectral density. In mathematical terms, based on the system model, the SNRs are
$$
\underset{bu,k}{\operatorname{SNR}} \triangleq \frac{8 \pi^2 \left|\beta_{bu,k}\right|^2}{N_{01}} \int_{-\infty}^{\infty}\left|S_{b,k}[f]\right|^2 d f,
$$
$$
\underset{bq,k}{\operatorname{SNR}} \triangleq \frac{8 \pi^2 \left|\beta_{bq,k}\right|^2}{N_{02}} \int_{-\infty}^{\infty}\left|S_{b,k}[f]\right|^2 d f,
$$
and
$$
\underset{qu,k}{\operatorname{SNR}} \triangleq \frac{8 \pi^2 \left|\beta_{qu,k}\right|^2}{N_{01}} \int_{-\infty}^{\infty}\left|S_{q,k}[f]\right|^2 d f.
$$
If there is no beam split, the channel gain is constant across all receive antennas and we have
$$
\underset{b,k}{\operatorname{SNR}} \triangleq \frac{8 \pi^2 \left|\beta_{b,k}\right|^2}{N_{01}} \int_{-\infty}^{\infty}\left|S_{b,k}[f]\right|^2 d f,
$$
$$
\underset{bq,k}{\operatorname{SNR}} \triangleq \frac{8 \pi^2 \left|\beta_{bq,k}\right|^2}{N_{02}} \int_{-\infty}^{\infty}\left|S_{b,k}[f]\right|^2 d f,
$$
and
$$
\underset{q,k}{\operatorname{SNR}} \triangleq \frac{8 \pi^2 \left|\beta_{q,k}\right|^2}{N_{01}} \int_{-\infty}^{\infty}\left|S_{q,k}[f]\right|^2 d f.
$$

If the same signal is transmitted across all $N_{K}$ time slots, and the channel gain is constant across all receive antennas and time slots, we have
$$
\underset{b}{\operatorname{SNR}} \triangleq \frac{8 \pi^2 \left|\beta_{b}\right|^2}{N_{01}} \int_{-\infty}^{\infty}\left|S_{b}[f]\right|^2 d f,
$$
$$
\underset{bq}{\operatorname{SNR}} \triangleq \frac{8 \pi^2 \left|\beta_{bq}\right|^2}{N_{02}} \int_{-\infty}^{\infty}\left|S_{b}[f]\right|^2 d f,
$$
and
$$
\underset{q}{\operatorname{SNR}} \triangleq \frac{8 \pi^2 \left|\beta_{q}\right|^2}{N_{01}} \int_{-\infty}^{\infty}\left|S_{q,k}[f]\right|^2 d f.
$$
The subsequent sections rely heavily on these signal properties.

\section{Available Information about Channel Parameters in the Received Signal}
 The information about the channel parameters in the received signal is presented in this section and serves as an intermediate step to investigate the information needed for localization. 

 \subsection{Geometric and nuisance channel parameters}
To derive the available information about the channel parameters in i) the signal received across the $N_U$ antennas from both the $N_B$ LEOs and the $N_Q$ BSs, and ii) the signal received at the $N_Q$ BSs from the $N_B$ LEOs during the $N_K$ transmission time slots, we highlight both the geometric and nuisance channel parameters. We start with the delays in the LEO-receiver link. We can vectorize the delays received across all the antennas during the $k^{\text{th}}$ time slot as
$$
\bm{\tau}_{bU,k}
\triangleq\left[{\tau}_{b1,k}, {\tau}_{b2,k}, \cdots,
{\tau}_{bN_U,k}\right]^{\mathrm{T}}, 
$$
the next vectorization occurs considering the time slots and the $b^{\text{th}}$ LEO
$$
\bm{\tau}_{bU}
\triangleq\left[\bm{\tau}_{bU,1}^{\mathrm{T}}, \bm{\tau}_{bU,2}^{\mathrm{T}}, \cdots, \bm{\tau}_{bU,N_K}^{\mathrm{T}}\right]^{\mathrm{T}}.
$$
Focusing on the $b^{\text{th}}$ LEO satellite, the Doppler across all the $N_{K}$ transmission time slots is
$$
\bm{\nu}_{bU}
\triangleq\left[{\nu}_{bU,1}, {\nu}_{bU,2}, \cdots, \nu_{bU,N_K}\right]^{\mathrm{T}}.
$$
The channel gain in the LEO-receiver link can be placed in vector form as
$$
\bm{\beta}_{bU,k}
\triangleq\left[{\beta}_{b1,k}, {\beta}_{b2,k}, \cdots, {\beta}_{bN_U,k}\right]^{\mathrm{T}}, 
$$
and 
$$
\bm{\beta}_{bU}
\triangleq\left[\bm{\beta}_{bU,1}^{\mathrm{T}}, \bm{\beta}_{bU,2}^{\mathrm{T}}, \cdots, \bm{\beta}_{bU,N_K}^{\mathrm{T}}\right]^{\mathrm{T}}.
$$
It is important to note that if the channel gain remains constant across the $N_K$ time slots and $N_U$ receive antennas, we can represent the channel gain by a scalar $\beta_{bU}$. We can also represent the observable parameters in signals from the $N_B$ LEOs across the $N_U$ antennas during the $N_K$ time slots in vector form as:
$$
\bm{\eta}_{bU} \triangleq\left[\bm{\tau}_{bU}^{\mathrm{T}}, \bm{\nu}_{bU}^{\mathrm{T}}, \bm{\beta}_{bU}^{\mathrm{T}}, \delta_{bU}, \epsilon_{bU}\right]^{\mathrm{T}}.
$$
Next, we focus on parameters in the BSs-receiver link. We can vectorize the delays received across all the antennas during the $k^{\text{th}}$ time slot as
$$
\bm{\tau}_{qU,k}
\triangleq\left[{\tau}_{q1,k}, {\tau}_{q2,k}, \cdots,
{\tau}_{qN_U,k}\right]^{\mathrm{T}}, 
$$
the next vectorization occurs considering the time slots and the $q^{\text{th}}$ BS
$$
\bm{\tau}_{qU}
\triangleq\left[\bm{\tau}_{qU,1}^{\mathrm{T}}, \bm{\tau}_{qU,2}^{\mathrm{T}}, \cdots, \bm{\tau}_{qU,N_K}^{\mathrm{T}}\right]^{\mathrm{T}}.
$$
Focusing on the $q^{\text{th}}$ BS, the Doppler across all the $N_{K}$ transmission time slots is
$$
\bm{\nu}_{qU}
\triangleq\left[{\nu}_{qU,1}, {\nu}_{qU,2}, \cdots, \nu_{qU,N_K}\right]^{\mathrm{T}}.
$$
The channel gain in the BSs-receiver link can be placed in vector form as
$$
\bm{\beta}_{qU,k}
\triangleq\left[{\beta}_{q1,k}, {\beta}_{q2,k}, \cdots, {\beta}_{qN_U,k}\right]^{\mathrm{T}}, 
$$
and 
$$
\bm{\beta}_{qU}
\triangleq\left[\bm{\beta}_{qU,1}^{\mathrm{T}}, \bm{\beta}_{qU,2}^{\mathrm{T}}, \cdots, \bm{\beta}_{qU,N_K}^{\mathrm{T}}\right]^{\mathrm{T}}.
$$
It is essential to highlight that if the channel gain remains unchanged across the $N_K$ time slots and $N_U$ receive antennas, it can be expressed as a scalar $\beta_{qU}$. Consequently, the observable parameters in signals from the $N_Q$ BSs across the $N_U$ antennas during the $N_K$ time slots  can be represented in vector form as:
$$
\bm{\eta}_{qU} \triangleq\left[\bm{\tau}_{qU}^{\mathrm{T}}, \bm{\nu}_{qU}^{\mathrm{T}}, \bm{\beta}_{qU}^{\mathrm{T}}, \delta_{QU}, \epsilon_{QU}\right]^{\mathrm{T}}.
$$
Lastly, we focus on parameters in the LEOs-BSs links. We can vectorize the delays received across all the BSs during the $k^{\text{th}}$ time slot as
$$
\bm{\tau}_{bQ,k}
\triangleq\left[{\tau}_{b1,k}, {\tau}_{b2,k}, \cdots,
{\tau}_{bN_Q,k}\right]^{\mathrm{T}}, 
$$
the next vectorization occurs considering the time slots and the $q^{\text{th}}$ BS
$$
\bm{\tau}_{bQ}
\triangleq\left[\bm{\tau}_{bQ,1}^{\mathrm{T}}, \bm{\tau}_{bQ,2}^{\mathrm{T}}, \cdots, \bm{\tau}_{bQ,N_K}^{\mathrm{T}}\right]^{\mathrm{T}}.
$$
Focusing on the $q^{\text{th}}$ BS, the Doppler observed with respect to the $b^{\text{th}}$ LEO across all the $N_{K}$ transmission time slots is
$$
\bm{\nu}_{bq}
\triangleq\left[{\nu}_{bq,1}, {\nu}_{bq,2}, \cdots, \nu_{bq,N_K}\right]^{\mathrm{T}},
$$
vectorizing all the Dopplers observed with respect to the $b^{\text{th}}$ LEO produces
$$
\bm{\nu}_{bQ}
\triangleq\left[\bm{\nu}_{b1}^{\mathrm{T}}, \bm{\nu}_{b2}^{\mathrm{T}}, \cdots, \bm{\nu}_{bN_Q}^{\mathrm{T}}\right]^{\mathrm{T}}.
$$
The channel gain in the LEOs-BSs links can be placed in vector form as
$$
\bm{\beta}_{bq}
\triangleq\left[{\beta}_{bq,1}, {\beta}_{bQ,2}, \cdots, {\beta}_{bQ,N_{K}}\right]^{\mathrm{T}}, 
$$
and 
$$
\bm{\beta}_{bQ}
\triangleq\left[\bm{\beta}_{b1}^{\mathrm{T}}, \bm{\beta}_{b2}^{\mathrm{T}}, \cdots, \bm{\beta}_{bN_Q}^{\mathrm{T}}\right]^{\mathrm{T}}.
$$
We can represent the observable parameters in signals from the $N_B$ LEOs at all the $N_Q$ BSs during the $N_K$ time slots in vector form as:
$$
\bm{\eta}_{bQ} \triangleq\left[\bm{\tau}_{bQ}^{\mathrm{T}}, \bm{\nu}_{bQ}^{\mathrm{T}}, \bm{\beta}_{bQ}^{\mathrm{T}}, \delta_{bQ}, \epsilon_{bQ}\right]^{\mathrm{T}}.
$$

\begin{remark}
    We have two cases to consider for parameterization: i) in the first case, both the signals received at the receiver and BSs are available, and we have
    $$
    \bm{\eta} = \left[\bm{\eta}_{1U}^{\mathrm{T}},\cdots,\bm{\eta}_{N_BU}^{\mathrm{T}}, \bm{\eta}_{1U}^{\mathrm{T}},\cdots,\bm{\eta}_{N_QU}^{\mathrm{T}}, \bm{\eta}_{1Q}^{\mathrm{T}},\cdots,\bm{\eta}_{N_BQ}^{\mathrm{T}}\right]^{\mathrm{T}},
    $$
    ii) in the second case, only the signals at the receiver are available, and we have
        $$
    \bm{\eta} = \left[\bm{\eta}_{1U}^{\mathrm{T}},\cdots,\bm{\eta}_{N_BU}^{\mathrm{T}}, \bm{\eta}_{1U}^{\mathrm{T}},\cdots,\bm{\eta}_{N_QU}^{\mathrm{T}}\right]^{\mathrm{T}}.
    $$
\end{remark}
We have specified all the parameters that are observable in the received signals. In the next section, we present mathematical preliminaries that help determine the information available about these parameters in the received signals.

\subsection{Mathematical Preliminaries}
In estimation theory, two questions of paramount importance are the parameters that can be estimated and the conditions that allow for the estimation of these parameters.
One way of answering these questions is through the FIM. To present the FIM, we assume that for the parameters and system model in our work, there exists an unbiased estimate $\hat{\bm{\eta}}$ such that the error covariance matrix satisfies the following information inequality
$
\mathbb{E}_{\bm{y}; \boldsymbol{\eta}}\left\{(\hat{\boldsymbol{\eta}}-\boldsymbol{\eta})(\hat{\boldsymbol{\eta}}-\boldsymbol{\eta})^{\mathrm{T}}\right\} \succeq \mathbf{J}_{ \bm{\bm{y}}; \bm{\eta}}^{-1},
$
where $\mathbf{J}_{ \bm{\bm{y}}; \bm{\eta}}$ is the FIM for the parameter vector $\boldsymbol{\eta}.$

\begin{definition}
\label{definition_FIM_1}
The FIM obtained from the likelihood due to the observations is defined as  $\mathbf{J}_{\bm{y};\bm{\eta}} =  \bm{F}_{{\bm{y} }}(\bm{y}| \bm{\eta} ;\bm{\eta},\bm{\eta})$. In mathematical terms, we have
\begin{equation}
\label{definition_equ:definition_FIM_1}
\begin{aligned}
\mathbf{J}_{ \bm{\bm{y}}; \bm{\eta}} &\triangleq 
-\mathbb{E}_{\bm{y};\bm{\eta}_{}}\left[\frac{\partial^{2} \ln \chi(\bm{y}_{};  \bm{\eta}_{} )}{\partial \bm{\eta}_{} \partial \bm{\eta}_{}^{\mathrm{T}}}\right]
\end{aligned}
\end{equation}
where  $\chi(\bm{y}_{};  \bm{\eta}_{} )$ denotes the likelihood function considering $\bm{y}$ and $\bm{\eta}$.
\end{definition}
The FIM is a very useful tool, however it grows quadratically with the size of the parameter vector. Hence, it might be advantageous to focus on a subset of the FIM. One way to do this is to use the equivalent FIM (EFIM) \cite{horn2012matrix}.

\begin{definition}
\label{definition_EFIM}
Given a parameter vector, $ \bm{\eta}_{} \triangleq\left[\bm{\eta}_{1}^{\mathrm{T}}, \bm{\eta}_{2}^{\mathrm{T}}\right]^{\mathrm{T}}$, where $\bm{\eta}_{1}$ is the parameter of interest, the resultant FIM has the structure 
$$
\mathbf{J}_{ \bm{\bm{y}}; \bm{\eta}}=\left[\begin{array}{cc}
\mathbf{J}_{ \bm{\bm{y}}; \bm{\eta}_1}^{}  & \mathbf{J}_{ \bm{\bm{y}}; \bm{\eta}_1, \bm{\eta}_2}^{} \\
 \mathbf{J}_{ \bm{\bm{y}}; \bm{\eta}_1, \bm{\eta}_2}^{\mathrm{T}} &\mathbf{J}_{ \bm{\bm{y}}; \bm{\eta}_2}^{}
\end{array}\right],
$$
where $\bm{\eta} \in \mathbb{R}^{N}, \bm{\eta}_{1} \in \mathbb{R}^{n}, \mathbf{J}_{ \bm{\bm{y}}; \bm{\eta}_1}^{} \in \mathbb{R}^{n \times n},  \mathbf{J}_{ \bm{\bm{y}}; \bm{\eta}_1, \bm{\eta}_2}\in \mathbb{R}^{n \times(N-n)}$, and $\mathbf{J}_{ \bm{\bm{y}}; \bm{\eta}_2}^{}\in$ $\mathbb{R}^{(N-n) \times(N-n)}$ with $n<N$, 
and the EFIM \cite{5571900} of  parameter ${\bm{\eta}_{1}}$ is given by 
$\mathbf{J}_{ \bm{\bm{y}}; \bm{\eta}_1}^{\mathrm{e}} =\mathbf{J}_{ \bm{\bm{y}}; \bm{\eta}_1}^{} - \mathbf{J}_{ \bm{\bm{y}}; \bm{\eta}_1}^{nu} =\mathbf{J}_{ \bm{\bm{y}}; \bm{\eta}_1}^{}-
\mathbf{J}_{ \bm{\bm{y}}; \bm{\eta}_1, \bm{\eta}_2}^{} \mathbf{J}_{ \bm{\bm{y}}; \bm{\eta}_2}^{-1} \mathbf{J}_{ \bm{\bm{y}}; \bm{\eta}_1, \bm{\eta}_2}^{\mathrm{T}}.$

Note that the term $\mathbf{J}_{ \bm{\bm{y}}; \bm{\eta}_1}^{nu}  = \mathbf{J}_{ \bm{\bm{y}}; \bm{\eta}_1, \bm{\eta}_2}^{} \mathbf{J}_{ \bm{\bm{y}}; \bm{\eta}_2}^{-1} \mathbf{J}_{ \bm{\bm{y}}; \bm{\eta}_1, \bm{\eta}_2}^{\mathrm{T}}$ describes the loss of information about ${\bm{\eta}_{1}}$  due to uncertainty in the nuisance parameters ${\bm{\eta}_{2}}$. This EFIM captures all the required information about the parameters of interest present in the FIM; as observed from the relation $(\mathbf{J}_{ \bm{\bm{y}}; \bm{\eta}_1}^{\mathrm{e}})^{-1} = [\mathbf{J}_{ \bm{\bm{y}}; \bm{\eta}}^{-1}]_{[1:n,1:n]}$.
\end{definition}

\subsection{FIM for channel parameters}
To derive the FIM for the channel parameters, we present the likelihood for two cases of parameterization: i)  both the signals received at the receiver and BSs are available, which results in (\ref{equ:likelihood}) as the likelihood function, and ii) only the signals received at the receiver are available, which results in (\ref{equ:likelihood_1}) as the likelihood function.
\begin{figure*}
\begin{align}
\begin{split}
\label{equ:likelihood}
    \chi(\bm{y}_{}[t]|  \bm{\eta}_{})  &\propto  \prod_{b = 1}^{N_B}\prod_{u = 1}^{N_U}\prod_{k = 1}^{N_K} \prod_{q = 1}^{N_Q} \prod_{u^{'} = 1}^{N_U}\prod_{k^{'} = 1}^{N_K}  \prod_{b^{'} = 1}^{N_B} \prod_{q^{'} = 1}^{N_Q} \prod_{k^{''} = 1}^{N_K}    \operatorname{exp} \left\{\frac{2}{N_{01}} \int_0^{T_{}} \Re\left\{ \mu_{bu,k}[t]
^{\mathrm{H}}{y}_{bu,k}[t]\right\} d t - \frac{1}{N_{01}} \int_0^{T_{}}|\mu_{bu,k}[t]|^{2} \; d t \right\}
\\&    \operatorname{exp} \left\{\frac{2}{N_{01}} \int_0^{T_{}} \Re\left\{ \mu_{qu^{'},k^{'}}[t]
^{\mathrm{H}}{y}_{qu^{'},k^{'}}[t]\right\} d t - \frac{1}{N_{01}} \int_0^{T_{}}|\mu_{qu^{'},k^{'}}[t]|^{2} \; d t \right\}
\\&    \operatorname{exp} \left\{\frac{2}{N_{02}} \int_0^{T_{}} \Re\left\{ \mu_{b^{'}q^{'},k^{''}}[t]
^{\mathrm{H}}{y}_{b^{'}q^{'},k^{''}}[t]\right\} d t - \frac{1}{N_{02}} \int_0^{T_{}}|\mu_{b^{'}q^{'},k^{''}}[t]|^{2} \; d t \right\}.
\end{split}
\end{align}
\end{figure*}
\begin{figure*}
\begin{align}
\begin{split}
\label{equ:likelihood_1}
    \chi(\bm{y}_{}[t]|  \bm{\eta}_{})  &\propto  \prod_{b = 1}^{N_B}\prod_{u = 1}^{N_U}\prod_{k = 1}^{N_K} \prod_{q = 1}^{N_Q} \prod_{u^{'} = 1}^{N_U}\prod_{k^{'} = 1}^{N_K}     \operatorname{exp} \left\{\frac{2}{N_{01}} \int_0^{T_{}} \Re\left\{ \mu_{bu,k}[t]
^{\mathrm{H}}{y}_{bu,k}[t]\right\} d t - \frac{1}{N_{01}} \int_0^{T_{}}|\mu_{bu,k}[t]|^{2} \; d t \right\}
\\&    \operatorname{exp} \left\{\frac{2}{N_{01}} \int_0^{T_{}} \Re\left\{ \mu_{qu^{'},k^{'}}[t]
^{\mathrm{H}}{y}_{qu^{'},k^{'}}[t]\right\} d t - \frac{1}{N_{01}} \int_0^{T_{}}|\mu_{qu^{'},k^{'}}[t]|^{2} \; d t \right\}
\end{split}
\end{align}
\end{figure*}
Considering all the $N_B$ LEOs, $N_K$ transmission time slots, and $N_U$ receive antennas, we can derive the FIMs for both cases of parameterization. 
The FIMs for both parameterization cases result in a block diagonal. The first case of parameterization produces
\begin{equation}
\begin{aligned}
&\mathbf{J}_{\bm{y}|\bm{\eta}} =  \bm{F}_{{\bm{y} }}(\bm{y}| \bm{\eta} ;\bm{\eta},\bm{\eta}) = \\ & \operatorname{diag}\left\{\bm{F}_{{\bm{y} }}(\bm{y}| \bm{\eta} ;\bm{\eta}_{1U},\bm{\eta}_{1U}), \ldots, \bm{F}_{{\bm{y} }}(\bm{y}| \bm{\eta} ;\bm{\eta}_{N_{B}U},\bm{\eta}_{N_{B}U}) \right. \\ & \left. 
\bm{F}_{{\bm{y} }}(\bm{y}| \bm{\eta} ;\bm{\eta}_{1U},\bm{\eta}_{1U}), \ldots, \bm{F}_{{\bm{y} }}(\bm{y}| \bm{\eta} ;\bm{\eta}_{N_{Q}U},\bm{\eta}_{N_{Q}U})
\right. \\ & \left. 
\bm{F}_{{\bm{y} }}(\bm{y}| \bm{\eta} ;\bm{\eta}_{1Q},\bm{\eta}_{1Q}), \ldots, \bm{F}_{{\bm{y} }}(\bm{y}| \bm{\eta} ;\bm{\eta}_{N_{B}Q},\bm{\eta}_{N_{B}Q})
\right\},
    \end{aligned}
\end{equation}
and the second case of parameterization produces
\begin{equation}
\begin{aligned}
&\mathbf{J}_{\bm{y}|\bm{\eta}} =  \bm{F}_{{\bm{y} }}(\bm{y}| \bm{\eta} ;\bm{\eta},\bm{\eta}) = \\ & \operatorname{diag}\left\{\bm{F}_{{\bm{y} }}(\bm{y}| \bm{\eta} ;\bm{\eta}_{1U},\bm{\eta}_{1U}), \ldots, \bm{F}_{{\bm{y} }}(\bm{y}| \bm{\eta} ;\bm{\eta}_{N_{B}U},\bm{\eta}_{N_{B}U}) \right. \\ & \left. 
\bm{F}_{{\bm{y} }}(\bm{y}| \bm{\eta} ;\bm{\eta}_{1U},\bm{\eta}_{1U}), \ldots, \bm{F}_{{\bm{y} }}(\bm{y}| \bm{\eta} ;\bm{\eta}_{N_{Q}U},\bm{\eta}_{N_{Q}U})
\right\}.
    \end{aligned}
\end{equation}
Considering the LEOs-receiver link, the entries in FIM due to the observations of the signals at the receiver from $b^{\text{th}}$ LEO satellite can be obtained through the simplified expression
$$
\begin{aligned}
     \bm{F}_{{\bm{y} }}(\bm{y}| \bm{\eta} ;\bm{\eta}_{bU},\bm{\eta}_{bU}) &=   \frac{1}{N_{01}} \times \\&  \sum_{u,k}^{N_U N_K}\Re\left\{ \int \nabla_{\bm{\eta}_{bU}}{\mu}_{bu,k} [t]\nabla_{\bm{\eta}_{bU}}{\mu}_{bu,k}^{\mathrm{H}}[t] \; \; d t \right\}.
\end{aligned}
$$
We now present the non-zero entries focusing on the $b^{\text{th}}$ LEO satellite. We start with the delays focusing on the FIM for the delay from the $b^{\text{th}}$ LEO satellite during the $k^{\text{th}}$ time slot on the $u^{\text{th}}$ receive antenna
$$
\begin{aligned}
\bm{F}_{{\bm{y} }}(\bm{y}| \bm{\eta} ;{\tau}_{bu,k},{\tau}_{bu,k}) &= -\bm{F}_{{\bm{y} }}(\bm{y}| \bm{\eta} ;{\tau}_{bu,k},{\delta}_{bU}) = \underset{bu,k}{\operatorname{SNR}} \; \omega_{bU,k},
\end{aligned}
$$
where $\omega_{bU,k} = \Bigg[ \alpha_{1b,k}^2 +  2f_{obU,k} \alpha_{1b,k} \alpha_{2b,k} +  f_{obU,k}^2   \Bigg].$
All other entries in the FIM focusing on delays in the $b^{\text{th}}$ LEO link are zero. The FIM focusing on the Dopplers related to the $b^{\text{th}}$ LEO satellite is
$$
\begin{aligned}
\bm{F}_{{\bm{y} }}(\bm{y}| \bm{\eta} ;{\nu}_{bU,k}, {\nu}_{bU,k}) = 0.5 \times \underset{bu,k}{\operatorname{SNR}} f_{c}^2 \alpha_{obu,k}^2.
\end{aligned}
$$
The FIM of the Doppler observed with respect to the $b^{\text{th}}$ LEO satellite and the corresponding frequency offset during the $k^{\text{th}}$ time slot is
$$
\begin{aligned}
\bm{F}_{{\bm{y} }}(\bm{y}| \bm{\eta} ;{\nu}_{bU,k}, {\epsilon}_{bU}) = - 0.5 \times \underset{bu,k}{\operatorname{SNR}} f_{c} \alpha_{obu,k}^2.
\end{aligned}
$$
All other entries in the FIM related to the Dopplers in the $b^{\text{th}}$ LEO link are zero. Now, we focus on the channel gain in the $b^{\text{th}}$ LEO link. The FIM of the channel gain considering the received signals from $b^{\text{th}}$ LEO satellite to the $u^{\text{th}}$ receive antenna during the $k^{\text{th}}$ time slot is
$$
\begin{aligned}
\bm{F}_{{\bm{y} }}(\bm{y}| \bm{\eta} ; {\beta}_{bu,k}, {\beta}_{bu,k}) = \frac{1}{4 \pi^2 \left|\beta_{bu,k}\right|^2}\underset{bu,k}{\operatorname{SNR}}.
\end{aligned}
$$
All other entries in the FIM related to the channel gain in the $b^{\text{th}}$ LEO link are zero. Now, we focus on the time offset in the $b^{\text{th}}$ LEO link. The FIM between the time offset and the delay in the FIM due to the observations of the received signals from $b^{\text{th}}$ LEO satellite to the $u^{\text{th}}$ receive antenna during the $k^{\text{th}}$ time slot is
$$
\begin{aligned}
\bm{F}_{{\bm{y} }}(\bm{y}| \bm{\eta} ;{\delta}_{bU}, {\tau}_{bu,k}) = \bm{F}_{{\bm{y} }}(\bm{y}| \bm{\eta} ; {\tau}_{bu,k}, {\delta}_{bU}).
\end{aligned}
$$
The FIM of the time offset in the FIM due to the observations of the received signals from $b^{\text{th}}$ LEO satellite to the $u^{\text{th}}$ receive antenna during the $k^{\text{th}}$ time slot is
$$
\begin{aligned}
\bm{F}_{{\bm{y} }}(\bm{y}| \bm{\eta} ;{\delta}_{bU}, {\delta}_{bU}) = \bm{F}_{{\bm{y} }}(\bm{y}| \bm{\eta} ; {\tau}_{bu,k},{\tau}_{bu,k}) = -\bm{F}_{{\bm{y} }}(\bm{y}| \bm{\eta} ;{\delta}_{bU}, {\tau}_{bu,k}).
\end{aligned}
$$
All other entries in the FIM related to the time offset in the $b^{\text{th}}$ LEO link are zero.  The FIMs related to the frequency offset are presented next. The FIM of the frequency offset and the corresponding Doppler observed with respect to the $b^{\text{th}}$ LEO satellite  during the $k^{\text{th}}$ time slot is
$$
\begin{aligned}
\bm{F}_{{\bm{y} }}(\bm{y}| \bm{\eta} ;{\epsilon}_{bU}, {\nu}_{b,k}) = - 0.5 \times \underset{bu,k}{\operatorname{SNR}} f_{c}\alpha_{obu,k}^2.
\end{aligned}
$$
The FIM of the frequency offset in the FIM due to the observations of the received signals from $b^{\text{th}}$ LEO satellite to the $u^{\text{th}}$ receive antenna during the $k^{\text{th}}$ time slot is
$$
\begin{aligned}
\bm{F}_{{\bm{y} }}(\bm{y}| \bm{\eta} ;{\epsilon}_{bU}, {\epsilon}_{bU}) = 0.5 \times \underset{bu,k}{\operatorname{SNR}}  \alpha_{obu,k}^2.
\end{aligned}
$$

Considering the BSs-receiver link, the entries in FIM due to the observations of the signals at the receiver from $q^{\text{th}}$ BS can be obtained through the simplified expression.
$$
\begin{aligned}
     \bm{F}_{{\bm{y} }}(\bm{y}| \bm{\eta} ;\bm{\eta}_{qU},\bm{\eta}_{qU}) &=   \frac{1}{N_{01}} \times \\&  \sum_{u,k}^{N_U N_K}\Re\left\{ \int \nabla_{\bm{\eta}_{qU}}{\mu}_{qu,k} [t]\nabla_{\bm{\eta}_{qU}}{\mu}_{qu,k}^{\mathrm{H}}[t] \; \; d t \right\}.
\end{aligned}
$$
We now present the non-zero entries focusing on the $q^{\text{th}}$ BS. We start with the delays focusing on the FIM for the delay from the $q^{\text{th}}$ BS during the $k^{\text{th}}$ time slot on the $u^{\text{th}}$ receive antenna
$$
\begin{aligned}
\bm{F}_{{\bm{y} }}(\bm{y}| \bm{\eta} ;{\tau}_{qu,k},{\tau}_{qu,k}) &= -\bm{F}_{{\bm{y} }}(\bm{y}| \bm{\eta} ;{\tau}_{qu,k},{\delta}_{QU}) = \underset{qu,k}{\operatorname{SNR}} \; \omega_{qU,k},
\end{aligned}
$$
where $\omega_{qU,k} = \Bigg[ \alpha_{1q,k}^2 +  2f_{oqU,k} \alpha_{1q,k} \alpha_{2q,k} +  f_{oqU,k}^2   \Bigg].$

All other entries in the FIM focusing on delays in the $q^{\text{th}}$ BS link are zero. The FIM focusing on the Dopplers related to the $q^{\text{th}}$ BS is
$$
\begin{aligned}
\bm{F}_{{\bm{y} }}(\bm{y}| \bm{\eta} ;{\nu}_{qU,k}, {\nu}_{qU,k}) = 0.5 \times \underset{qu,k}{\operatorname{SNR}} f_{c}^2 \alpha_{oqu,k}^2.
\end{aligned}
$$
The FIM of the Doppler observed with respect to the $q^{\text{th}}$ BS and the corresponding frequency offset during the $k^{\text{th}}$ time slot is
$$
\begin{aligned}
\bm{F}_{{\bm{y} }}(\bm{y}| \bm{\eta} ;{\nu}_{qU,k}, {\epsilon}_{QU}) = - 0.5 \times \underset{qu,k}{\operatorname{SNR}} f_{c} \alpha_{oqu,k}^2.
\end{aligned}
$$
All other entries in the FIM related to the Dopplers in the $q^{\text{th}}$ BS link are zero. Now, we focus on the channel gain in the $q^{\text{th}}$ BS link. The FIM of the channel gain considering the received signals from $q^{\text{th}}$ BS to the $u^{\text{th}}$ receive antenna during the $k^{\text{th}}$ time slot is
$$
\begin{aligned}
\bm{F}_{{\bm{y} }}(\bm{y}| \bm{\eta} ; {\beta}_{qu,k}, {\beta}_{qu,k}) = \frac{1}{4 \pi^2 \left|\beta_{qu,k}\right|^2}\underset{qu,k}{\operatorname{SNR}}.
\end{aligned}
$$

All other entries in the FIM related to the channel gain in the $q^{\text{th}}$ BS link are zero. Now, we focus on the time offset in the $q^{\text{th}}$ BS link. The FIM between the time offset and the delay in the FIM due to the observations of the received signals from $q^{\text{th}}$ BS to the $u^{\text{th}}$ receive antenna during the $k^{\text{th}}$ time slot is
$$
\begin{aligned}
\bm{F}_{{\bm{y} }}(\bm{y}| \bm{\eta} ;{\delta}_{QU}, {\tau}_{qu,k}) = \bm{F}_{{\bm{y} }}(\bm{y}| \bm{\eta} ; {\tau}_{qu,k}, {\delta}_{QU}).
\end{aligned}
$$
The FIM of the time offset in the FIM due to the observations of the received signals from $q^{\text{th}}$ BS to the $u^{\text{th}}$ receive antenna during the $k^{\text{th}}$ time slot is
$$
\begin{aligned}
\bm{F}_{{\bm{y} }}(\bm{y}| \bm{\eta} ;{\delta}_{QU}, {\delta}_{QU}) &= \bm{F}_{{\bm{y} }}(\bm{y}| \bm{\eta} ; {\tau}_{qu,k},{\tau}_{qu,k}) \\&= -\bm{F}_{{\bm{y} }}(\bm{y}| \bm{\eta} ;{\delta}_{QU}, {\tau}_{qu,k}).
\end{aligned}
$$
All other entries in the FIM related to the time offset in the $q^{\text{th}}$ BS link are zero.  The FIMs related to the frequency offset are presented next. The FIM of the frequency offset and the corresponding Doppler observed with respect to the $q^{\text{th}}$ BS  during the $k^{\text{th}}$ time slot is
$$
\begin{aligned}
\bm{F}_{{\bm{y} }}(\bm{y}| \bm{\eta} ;{\epsilon}_{QU}, {\nu}_{qU,k}) = - 0.5 \times \underset{bu,k}{\operatorname{SNR}} f_{c}\alpha_{oqu,k}^2.
\end{aligned}
$$
The FIM of the frequency offset in the FIM due to the observations of the received signals from $q^{\text{th}}$ BS to the $u^{\text{th}}$ receive antenna during the $k^{\text{th}}$ time slot is
$$
\begin{aligned}
\bm{F}_{{\bm{y} }}(\bm{y}| \bm{\eta} ;{\epsilon}_{QU}, {\epsilon}_{QU}) = 0.5 \times \underset{qu,k}{\operatorname{SNR}}  \alpha_{oqu,k}^2.
\end{aligned}
$$

Considering the LEO-BS link, the entries in FIM due to the observations of the signals at the receiver from $b^{\text{th}}$ LEO can be obtained through the simplified expression.
$$
\begin{aligned}
     \bm{F}_{{\bm{y} }}(\bm{y}| \bm{\eta} ;\bm{\eta}_{bQ},\bm{\eta}_{bQ}) &=   \frac{1}{N_{02}} \times \\&  \sum_{k}^{ N_K}\Re\left\{ \int \nabla_{\bm{\eta}_{bQ}}{\mu}_{bq,k} [t]\nabla_{\bm{\eta}_{bQ}}{\mu}_{bq,k}^{\mathrm{H}}[t] \; \; d t \right\}.
\end{aligned}
$$
We now present the non-zero entries focusing on the $b^{\text{th}}$ LEO. We start with the delays focusing on the FIM for the delay from the $b^{\text{th}}$ LEO during the $k^{\text{th}}$ time slot at the $q^{\text{th}}$ BS
$$
\begin{aligned}
\bm{F}_{{\bm{y} }}(\bm{y}| \bm{\eta} ;{\tau}_{bq,k},{\tau}_{bq,k}) &= -\bm{F}_{{\bm{y} }}(\bm{y}| \bm{\eta} ;{\tau}_{bq,k},{\delta}_{bQ}) = \underset{bq,k}{\operatorname{SNR}} \; \omega_{bq,k},
\end{aligned}
$$
where $\omega_{bq,k} = \Bigg[ \alpha_{1q,k}^2 +  2f_{obq,k} \alpha_{1q,k} \alpha_{2q,k} +  f_{obq,k}^2   \Bigg].$

All other entries in the FIM focusing on delays in the $b^{\text{th}}$ LEO link are zero. The FIM focusing on the Dopplers related to the $b^{\text{th}}$ LEO is
$$
\begin{aligned}
\bm{F}_{{\bm{y} }}(\bm{y}| \bm{\eta} ;{\nu}_{bq,k}, {\nu}_{bq,k}) = 0.5 \times \underset{bq,k}{\operatorname{SNR}} f_{c}^2 \alpha_{obq,k}^2.
\end{aligned}
$$
The FIM of the Doppler observed with respect to the $b^{\text{th}}$ LEO and the corresponding frebqency offset during the $k^{\text{th}}$ time slot is
$$
\begin{aligned}
\bm{F}_{{\bm{y} }}(\bm{y}| \bm{\eta} ;{\nu}_{bq,k}, {\epsilon}_{bQ}) = - 0.5 \times \underset{bq,k}{\operatorname{SNR}} f_{c} \alpha_{obq,k}^2.
\end{aligned}
$$
All other entries in the FIM related to the Dopplers in the $b^{\text{th}}$ LEO link are zero. Now, we focus on the channel gain in the $b^{\text{th}}$ LEO link. The FIM of the channel gain considering the received signals from $b^{\text{th}}$ LEO to the $q^{\text{th}}$ BS during the $k^{\text{th}}$ time slot is
$$
\begin{aligned}
\bm{F}_{{\bm{y} }}(\bm{y}| \bm{\eta} ; {\beta}_{bq,k}, {\beta}_{bq,k}) = \frac{1}{4 \pi^2 \left|\beta_{bq,k}\right|^2}\underset{bq,k}{\operatorname{SNR}}.
\end{aligned}
$$

All other entries in the FIM related to the channel gain in the $b^{\text{th}}$ LEO link are zero. Now, we focus on the time offset in the $b^{\text{th}}$ LEO link. The FIM between the time offset and the delay in the FIM due to the observations of the received signals from $b^{\text{th}}$ LEO to the $q^{\text{th}}$ BS during the $k^{\text{th}}$ time slot is
$$
\begin{aligned}
\bm{F}_{{\bm{y} }}(\bm{y}| \bm{\eta} ;{\delta}_{bQ}, {\tau}_{bq,k}) = \bm{F}_{{\bm{y} }}(\bm{y}| \bm{\eta} ; {\tau}_{bq,k}, {\delta}_{bQ}).
\end{aligned}
$$
The FIM of the time offset in the FIM due to the observations of the received signals from $b^{\text{th}}$ LEO to the $q^{\text{th}}$ BS the $k^{\text{th}}$ time slot is
$$
\begin{aligned}
\bm{F}_{{\bm{y} }}(\bm{y}| \bm{\eta} ;{\delta}_{bQ}, {\delta}_{bQ}) = \bm{F}_{{\bm{y} }}(\bm{y}| \bm{\eta} ; {\tau}_{bq,k},{\tau}_{bq,k}) = -\bm{F}_{{\bm{y} }}(\bm{y}| \bm{\eta} ;{\delta}_{bQ}, {\tau}_{bq,k}).
\end{aligned}
$$
All other entries in the FIM related to the time offset in the $b^{\text{th}}$ LEO link are zero.  The FIMs related to the frequency offset are presented next. The FIM of the frequency offset and the corresponding Doppler observed with respect to the $b^{\text{th}}$ LEO  during the $k^{\text{th}}$ time slot is
$$
\begin{aligned}
\bm{F}_{{\bm{y} }}(\bm{y}| \bm{\eta} ;{\epsilon}_{bQ}, {\nu}_{bq,k}) = - 0.5 \times \underset{bq,k}{\operatorname{SNR}} f_{c}\alpha_{obq,k}^2.
\end{aligned}
$$
The FIM of the frequency offset in the FIM due to the observations of the received signals from $b^{\text{th}}$ LEO to the $q^{\text{th}}$ BS during the $k^{\text{th}}$ time slot is
$$
\begin{aligned}
\bm{F}_{{\bm{y} }}(\bm{y}| \bm{\eta} ;{\epsilon}_{bQ}, {\epsilon}_{bQ}) = 0.5 \times \underset{bq,k}{\operatorname{SNR}}  \alpha_{obq,k}^2.
\end{aligned}
$$
We have derived the FIM for the channel parameters. In the next section, we will use these derivations to present the FIM for the location parameters.

\section{FIM for Location Parameters}
In the previous sections, we have presented a system model that captures unsynchronized LEOs in time and frequency communicating with a receiver and a set of synchronized BSs. We also presented a system model incorporating the BSs communicating with the receiver. Subsequently, we derived the available information in the received signals using the FIM. In this section, we first highlight the location parameters and transform the FIM for the channel parameters into the FIM for location parameters. To highlight the location parameters, we i) focus on the unknown receiver position at $k = 0$, $\bm{p}_{U,0}$, ii) assume that the receiver velocity remains constant across all $N_K$ time slots, $\bm{v}_{U,k} = \bm{v}_{U,0} \; \; \forall k$, iii)  assume that the uncertainty associated with the
position of the $b^{\text{th}}$ LEO remains constant across all $N_K$ time slots $\check{\bm{p}}_{b,k} = \check{\bm{p}}_{b,0} \; \; \forall k$, and iv) assume that the uncertainty associated with the
velocity of the $b^{\text{th}}$ LEO remains constant across all $N_K$ time slots $\check{\bm{v}}_{b,k} = \check{\bm{v}}_{b,0} \; \; \forall k$. Now, we gather the location parameters as 

$$
\begin{aligned}
   & \bm{\kappa} = \\ &[\bm{p}_{U,0}, \bm{v}_{U,0}, \bm{\Phi}_{U},\check{\bm{p}}_{B,0}, \check{\bm{v}}_{B,0}, \bm{\zeta}_{1U},\cdots, \bm{\zeta}_{N_BU}, \bm{\zeta}_{1Q},\cdots, \bm{\zeta}_{N_BQ}, \\& \bm{\zeta}_{1U}, \cdots, \bm{\zeta}_{NQU}   ],
    \end{aligned}
    $$
$$\text{where}$$ 
$$
\begin{aligned}
\check{\bm{p}}_{B,0} &= \left[\check{\bm{p}}_{1,0}^{\mathrm{T}}, \cdots, \check{\bm{p}}_{N_{B},0}^{\mathrm{T}}\right]^{\mathrm{T}}, \\
\check{\bm{v}}_{B,0} &= \left[\check{\bm{v}}_{1,0}^{\mathrm{T}}, \cdots, \check{\bm{v}}_{N_{B},0}^{\mathrm{T}}\right]^{\mathrm{T}}, \\
\bm{\zeta}_{bU} &= \left[\bm{\beta}_{bU}^{\mathrm{T}}, \delta_{bU}, \epsilon_{bU}\right]^{\mathrm{T}}, \\
\bm{\zeta}_{bQ} &= \left[\bm{\beta}_{bQ}^{\mathrm{T}}, \delta_{bQ}, \epsilon_{bQ}\right]^{\mathrm{T}}, \\
\bm{\zeta}_{qU} &= \left[\bm{\beta}_{qU}^{\mathrm{T}}, \delta_{QU}, \epsilon_{QU}\right]^{\mathrm{T}}.
\end{aligned}
$$
The location parameter vector, $\bm{\kappa}$, can be divided into $\bm{\kappa}_{1} = [\bm{p}_{U,0}, \bm{v}_{U,0}, \bm{\Phi}_{U},\check{\bm{p}}_{B,0}, \check{\bm{v}}_{B,0}]$ and $\bm{\kappa}_{2} =   [\bm{\zeta}_{1U},\cdots, \bm{\zeta}_{N_BU}, \bm{\zeta}_{1Q}, \\ \cdots, \bm{\zeta}_{N_BQ},  \bm{\zeta}_{1U}, \cdots, \bm{\zeta}_{NQU}   ].$ The FIM for the location parameters is extracted from the FIM for the channel parameters, $\mathbf{J}{\bm{y}|\bm{\eta}}$, through the bijective transformation $\mathbf{J}{\bm{y}|\bm{\kappa}} \triangleq \mathbf{\Upsilon}{\bm{\kappa}} \mathbf{J}{\bm{y}|\bm{\eta}} \mathbf{\Upsilon}{\bm{\kappa}}^{\mathrm{T}}$. The matrix $\mathbf{\Upsilon}{\bm{\kappa}}$ captures the derivatives of the non-linear relationship between the geometric channel parameters, $ \bm{\eta}$, and the location parameters \cite{kay1993fundamentals}. The entries of the transformation matrix $\mathbf{\Upsilon}_{\bm{\kappa}}$ are laid out in Appendix \ref{Appendix_Entries_in_transformation_matrix}. The EFIM, taking  $\bm{\kappa}_{1} = [\bm{p}_{U,0}, \bm{v}_{U,0}, \bm{\Phi}_{U},\check{\bm{p}}_{b,0}, \check{\bm{v}}_{b,0}]$ as the parameter of interest and $\bm{\kappa}_{2} = [\bm{\zeta}_{1U},\cdots, \bm{\zeta}_{N_BU}, \bm{\zeta}_{1Q},\cdots, \bm{\zeta}_{N_BQ}, \bm{\zeta}_{1U}, \cdots, \bm{\zeta}_{NQU}     ]$ as the nuisance parameters, is now derived. 

\subsection{FIM for the parameters of interest}
Here, we present the FIM for the parameters of interest, $\bm{\kappa}_{1}.$ This FIM is represented by $\mathbf{J}_{\bm{y}|\bm{\kappa_{1}}}$, and the entries in this FIM are presented in the following Lemmas.

\begin{lemma}
\label{lemma:FIM_3D_position}
The FIM of the $3$D position of the receiver is

\begin{equation}
\begin{aligned}
\label{equ_lemma:FIM_3D_position}
&{\bm{F}_{{{y} }}(\bm{y}_{}| \bm{\eta} ;\bm{p}_{U,0},\bm{p}_{U,0}) = } \\&{\sum_{b,k^{},u^{}} \underset{bu,k}{\operatorname{SNR}}  \Bigg[\frac{\omega_{bU,k} }{c^2}   \bm{\Delta}_{bu,k} 
 \bm{\Delta}_{bu^{},k^{}}^{\mathrm{T}}}  + {\frac{f_{c}^2 \alpha_{obu,k}^2\nabla_{\bm{p}_{U,0}} \nu_{bU,k} \nabla_{\bm{p}_{U,0}}^{\mathrm{T}}\nu_{bU,k}}{2}  \Bigg]} + \\  &{\sum_{q,k^{},u^{}} \underset{qu,k}{\operatorname{SNR}}  \Bigg[\frac{\omega_{qU,k} }{c^2}   \bm{\Delta}_{qu,k} 
 \bm{\Delta}_{qu^{},k^{}}^{\mathrm{T}}}  + {\frac{f_{c}^2 \alpha_{oqu,k}^2\nabla_{\bm{p}_{U,0}} \nu_{qU,k} \nabla_{\bm{p}_{U,0}}^{\mathrm{T}}\nu_{qU,k}}{2}  \Bigg]} 
\end{aligned}
\end{equation}

\end{lemma}
\begin{proof}
See Appendix \ref{Appendix_lemma_FIM_3D_position}.
\end{proof}

\begin{lemma}
\label{lemma:FIM_3D_position_3D_velocity}
The FIM relating the $3$D position and $3$D velocity of the receiver is
\begin{equation}
\label{equ_lemma:FIM_3D_position_3D_velocity}
\begin{aligned}
&{\bm{F}_{{{y} }}(\bm{y}_{}| \bm{\eta} ;\bm{p}_{U,0},\bm{v}_{U,0}) = } \\ & 
\medmath{{\sum_{b,k^{},u^{}} \underset{bu,k}{\operatorname{SNR}}  \Bigg[ \frac{(k) \omega_{bU,k}\Delta_{t}}{c^2}    \bm{\Delta}_{bu,k} 
\bm{\Delta}_{bu,k}^{\mathrm{T}} }    - {\frac{f_{c}^2 \alpha_{obu,k}^2\nabla_{\bm{p}_{U,0}} \nu_{bU,k} \bm{\Delta}_{bU,k}^{\mathrm{T}}}{2 \; c}  \Bigg]}} + \\
&\medmath{{\sum_{q,k^{},u^{}} \underset{qu,k}{\operatorname{SNR}}  \Bigg[ \frac{(k) \omega_{qU,k}\Delta_{t}}{c^2}    \bm{\Delta}_{qu,k} 
\bm{\Delta}_{qu,k}^{\mathrm{T}} }    - {\frac{f_{c}^2 \alpha_{oqu,k}^2\nabla_{\bm{p}_{U,0}} \nu_{qU,k} \bm{\Delta}_{qU,k}^{\mathrm{T}}}{2 \; c}  \Bigg]}}.
\end{aligned}
\end{equation}

\end{lemma}
\begin{proof}
See Appendix \ref{Appendix_lemma_FIM_3D_position_3D_velocity}.
\end{proof}

\begin{lemma}
\label{lemma:FIM_3D_position_3D_orientation}
The FIM relating the $3$D position and $3$D orientation of the receiver is
\begin{equation}
\label{equ_lemma:FIM_3D_position_3D_orientation}
\begin{aligned}
{\bm{F}_{{{y} }}(\bm{y}_{}| \bm{\eta} ;\bm{p}_{U,0},\bm{\Phi}_{U})  } &=  {\sum_{b,k^{},u^{}} \underset{bu,k}{\operatorname{SNR}}  \Bigg[ \frac{\omega_{bU,k}}{c}   \bm{\Delta}_{bu,k} 
 \nabla_{\bm{\Phi}_{U}}^{\mathrm{T}} \tau_{bu^{},k^{}}}    \Bigg] \\ &+ {\sum_{q,k^{},u^{}} \underset{qu,k}{\operatorname{SNR}}  \Bigg[ \frac{\omega_{qU,k}}{c}   \bm{\Delta}_{qu,k} 
 \nabla_{\bm{\Phi}_{U}}^{\mathrm{T}} \tau_{qu^{},k^{}}}    \Bigg].
\end{aligned}
\end{equation}

\end{lemma}
\begin{proof}
See Appendix \ref{Appendix_lemma_FIM_3D_position_3D_orientation}.
\end{proof}

\begin{lemma}
\label{lemma:FIM_3D_position_3D_b_th_position_offset}
The FIM relating the $3$D position of the receiver and $\check{\bm{p}}_{b,0}$ is

\begin{equation}
\begin{aligned}
\label{equ_lemma:FIM_3D_position_3D_b_th_position_offset}
&{\bm{F}_{{{y} }}(\bm{y}_{}| \bm{\eta} ;\bm{p}_{U,0},\check{\bm{p}}_{b,0}) = } \\&{\sum_{k^{},u^{}}  \underset{bu,k}{\operatorname{SNR}}  \Bigg[\frac{-\omega_{bU,k} }{c^2}   \bm{\Delta}_{bu,k} 
 \bm{\Delta}_{bu^{},k^{}}^{\mathrm{T}}}  + {\frac{f_{c}^2 \alpha_{obu,k}^2\nabla_{\bm{p}_{U,0}} \nu_{bU,k} \nabla_{\check{\bm{p}}_{b,0}}^{\mathrm{T}}\nu_{bU,k}}{2}  \Bigg]} .
\end{aligned}
\end{equation}

\end{lemma}
\begin{proof}
See Appendix \ref{Appendix_lemma_FIM_3D_position_3D_b_th_position_offset}.
\end{proof}

\begin{lemma}
\label{lemma:FIM_3D_position_3D_b_th_velocity_offset}
The FIM relating the $3$D position of the receiver and $\check{\bm{v}}_{b,0}$ is
\begin{equation}
\label{equ_lemma:FIM_3D_position_3D_b_th_velocity_offset}
\begin{aligned}
&{\bm{F}_{{{y} }}(\bm{y}_{}| \bm{\eta} ;\bm{p}_{U,0},\check{\bm{v}}_{b,0}) = } \\ & 
\medmath{{\sum_{k^{},u^{}} \underset{bu,k}{\operatorname{SNR}}  \Bigg[ \frac{-(k) \omega_{bU,k}\Delta_{t}}{c^2}    \bm{\Delta}_{bu,k} 
\bm{\Delta}_{bu,k}^{\mathrm{T}} }    + {\frac{f_{c}^2 \alpha_{obu,k}^2\nabla_{\bm{p}_{U,0}} \nu_{bU,k} \bm{\Delta}_{bU,k}^{\mathrm{T}}}{2 \; c^{}}  \Bigg]}}.
\end{aligned}
\end{equation}

\end{lemma}
\begin{proof}
See Appendix \ref{Appendix_lemma_FIM_3D_position_3D_b_th_velocity_offset}.
\end{proof}

\begin{lemma}
\label{lemma:FIM_3D_velocity_3D_velocity}
The FIM of the $3$D velocity of the receiver is
\begin{equation}
\label{equ_lemma:FIM_3D_velocity_3D_velocity}
\begin{aligned}
&{\bm{F}_{{{y} }}(\bm{y}_{}| \bm{\eta} ;\bm{v}_{U,0},\bm{v}_{U,0}) = } \\ & 
\medmath{{\sum_{b,k^{},u^{}} \underset{bu,k}{\operatorname{SNR}}  \Bigg[ \frac{(k)^2 \omega_{bU,k}\Delta_{t}^{2}}{c^2}    \bm{\Delta}_{bu,k} 
\bm{\Delta}_{bu,k}^{\mathrm{T}} }    + {\frac{f_{c}^2 \alpha_{obu,k}^2 \bm{\Delta}_{bU,k} \bm{\Delta}_{bU,k}^{\mathrm{T}}}{2 \; c^{2}}  \Bigg]}} + \\
&\medmath{{\sum_{q,k^{},u^{}} \underset{qu,k}{\operatorname{SNR}}  \Bigg[ \frac{(k)^{2} \omega_{qU,k}\Delta_{t}^{2}}{c^2}    \bm{\Delta}_{qu,k} 
\bm{\Delta}_{qu,k}^{\mathrm{T}} }    + {\frac{f_{c}^2 \alpha_{oqu,k}^2\bm{\Delta}_{qU,k} \bm{\Delta}_{qU,k}^{\mathrm{T}}}{2 \; c^{2}}  \Bigg]}}.
\end{aligned}
\end{equation}

\end{lemma}
\begin{proof}
See Appendix \ref{Appendix_lemma_FIM_3D_velocity_3D_velocity}.
\end{proof}

\begin{lemma}
\label{lemma:FIM_3D_velocity_3D_orientation}
The FIM relating the $3$D velocity and $3$D orientation of the receiver is
\begin{equation}
\label{equ_lemma:FIM_3D_velocity_3D_orientation}
\begin{aligned}
&{\bm{F}_{{{y} }}(\bm{y}_{}| \bm{\eta} ;\bm{v}_{U,0},\bm{\Phi}_{U})  } =
\\ &{\sum_{b,k^{},u^{}} \underset{bu,k}{\operatorname{SNR}}  \Bigg[ \frac{(k) \omega_{bU,k}\Delta_{t}^{}}{c}  \bm{\Delta}_{bu,k} 
 \nabla_{\bm{\Phi}_{U}}^{\mathrm{T}} \tau_{bu^{},k^{}}}    \Bigg] \\ &+ {\sum_{q,k^{},u^{}} \underset{qu,k}{\operatorname{SNR}}  \Bigg[  \frac{(k) \omega_{qU,k}\Delta_{t}^{}}{c}  \bm{\Delta}_{qu,k} 
 \nabla_{\bm{\Phi}_{U}}^{\mathrm{T}} \tau_{qu^{},k^{}}}    \Bigg].
\end{aligned}
\end{equation}

\end{lemma}
\begin{proof}
See Appendix \ref{Appendix_lemma_FIM_3D_velocity_3D_orientation}.
\end{proof}

\begin{lemma}
\label{lemma:FIM_3D_velocity_3D_b_th_position_offset}
The FIM relating the $3$D velocity of the receiver and $\check{\bm{p}}_{b,0}$ is

\begin{equation}
\begin{aligned}
\label{equ_lemma:FIM_3D_velocity_3D_b_th_position_offset}
&{\bm{F}_{{{y} }}(\bm{y}_{}| \bm{\eta} ;\bm{v}_{U,0},\check{\bm{p}}_{b,0}) = \sum_{k^{},u^{}}\underset{bu,k}{\operatorname{SNR}}} \\&{    \Bigg[\frac{-(k) \omega_{bU,k}\Delta_{t}^{}}{c^{2}}   \bm{\Delta}_{bu,k} 
 \bm{\Delta}_{bu^{},k^{}}^{\mathrm{T}}}  - {\frac{f_{c}^2 \alpha_{obu,k}^2\bm{\Delta}_{bU,k} \nabla_{\check{\bm{p}}_{b,0}}^{\mathrm{T}}\nu_{bU,k}}{2 \; c}  \Bigg]} .
\end{aligned}
\end{equation}

\end{lemma}
\begin{proof}
See Appendix \ref{Appendix_lemma_FIM_3D_velocity_3D_b_th_position_offset}.
\end{proof}

\begin{lemma}
\label{lemma:FIM_3D_velocity_3D_b_th_velocity_offset}
The FIM relating the $3$D velocity of the receiver and $\check{\bm{v}}_{b,0}$ is
\begin{equation}
\label{equ_lemma:FIM_3D_velocity_3D_b_th_velocity_offset}
\begin{aligned}
&{\bm{F}_{{{y} }}(\bm{y}_{}| \bm{\eta} ;\bm{v}_{U,0},\check{\bm{v}}_{b,0}) = } \\ & 
\medmath{{\sum_{k^{},u^{}} \underset{bu,k}{\operatorname{SNR}}  \Bigg[ \frac{-(k)^{2} \omega_{bU,k}\Delta_{t}^{2}}{c^2}    \bm{\Delta}_{bu,k} 
\bm{\Delta}_{bu,k}^{\mathrm{T}} }    - {\frac{f_{c}^2 \alpha_{obu,k}^2\bm{\Delta}_{bU,k} \bm{\Delta}_{bU,k}^{\mathrm{T}}}{2 \; c^{2}}  \Bigg]}}.
\end{aligned}
\end{equation}

\end{lemma}
\begin{proof}
See Appendix \ref{Appendix_lemma_FIM_3D_velocity_3D_b_th_velocity_offset}.
\end{proof}

\begin{lemma}
\label{lemma:FIM_3D_orientation_3D_orientation}
The FIM for the $3$D orientation of the receiver is
\begin{equation}
\label{equ_lemma:FIM_3D_orientation_3D_orientation}
\begin{aligned}
&{\bm{F}_{{{y} }}(\bm{y}_{}| \bm{\eta} ;\bm{\Phi}_{U},\bm{\Phi}_{U})  } =
\\ &{\sum_{b,k^{},u^{}} \underset{bu,k}{\operatorname{SNR}}  \Bigg[\omega_{bU,k}  \nabla_{\bm{\Phi}_{U}}\tau_{bu^{},k^{}} 
 \nabla_{\bm{\Phi}_{U}}^{\mathrm{T}} \tau_{bu^{},k^{}}}    \Bigg] \\ &+ {\sum_{q,k^{},u^{}} \underset{qu,k}{\operatorname{SNR}}  \Bigg[\omega_{qU,k}   \nabla_{\bm{\Phi}_{U}} \tau_{qu^{},k^{}}
 \nabla_{\bm{\Phi}_{U}}^{\mathrm{T}} \tau_{qu^{},k^{}}}    \Bigg].
\end{aligned}
\end{equation}

\end{lemma}
\begin{proof}
See Appendix \ref{Appendix_lemma_FIM_3D_orientation_3D_orientation}.
\end{proof}

\begin{lemma}
\label{lemma:FIM_3D_orientation_3D_b_th_position_offset}
The FIM relating the $3$D orientation of the receiver and $\check{\bm{p}}_{b,0}$ is

\begin{equation}
\begin{aligned}
\label{equ_lemma:FIM_3D_orientation_3D_b_th_position_offset}
&{\bm{F}_{{{y} }}(\bm{y}_{}| \bm{\eta} ;\bm{\Phi}_{U},\check{\bm{p}}_{b,0}) = - \sum_{k^{},u^{}}\underset{bu,k}{\operatorname{SNR}}} {    \frac{\omega_{bU,k}}{c^{}}  \nabla_{\bm{\Phi}_{U}}\tau_{bu^{},k^{}}  
 \bm{\Delta}_{bu^{},k^{}}^{\mathrm{T}}}  .
\end{aligned}
\end{equation}

\end{lemma}
\begin{proof}
See Appendix \ref{Appendix_lemma_FIM_3D_orientation_3D_b_th_position_offset}.
\end{proof}

\begin{lemma}
\label{lemma:FIM_3D_orientation_3D_b_th_velocity_offset}
The FIM relating the $3$D orientation of the receiver and $\check{\bm{v}}_{b,0}$ is
\begin{equation}
\label{equ_lemma:FIM_3D_orientation_3D_b_th_velocity_offset}
\begin{aligned}
&{\bm{F}_{{{y} }}(\bm{y}_{}| \bm{\eta} ;\bm{\Phi}_{U},\check{\bm{v}}_{b,0}) = }  
\medmath{-{\sum_{k^{},u^{}} \underset{bu,k}{\operatorname{SNR}}  \frac{(k)^{} \omega_{bU,k}\Delta_{t}^{}}{c} \nabla_{\bm{\Phi}_{U}}\tau_{bu^{},k^{}}
\bm{\Delta}_{bu,k}^{\mathrm{T}} }   }.
\end{aligned}
\end{equation}

\end{lemma}
\begin{proof}
See Appendix \ref{Appendix_lemma_FIM_3D_orientation_3D_b_th_velocity_offset}.
\end{proof}

\begin{lemma}
\label{lemma:FIM_3D_b_th_position_offset_3D_b_th_position_offset}
The FIM for the $3$D position uncertainty associated with the $b^{\text{th}}$ LEO, $\check{\bm{p}}_{b,0}$ is

\begin{equation}
\begin{aligned}
\label{equ_lemma:FIM_3D_b_th_position_offset_3D_b_th_position_offset}
&\bm{F}_{{{y} }}(\bm{y}_{}| \bm{\eta} ;\check{\bm{p}}_{b,0},\check{\bm{p}}_{b,0}) = \\&{\sum_{k^{},u^{}} \underset{bu,k}{\operatorname{SNR}}  \Bigg[\frac{\omega_{bU,k} }{c^2}   \bm{\Delta}_{bu,k} 
 \bm{\Delta}_{bu^{},k^{}}^{\mathrm{T}}}  + {\frac{f_{c}^2 \alpha_{obu,k}^2\nabla_{\check{\bm{p}}_{b,0}} \nu_{bU,k} \nabla_{\check{\bm{p}}_{b,0}}^{\mathrm{T}}\nu_{bU,k}}{2}  \Bigg]} + \\  &{\sum_{q,k^{}} \underset{bq,k}{\operatorname{SNR}}  \Bigg[\frac{\omega_{bq,k} }{c^2}   \bm{\Delta}_{bq,k} 
 \bm{\Delta}_{bq,k}^{\mathrm{T}}}  + {\frac{f_{c}^2 \alpha_{obq,k}^2\nabla_{\check{\bm{p}}_{b,0}} \nu_{bq,k} \nabla_{\check{\bm{p}}_{b,0}}^{\mathrm{T}}\nu_{bq,k}}{2}  \Bigg]}.
\end{aligned}
\end{equation}

\end{lemma}
\begin{proof}
See Appendix \ref{Appendix_lemma_FIM_3D_b_th_position_offset_3D_b_th_position_offset}.
\end{proof}

\begin{lemma}
\label{lemma:FIM_3D_b_th_position_offset_3D_b_th_velocity_offset}
The FIM relating the $3$D position uncertainty associated with the $b^{\text{th}}$ LEO, $\check{\bm{p}}_{b,0}$ and the $3$D velocity uncertainty associated with the $b^{\text{th}}$ LEO, $\check{\bm{v}}_{b,0}$ is
\begin{equation}
\label{equ_lemma:FIM_3D_b_th_position_offset_3D_b_th_velocity_offset}
\begin{aligned}
&\bm{F}_{{{y} }}(\bm{y}_{}| \bm{\eta} ;\check{\bm{p}}_{b,0},\check{\bm{v}}_{b,0}) = \\&\medmath{\sum_{k^{},u^{}} \underset{bu,k}{\operatorname{SNR}}  \Bigg[\frac{(k) \omega_{bU,k}\Delta_{t}^{} }{c^2}   \bm{\Delta}_{bu,k} 
 \bm{\Delta}_{bu^{},k^{}}^{\mathrm{T}}}  + \medmath{\frac{f_{c}^2 \alpha_{obu,k}^2\nabla_{\check{\bm{p}}_{b,0}} \nu_{bU,k} \bm{\Delta}_{bU,k}^{\mathrm{T}}}{2 \; c}  \Bigg]} + \\  &\medmath{\sum_{q,k^{}} \underset{bq,k}{\operatorname{SNR}}  \Bigg[\frac{(k) \omega_{bq,k}\Delta_{t}^{} }{c^2}   \bm{\Delta}_{bq,k} 
 \bm{\Delta}_{bq^{},k^{}}^{\mathrm{T}}}  + \medmath{\frac{f_{c}^2 \alpha_{obq,k}^2\nabla_{\check{\bm{p}}_{b,0}} \nu_{bq,k} \bm{\Delta}_{bq,k}^{\mathrm{T}}}{2 \; c}  \Bigg]}.
\end{aligned}
\end{equation}

\end{lemma}
\begin{proof}
See Appendix \ref{Appendix_lemma_FIM_3D_b_th_position_offset_3D_b_th_velocity_offset}.
\end{proof}

\begin{lemma}
\label{lemma:FIM_3D_b_th_velocity_offset_3D_b_th_velocity_offset}
The FIM for the $3$D velocity uncertainty associated with the $b^{\text{th}}$ LEO, $\check{\bm{v}}_{b,0}$ is
\begin{equation}
\label{equ_lemma:FIM_3D_b_th_velocity_offset_3D_b_th_velocity_offset}
\begin{aligned}
&\bm{F}_{{{y} }}(\bm{y}_{}| \bm{\eta} ;\check{\bm{v}}_{b,0},\check{\bm{v}}_{b,0}) = \\&\medmath{\sum_{k^{},u^{}} \underset{bu,k}{\operatorname{SNR}}  \Bigg[\frac{(k)^2 \omega_{bU,k}\Delta_{t}^{2} }{c^2}   \bm{\Delta}_{bu,k} 
 \bm{\Delta}_{bu^{},k^{}}^{\mathrm{T}}}  + \medmath{\frac{f_{c}^2 \alpha_{obu,k}^2\bm{\Delta}_{bU,k} \bm{\Delta}_{bU,k}^{\mathrm{T}}}{2 \; c^2}  \Bigg]} + \\  &\medmath{\sum_{q,k^{}} \underset{bq,k}{\operatorname{SNR}}  \Bigg[\frac{(k)^{2} \omega_{bq,k}\Delta_{t}^{2} }{c^{2}}   \bm{\Delta}_{bq,k} 
 \bm{\Delta}_{bq^{},k^{}}^{\mathrm{T}}}  + \medmath{\frac{f_{c}^2 \alpha_{obq,k}^2\bm{\Delta}_{bq,k} \bm{\Delta}_{bq,k}^{\mathrm{T}}}{2 \; c^2}  \Bigg]}.
\end{aligned}
\end{equation}

\end{lemma}
\begin{proof}
See Appendix \ref{Appendix_lemma_FIM_3D_b_th_velocity_offset_3D_b_th_velocity_offset}.
\end{proof}

\subsection{Loss in information due to uncertainty about the parameters of interest}
The reduction in information about the parameter of interest due to the nuisance parameters is presented in this section. This reduction in information is defined by $\mathbf{J}_{ \bm{\bm{y}}; \bm{\kappa}_1}^{nu}$.

\begin{lemma}
\label{lemma:information_loss_FIM_3D_position}
The loss of information about $3$D position of the receiver due to uncertainty in the nuisance parameters ${\bm{\kappa}_{2}}$ is given by (\ref{equ_lemma:information_loss_FIM_3D_position}).

\begin{figure*}
\begin{align}
\begin{split}
\label{equ_lemma:information_loss_FIM_3D_position}
&{\bm{G}_{{{y} }}(\bm{y}_{}| \bm{\eta} ;\bm{p}_{U,0},\bm{p}_{U,0}) = } \medmath{\sum_{b}   \norm{\sum_{k^{},u^{}} \underset{bu^{},k^{}}{\operatorname{SNR}}\bm{\Delta}_{bu^{},k^{}}^{\mathrm{T}} \frac{ \omega_{bU,k}}{c} }}^{2}    \medmath{   \left(\sum_{u,k} \underset{bu,k}{\operatorname{SNR}} \omega_{bU,k}\right)^{\mathrm{-1}}  }  +
 \medmath{   \norm{\sum_{q,u^{},k^{}} \underset{qu^{},k^{}}{\operatorname{SNR}}\bm{\Delta}_{qu^{},k^{}}^{\mathrm{T}} \frac{ \omega_{qU,k}}{c} }}^{2}    \medmath{   \left(\sum_{q,u,k} \underset{qu,k}{\operatorname{SNR}} \omega_{qU,k}\right)^{\mathrm{-1}}  } \\& +  
\medmath{\sum_{b}\norm{{\sum_{u^{},k^{} } \underset{bu^{},k^{}}{\operatorname{SNR}} \; \; \nabla_{\bm{p}_{U,0}}^{\mathrm{T}} \nu_{bU,k^{}}   }  \frac{(f_{c}^{}) (\alpha_{obu,k^{}}^{2})}{2} }^{2}  \left(\sum_{u,k} \frac{\underset{bu,k}{\operatorname{SNR}}  \alpha_{obu,k}^{2}}{2}\right)^{-1}}  + \medmath{\norm{{\sum_{q^{},u,k^{} } \underset{q^{},u,k^{}}{\operatorname{SNR}} \; \; \nabla_{\bm{p}_{U,0}}^{\mathrm{T}} \nu_{qU,k^{}}   }  \frac{(f_{c}^{}) (\alpha_{oqu,k^{}}^{2})}{2} }^{2}  \left(\sum_{q,u,k} \frac{\underset{qu,k}{\operatorname{SNR}}  \alpha_{oqu,k}^{2}}{2}\right)^{-1}}
\end{split}
\end{align}
\end{figure*}

\end{lemma}
\begin{proof}
See Appendix \ref{Appendix_lemma_information_loss_FIM_3D_position}.
\end{proof}

\begin{lemma}
\label{lemma:information_loss_FIM_3D_position_3D_veloctiy}
The loss of information about the FIM of the $3$D position and $3$D velocity of the receiver due to uncertainty in the nuisance parameters ${\bm{\kappa}_{2}}$ is given by (\ref{equ_lemma:information_loss_FIM_3D_position_3D_veloctiy}).

\begin{figure*}
\begin{align}
\begin{split}
\label{equ_lemma:information_loss_FIM_3D_position_3D_veloctiy}
&\medmath{\bm{G}_{{{y} }}(\bm{y}_{}| \bm{\eta} ;\bm{p}_{U,0},\bm{v}_{U,0}) = \frac{\Delta_{t}}{c^2} \sum_{b,k^{},u^{}k^{'},u^{'}} \underset{bu^{},k^{}}{\operatorname{SNR}} \underset{bu^{'},k^{'}}{\operatorname{SNR}} } \medmath{   \bm{\Delta}_{bu^{},k^{}}  (k^{'}) {\bm{\Delta}_{bu^{'},k^{'}}^{\mathrm{T}}} \omega_{bU,k} \omega_{bU,k^{'}}}    \medmath{   \left(\sum_{u,k} \underset{bu,k}{\operatorname{SNR}} \omega_{bU,k}\right)^{\mathrm{-1}}  + \frac{\Delta_{t}}{c^2} \sum_{q,q^{'},u^{},k^{}u^{'},k^{'}} \underset{q^{},u,k^{}}{\operatorname{SNR}} \underset{q^{'},u^{'},k^{'}}{\operatorname{SNR}} } \medmath{   \bm{\Delta}_{q^{}u^{},k^{}}}  \\& \medmath{(k^{'})} \medmath{{\bm{\Delta}_{q^{'}u^{'},k^{'}}^{\mathrm{T}}} \omega_{qU,k} \omega_{q^{'}U,k^{'}}}    \medmath{   \left(\sum_{qu,k} \underset{qu,k}{\operatorname{SNR}} \omega_{qU,k}\right)^{\mathrm{-1}} - \frac{1}{c}\sum_{b,k^{},u^{}k^{'},u^{'}}}  \medmath{  \underset{bu^{},k^{}}{\operatorname{SNR}} \underset{bu^{'},k^{'}}{\operatorname{SNR}}  }  \medmath{{ \nabla_{\bm{p}_{U,0}} \nu_{bU,k^{}} \bm{\Delta}_{bU,k^{'}}^{\mathrm{T}}   }  \frac{(f_{c}^{2}) (\alpha_{obu,k^{}}^{2}\alpha_{obu^{'},k^{'}}^{2})}{4}   \left(\sum_{u,k} \frac{\underset{b^{}u,k^{}}{\operatorname{SNR}}  \alpha_{obu,k}^{2}}{2}\right)^{-1}}  \\& - \frac{1}{c}\sum_{q^{},u^{}q^{'},u^{'},k^{},k^{'}} \medmath{  \underset{qu^{},k^{}}{\operatorname{SNR}} \underset{q^{'}u^{'},k^{'}}{\operatorname{SNR}}  }  \medmath{{ \nabla_{\bm{p}_{U,0}} \nu_{qU,k^{}} \bm{\Delta}_{q^{'}U^{},k^{'}}^{\mathrm{T}}   }  \frac{(f_{c}^{2}) (\alpha_{oqu,k^{}}^{2}\alpha_{oq^{'}u^{'},k^{'}}^{2})}{4}  }  \medmath{ \left(\sum_{q,u,k} \frac{\underset{q^{}u,k^{}}{\operatorname{SNR}}  \alpha_{oqu,k}^{2}}{2}\right)^{-1}}
\end{split}
\end{align}
\end{figure*}
\end{lemma}
\begin{proof}
See Appendix \ref{Appendix_lemma_information_loss_FIM_3D_position_3D_velocity}.
\end{proof}

\begin{lemma}
\label{lemma:information_loss_FIM_3D_position_3D_orientation}
The loss of information about the FIM of the $3$D position and $3$D orientation of the receiver due to uncertainty in the nuisance parameters ${\bm{\kappa}_{2}}$ is given by (\ref{equ_lemma:information_loss_FIM_3D_position_3D_orientation}).

\begin{figure*}
\begin{align}
\begin{split}
\label{equ_lemma:information_loss_FIM_3D_position_3D_orientation}
\medmath{\bm{G}_{{{y} }}(\bm{y}_{}| \bm{\eta} ;\bm{p}_{U,0},\bm{\Phi}_{U})} &= \medmath{\frac{1}{c}\sum_{b,k^{},u^{}k^{'},u^{'}} \underset{bu^{},k^{}}{\operatorname{SNR}} \underset{bu^{'},k^{'}}{\operatorname{SNR}} } \medmath{   \bm{\Delta}_{bu^{},k^{}}  \nabla_{\bm{\Phi}_{U}}^{\mathrm{T}} \tau_{bu^{'},k^{'}}  \omega_{bU,k} \omega_{bU,k^{'}}}    \medmath{   \left(\sum_{u,k} \underset{bu,k}{\operatorname{SNR}} \omega_{bU,k}\right)^{\mathrm{-1}}} \\ &+  \medmath{\frac{1}{c}\sum_{q,q^{'},u^{},k^{}u^{'},k^{'}} \underset{q^{},u,k^{}}{\operatorname{SNR}} \underset{q^{'},u^{'},k^{'}}{\operatorname{SNR}} } \medmath{   \bm{\Delta}_{q^{}u^{},k^{}}   \nabla_{\bm{\Phi}_{U}}^{\mathrm{T}} \tau_{q^{'}u^{'},k^{'}}} \medmath{ \omega_{qU,k} \omega_{q^{'}U,k^{'}}}    \medmath{   \left(\sum_{qu,k} \underset{qu,k}{\operatorname{SNR}} \omega_{qU,k}\right)^{\mathrm{-1}}} 
\end{split}
\end{align}
\end{figure*}
\end{lemma}
\begin{proof}
See Appendix \ref{Appendix_lemma_information_loss_FIM_3D_position_3D_orientation}.
\end{proof}

\begin{lemma}
\label{lemma:information_loss_FIM_3D_position_3D_bth_position_uncertainty}
The loss of information about the FIM of the $\bm{p}_{U,0}$ and $\check{\bm{p}}_{b,0}$  due to uncertainty in the nuisance parameters ${\bm{\kappa}_{2}}$ is given by (\ref{equ_lemma:information_loss_3D_position_3D_bth_position_uncertainty}).

\begin{figure*}
\begin{align}
\begin{split}
\label{equ_lemma:information_loss_3D_position_3D_bth_position_uncertainty}
&\medmath{\bm{G}_{{{y} }}(\bm{y}_{}| \bm{\eta} ;\bm{p}_{U,0},\check{\bm{p}}_{b,0}) = \frac{-1}{c^2} \sum_{k^{},u^{}k^{'},u^{'}} \underset{bu^{},k^{}}{\operatorname{SNR}} \underset{bu^{'},k^{'}}{\operatorname{SNR}} } \medmath{   \bm{\Delta}_{bu^{},k^{}}  {\bm{\Delta}_{bu^{'},k^{'}}^{\mathrm{T}}} \omega_{bU,k} \omega_{bU,k^{'}}}    \medmath{   \left(\sum_{u,k} \underset{bu,k}{\operatorname{SNR}} \omega_{bU,k}\right)^{\mathrm{-1}}}  \\ &+ \sum_{k^{},u^{}, k^{'},u^{'}}  \medmath{  \underset{bu^{},k^{}}{\operatorname{SNR}} \underset{bu^{'},k^{'}}{\operatorname{SNR}}  }  \medmath{{ \nabla_{\bm{p}_{U,0}} \nu_{bU,k^{}} \nabla_{\check{\bm{p}}_{b,0}}^{\mathrm{T}} \nu_{bU,k^{'}}   }  \frac{(f_{c}^{2}) (\alpha_{obu,k^{}}^{2}\alpha_{obu^{'},k^{'}}^{2})}{4}   \left(\sum_{u,k} \frac{\underset{b^{}u,k^{}}{\operatorname{SNR}}  \alpha_{obu,k}^{2}}{2}\right)^{-1}} 
\end{split}
\end{align}
\end{figure*}
\end{lemma}
\begin{proof}
See Appendix \ref{Appendix_lemma_information_loss_FIM_3D_position_3D_bth_position_uncertainty}.
\end{proof}

\begin{lemma}
\label{lemma:information_loss_FIM_3D_position_3D_bth_velocity_uncertainty}
The loss of information about the FIM of the $\bm{p}_{U,0}$ and $\check{\bm{v}}_{b,0}$  due to uncertainty in the nuisance parameters ${\bm{\kappa}_{2}}$ is given by (\ref{equ_lemma:information_loss_FIM_3D_position_3D_bth_velocity_uncertainty}).

\begin{figure*}
\begin{align}
\begin{split}
\label{equ_lemma:information_loss_FIM_3D_position_3D_bth_velocity_uncertainty}
\medmath{\bm{G}_{{{y} }}(\bm{y}_{}| \bm{\eta} ;\bm{p}_{U,0},\check{\bm{v}}_{b,0})} =& \medmath{ -\frac{\Delta_{t}}{c^2} \sum_{k^{},u^{}k^{'},u^{'}} \underset{bu^{},k^{}}{\operatorname{SNR}} \underset{bu^{'},k^{'}}{\operatorname{SNR}} } \medmath{   \bm{\Delta}_{bu^{},k^{}}  (k^{'}) {\bm{\Delta}_{bu^{'},k^{'}}^{\mathrm{T}}} \omega_{bU,k} \omega_{bU,k^{'}}}    \medmath{   \left(\sum_{u,k} \underset{bu,k}{\operatorname{SNR}} \omega_{bU,k}\right)^{\mathrm{-1}}} \\&+ \frac{1}{c}\sum_{k^{},u^{}k^{'},u^{'}}  \medmath{  \underset{bu^{},k^{}}{\operatorname{SNR}} \underset{bu^{'},k^{'}}{\operatorname{SNR}}  }  \medmath{{ \nabla_{\bm{p}_{U,0}} \nu_{bU,k^{}} \bm{\Delta}_{bU,k^{'}}^{\mathrm{T}}   }  \frac{(f_{c}^{2}) (\alpha_{obu,k^{}}^{2}\alpha_{obu^{'},k^{'}}^{2})}{4}   \left(\sum_{u,k} \frac{\underset{b^{}u,k^{}}{\operatorname{SNR}}  \alpha_{obu,k}^{2}}{2}\right)^{-1}} 
\end{split}
\end{align}
\end{figure*}
\end{lemma}
\begin{proof}
See Appendix \ref{Appendix_lemma_information_loss_FIM_3D_position_3D_bth_velocity_uncertainty}.
\end{proof}

\begin{lemma}
\label{lemma:information_loss_FIM_3D_velocity_3D_velocity}
The loss of information about $3$D velocity of the receiver due to uncertainty in the nuisance parameters ${\bm{\kappa}_{2}}$ is given by (\ref{equ_lemma:information_loss_FIM_3D_velocity_3D_velocity}).

\begin{figure*}
\begin{align}
\begin{split}
\label{equ_lemma:information_loss_FIM_3D_velocity_3D_velocity}
&{\bm{G}_{{{y} }}(\bm{y}_{}| \bm{\eta} ;\bm{v}_{U,0},\bm{v}_{U,0}) = } \medmath{\sum_{b}   \norm{\sum_{k^{},u^{}} \underset{bu^{},k^{}}{\operatorname{SNR}}\bm{\Delta}_{bu^{},k^{}}^{\mathrm{T}} \frac{\Delta_{t}(k) \omega_{bU,k}}{c}}}^{2}    \medmath{   \left(\sum_{u,k} \underset{bu,k}{\operatorname{SNR}} \omega_{bU,k}\right)^{\mathrm{-1}}  }  + \medmath{   \norm{\sum_{q,u^{},k^{}} \underset{qu^{},k^{}}{\operatorname{SNR}}\bm{\Delta}_{qu^{},k^{}}^{\mathrm{T}} \frac{ \Delta_{t}(k)\omega_{qU,k}}{c} }}^{2}    \medmath{   \left(\sum_{q,u,k} \underset{qu,k}{\operatorname{SNR}} \omega_{qU,k}\right)^{\mathrm{-1}}  } \\ &+  
\medmath{\sum_{b}\norm{{\sum_{u^{},k^{} } \underset{bu^{},k^{}}{\operatorname{SNR}} \; \; \bm{\Delta}_{b^{}U^{},k^{}}^{\mathrm{T}}      }  \frac{(f_{c}^{}) (\alpha_{obu,k^{}}^{2})}{2 c} }^{2}  \left(\sum_{u,k} \frac{\underset{bu,k}{\operatorname{SNR}}  \alpha_{obu,k}^{2}}{2}\right)^{-1}}  + \medmath{\norm{{\sum_{q^{},u,k^{} } \underset{q^{},u,k^{}}{\operatorname{SNR}} \; \; \bm{\Delta}_{q^{}U^{},k^{}}^{\mathrm{T}}   }  \frac{(f_{c}^{}) (\alpha_{oqu,k^{}}^{2})}{2 c} }^{2 }  \left(\sum_{q,u,k} \frac{\underset{qu,k}{\operatorname{SNR}}  \alpha_{oqu,k}^{2}}{2 }\right)^{-1}}
\end{split}
\end{align}
\end{figure*}

\end{lemma}
\begin{proof}
See Appendix \ref{Appendix_lemma_information_loss_FIM_3D_velocity_3D_velocity}.
\end{proof}

\begin{lemma}
\label{lemma:information_loss_FIM_3D_velocity_3D_orientation}
The loss of information about the FIM of the $3$D velocity and $3$D orientation of the receiver due to uncertainty in the nuisance parameters ${\bm{\kappa}_{2}}$ is given by (\ref{equ_lemma:information_loss_FIM_3D_velocity_3D_orientation}).
\begin{figure*}
\begin{align}
\begin{split}
\label{equ_lemma:information_loss_FIM_3D_velocity_3D_orientation}
\medmath{\bm{G}_{{{y} }}(\bm{y}_{}| \bm{\eta} ;\bm{v}_{U,0},\bm{\Phi}_{U})} &= \medmath{\frac{\Delta_{t}}{c}\sum_{b,k^{},u^{}k^{'},u^{'}} \underset{bu^{},k^{}}{\operatorname{SNR}} \underset{bu^{'},k^{'}}{\operatorname{SNR}} } \medmath{ (k)  \bm{\Delta}_{bu^{},k^{}}  \nabla_{\bm{\Phi}_{U}}^{\mathrm{T}} \tau_{bu^{'},k^{'}}  \omega_{bU,k} \omega_{bU,k^{'}}}    \medmath{   \left(\sum_{u,k} \underset{bu,k}{\operatorname{SNR}} \omega_{bU,k}\right)^{\mathrm{-1}}} \\ &+  \medmath{\frac{\Delta_{t}}{c}\sum_{q,q^{'},u^{},k^{}u^{'},k^{'}} \underset{q^{},u,k^{}}{\operatorname{SNR}} \underset{q^{'},u^{'},k^{'}}{\operatorname{SNR}} } \medmath{ (k)   \bm{\Delta}_{q^{}u^{},k^{}}   \nabla_{\bm{\Phi}_{U}}^{\mathrm{T}} \tau_{q^{'}u^{'},k^{'}}} \medmath{ \omega_{qU,k} \omega_{q^{'}U,k^{'}}}    \medmath{   \left(\sum_{qu,k} \underset{qu,k}{\operatorname{SNR}} \omega_{qU,k}\right)^{\mathrm{-1}}} 
\end{split}
\end{align}
\end{figure*}
\end{lemma}
\begin{proof}
See Appendix \ref{Appendix_lemma_information_loss_FIM_3D_velocity_3D_orientation}.
\end{proof}

\begin{lemma}
\label{lemma:information_loss_FIM_3D_velocity_3D_bth_position_uncertainty}
The loss of information about the FIM of the $\bm{v}_{U,0}$ and $\check{\bm{p}}_{b,0}$  due to uncertainty in the nuisance parameters ${\bm{\kappa}_{2}}$ is given by (\ref{equ_lemma:information_loss_3D_velocity_3D_bth_position_uncertainty}).

\begin{figure*}
\begin{align}
\begin{split}
\label{equ_lemma:information_loss_3D_velocity_3D_bth_position_uncertainty}
\medmath{\bm{G}_{{{y} }}(\bm{y}_{}| \bm{\eta} ;\bm{v}_{U,0},\check{\bm{p}}_{b,0})} &= \medmath{\frac{-\Delta_{t}}{c^2} \sum_{k^{},u^{}k^{'},u^{'}} (k) \underset{bu^{},k^{}}{\operatorname{SNR}} \underset{bu^{'},k^{'}}{\operatorname{SNR}} } \medmath{   \bm{\Delta}_{bu^{},k^{}}  {\bm{\Delta}_{bu^{'},k^{'}}^{\mathrm{T}}} \omega_{bU,k} \omega_{bU,k^{'}}}    \medmath{   \left(\sum_{u,k} \underset{bu,k}{\operatorname{SNR}} \omega_{bU,k}\right)^{\mathrm{-1}}} \\ &- \frac{1}{c}\medmath{\sum_{k^{},u^{}k^{'},u^{'}}    \underset{bu^{},k^{}}{\operatorname{SNR}} \underset{bu^{'},k^{'}}{\operatorname{SNR}}  }  \medmath{{ \bm{\Delta}_{b^{}U^{},k^{}} \nabla_{\check{\bm{p}}_{b,0}}^{\mathrm{T}} \nu_{bU,k^{'}}   }  \frac{(f_{c}^{2}) (\alpha_{obu,k^{}}^{2}\alpha_{obu^{'},k^{'}}^{2})}{4}   \left(\sum_{u,k} \frac{\underset{b^{}u,k^{}}{\operatorname{SNR}}  \alpha_{obu,k}^{2}}{2}\right)^{-1}}  
\end{split}
\end{align}
\end{figure*}
\end{lemma}
\begin{proof}
See Appendix \ref{Appendix_lemma_information_loss_FIM_3D_velocity_3D_bth_position_uncertainty}.
\end{proof}

\begin{lemma}
\label{lemma:information_loss_FIM_3D_velocity_3D_bth_velocity_uncertainty}
The loss of information about the FIM of the $\bm{v}_{U,0}$ and $\check{\bm{v}}_{b,0}$  due to uncertainty in the nuisance parameters ${\bm{\kappa}_{2}}$ is given by (\ref{equ_lemma:information_loss_FIM_3D_velocity_3D_bth_velocity_uncertainty}).

\begin{figure*}
\begin{align}
\begin{split}
\label{equ_lemma:information_loss_FIM_3D_velocity_3D_bth_velocity_uncertainty}
\medmath{\bm{G}_{{{y} }}(\bm{y}_{}| \bm{\eta} ;\bm{v}_{U,0},\check{\bm{v}}_{b,0})} &=   -\medmath{\frac{\Delta_{t}^{2}}{c^2} \sum_{k^{},u^{}k^{'},u^{'}} \underset{bu^{},k^{}}{\operatorname{SNR}} \underset{bu^{'},k^{'}}{\operatorname{SNR}} } \medmath{   \bm{\Delta}_{bu^{},k^{}} (k) (k^{'}) {\bm{\Delta}_{bu^{'},k^{'}}^{\mathrm{T}}} \omega_{bU,k} \omega_{bU,k^{'}}}    \medmath{   \left(\sum_{u,k} \underset{bu,k}{\operatorname{SNR}} \omega_{bU,k}\right)^{\mathrm{-1}}} \\ -& \frac{1}{c^2}\medmath{  \sum_{k^{},u^{}k^{'},u^{'}}}  \medmath{  \underset{bu^{},k^{}}{\operatorname{SNR}} \underset{bu^{'},k^{'}}{\operatorname{SNR}}  }  \medmath{{ \bm{\Delta}_{bU,k^{}} \bm{\Delta}_{bU,k^{'}}^{\mathrm{T}}   }  \frac{(f_{c}^{2}) (\alpha_{obu,k^{}}^{2}\alpha_{obu^{'},k^{'}}^{2})}{4}   \left(\sum_{u,k} \frac{\underset{b^{}u,k^{}}{\operatorname{SNR}}  \alpha_{obu,k}^{2}}{2}\right)^{-1}}  
\end{split}
\end{align}
\end{figure*}
\end{lemma}
\begin{proof}
See Appendix \ref{Appendix_lemma_information_loss_FIM_3D_velocity_3D_bth_velocity_uncertainty}.
\end{proof}

\begin{lemma}
\label{lemma:information_loss_FIM_3D_orientation_3D_orientation}
The loss of information about the FIM of the $3$D orientation of the receiver due to uncertainty in the nuisance parameters ${\bm{\kappa}_{2}}$ is given by (\ref{equ_lemma:information_loss_FIM_3D_orientation_3D_orientation}).

\begin{figure*}
\begin{align}
\begin{split}
\label{equ_lemma:information_loss_FIM_3D_orientation_3D_orientation}
&\medmath{\bm{G}_{{{y} }}(\bm{y}_{}| \bm{\eta} ;\bm{\Phi}_{U},\bm{\Phi}_{U}) = 
\medmath{\sum_{b}   \norm{\sum_{k^{},u^{}} \underset{bu^{},k^{}}{\operatorname{SNR}} \; \omega_{bU,k} \nabla_{\bm{\Phi}_{U}}^{\mathrm{T}} \tau_{bu^{},k^{}}   }}^{2}}    \medmath{   \left(\sum_{u,k} \underset{bu,k}{\operatorname{SNR}} \; \omega_{bU,k}\right)^{\mathrm{-1}}  }  + 
\medmath{  \norm{\sum_{q^{},k^{},u^{}} \underset{qu^{},k^{}}{\operatorname{SNR}} \; \omega_{qU,k} \nabla_{\bm{\Phi}_{U}}^{\mathrm{T}} \tau_{qu^{},k^{}}   }}^{2}    \medmath{   \left(\sum_{q,u,k} \underset{qu,k}{\operatorname{SNR}} \; \omega_{qu,k}\right)^{\mathrm{-1}}  } 
\end{split}
\end{align}
\end{figure*}
\end{lemma}
\begin{proof}
See Appendix \ref{Appendix_lemma_information_loss_FIM_3D_orientation}.
\end{proof}

\begin{lemma}
\label{lemma:information_loss_3D_orientation_3D_bth_position_uncertainty}
The loss of information about the FIM of the $\bm{\Phi}_{U,0}$ and $\check{\bm{p}}_{b,0}$  due to uncertainty in the nuisance parameters ${\bm{\kappa}_{2}}$ is given by (\ref{equ_lemma:information_loss_3D_orientation_3D_bth_position_uncertainty}).

\begin{figure*}
\begin{align}
\begin{split}
\label{equ_lemma:information_loss_3D_orientation_3D_bth_position_uncertainty}
\medmath{\bm{G}_{{{y} }}(\bm{y}_{}| \bm{\eta} ;\bm{\Phi}_{U},\check{\bm{p}}_{b,0})} &= \medmath{\frac{-1}{c} \sum_{k^{},u^{}k^{'},u^{'}} \underset{bu^{},k^{}}{\operatorname{SNR}} \underset{bu^{'},k^{'}}{\operatorname{SNR}} } \medmath{   \nabla_{\bm{\Phi}_{U}} \tau_{bu^{},k^{}}  {\bm{\Delta}_{bu^{'},k^{'}}^{\mathrm{T}}} \omega_{bU,k} \omega_{bU,k^{'}}}    \medmath{   \left(\sum_{u,k} \underset{bu,k}{\operatorname{SNR}} \omega_{bU,k}\right)^{\mathrm{-1}}}
\end{split}
\end{align}
\end{figure*}
\end{lemma}
\begin{proof}
See Appendix \ref{Appendix_lemma_information_loss_FIM_3D_orientation_3D_bth_position_uncertainty}.
\end{proof}

\begin{lemma}
\label{lemma:information_loss_FIM_3D_orientation_3D_bth_velocity_uncertainty}
The loss of information about the FIM of the $\bm{\Phi}_{U,0}$ and $\check{\bm{v}}_{b,0}$  due to uncertainty in the nuisance parameters ${\bm{\kappa}_{2}}$ is given by (\ref{equ_lemma:information_loss_FIM_3D_orientation_3D_bth_velocity_uncertainty}).

\begin{figure*}
\begin{align}
\begin{split}
\label{equ_lemma:information_loss_FIM_3D_orientation_3D_bth_velocity_uncertainty}
\medmath{\bm{G}_{{{y} }}(\bm{y}_{}| \bm{\eta} ;\bm{\Phi}_{U},\check{\bm{v}}_{b,0})} &= \medmath{-\frac{\Delta_{t}}{c} \sum_{k^{},u^{}k^{'},u^{'}} \underset{bu^{},k^{}}{\operatorname{SNR}} \underset{bu^{'},k^{'}}{\operatorname{SNR}} } \medmath{   \nabla_{\bm{\Phi}_{U}} \tau_{bu^{},k^{}}  (k^{'}) {\bm{\Delta}_{bu^{'},k^{'}}^{\mathrm{T}}} \omega_{bU,k} \omega_{bU,k^{'}}}    \medmath{   \left(\sum_{u,k} \underset{bu,k}{\operatorname{SNR}} \omega_{bU,k}\right)^{\mathrm{-1}}}  
\end{split}
\end{align}
\end{figure*}
\end{lemma}
\begin{proof}
See Appendix \ref{Appendix_lemma_information_loss_FIM_3D_orientation_3D_bth_velocity_uncertainty}.
\end{proof}

\begin{lemma}
\label{lemma:information_loss_FIM_3D_bth_position_3D_bth_position}
The loss of information about the FIM of the $\check{\bm{p}}_{b,0}$  due to uncertainty in the nuisance parameters ${\bm{\kappa}_{2}}$ is given by (\ref{equ_lemma:information_loss_3D_position_3D_bth_position_uncertainty_1}).

\begin{figure*}
\begin{align}
\begin{split}
\label{equ_lemma:information_loss_3D_position_3D_bth_position_uncertainty_1}
\medmath{\bm{G}_{{{y} }}(\bm{y}_{}| \bm{\eta} ;\check{\bm{p}}_{b,0},\check{\bm{p}}_{b,0})} =& \medmath{   \norm{\sum_{k^{},u^{}} \underset{bu^{},k^{}}{\operatorname{SNR}}\bm{\Delta}_{bu^{},k^{}}^{\mathrm{T}} \frac{ \omega_{bU,k}}{c} }}^{2}    \medmath{   \left(\sum_{u,k} \underset{bu,k}{\operatorname{SNR}} \omega_{bU,k}\right)^{\mathrm{-1}}  }    +
\medmath{   \norm{\sum_{k^{},q^{}} \underset{bq^{},k^{}}{\operatorname{SNR}}\bm{\Delta}_{bq^{},k^{}}^{\mathrm{T}} \frac{ \omega_{bq,k}}{c} }}^{2}    \medmath{   \left(\sum_{q,k} \underset{bu,k}{\operatorname{SNR}} \omega_{bq,k}\right)^{\mathrm{-1}}  } \\&+  
\medmath{\norm{{\sum_{u^{},k^{} } \underset{bu^{},k^{}}{\operatorname{SNR}} \; \; \nabla_{\check{\bm{p}}_{b,0}}^{\mathrm{T}} \nu_{bU,k^{}}   }  \frac{(f_{c}^{}) (\alpha_{obu,k^{}}^{2})}{2} }^{2}  \left(\sum_{u,k} \frac{\underset{bu,k}{\operatorname{SNR}}  \alpha_{obu,k}^{2}}{2}\right)^{-1}}  + \medmath{\norm{{\sum_{q^{},k^{} } \underset{bq^{},k^{}}{\operatorname{SNR}} \; \; \nabla_{\check{\bm{p}}_{b,0}}^{\mathrm{T}} \nu_{bq,k^{}}   }  \frac{(f_{c}^{}) (\alpha_{obq,k^{}}^{2})}{2} }^{2}  \left(\sum_{q,k} \frac{\underset{bu,k}{\operatorname{SNR}}  \alpha_{obq,k}^{2}}{2}\right)^{-1}}
\end{split}
\end{align}
\end{figure*}
\end{lemma}
\begin{proof}
See Appendix \ref{Appendix_lemma_information_loss_FIM_3D_bth_position_3D_bth_position}.
\end{proof}

\begin{lemma}
\label{lemma:information_loss_FIM_3D_bth_position_3D_bth_velocity_uncertainty}
The loss of information about the FIM of $\check{\bm{p}}_{b,0}$ and $\check{\bm{v}}_{b,0}$  due to uncertainty in the nuisance parameters ${\bm{\kappa}_{2}}$ is given by (\ref{equ_lemma:information_loss_FIM_3D_bth_position_3D_bth_velocity_uncertainty}).

\begin{figure*}
\begin{align}
\begin{split}
\label{equ_lemma:information_loss_FIM_3D_bth_position_3D_bth_velocity_uncertainty}
\medmath{\bm{G}_{{{y} }}(\bm{y}_{}| \bm{\eta} ;\check{\bm{p}}_{b,0},\check{\bm{v}}_{b,0})} &= \medmath{ \frac{\Delta_{t}}{c^2} \sum_{b,k^{},u^{}k^{'},u^{'}} \underset{bu^{},k^{}}{\operatorname{SNR}} \underset{bu^{'},k^{'}}{\operatorname{SNR}} } \medmath{   \bm{\Delta}_{bu^{},k^{}}  (k^{'}) {\bm{\Delta}_{bu^{'},k^{'}}^{\mathrm{T}}} \omega_{bU,k} \omega_{bU,k^{'}}}    \medmath{   \left(\sum_{u,k} \underset{bu,k}{\operatorname{SNR}} \omega_{bU,k}\right)^{\mathrm{-1}}}  \\ & +  \medmath{ \frac{\Delta_{t}}{c^2} \sum_{b,k^{},q^{}k^{'},q^{'}} \underset{bq^{},k^{}}{\operatorname{SNR}} \underset{bq^{'},k^{'}}{\operatorname{SNR}} } \medmath{   \bm{\Delta}_{bq^{},k^{}}  (k^{'}) {\bm{\Delta}_{bq^{'},k^{'}}^{\mathrm{T}}} \omega_{bq,k} \omega_{bq^{'},k^{'}}}    \medmath{   \left(\sum_{q,k} \underset{bq,k}{\operatorname{SNR}} \omega_{bq,k}\right)^{\mathrm{-1}}} \\ &+ \medmath{  \sum_{b,k^{},u^{}k^{'},u^{'}} \underset{bu^{},k^{}}{\operatorname{SNR}} \underset{bu^{'},k^{'}}{\operatorname{SNR}}  }  \medmath{{ \nabla_{\check{\bm{p}}_{b,0}} \nu_{bU,k^{}} \bm{\Delta}_{bU,k^{'}}^{\mathrm{T}}   }  \frac{(f_{c}^{2}) (\alpha_{obu,k^{}}^{2}\alpha_{obu^{'},k^{'}}^{2})}{4c}   \left(\sum_{u,k} \frac{\underset{b^{}u,k^{}}{\operatorname{SNR}}  \alpha_{obu,k}^{2}}{2}\right)^{-1}} \\ & + \medmath{  \sum_{b,k^{},q^{}k^{'},q^{'}} \underset{bq^{},k^{}}{\operatorname{SNR}} \underset{bq^{'},k^{'}}{\operatorname{SNR}}  }  \medmath{{ \nabla_{\check{\bm{p}}_{b,0}} \nu_{bq,k^{}} \bm{\Delta}_{bq^{'},k^{'}}^{\mathrm{T}}   }  \frac{(f_{c}^{2}) (\alpha_{obu,k^{}}^{2}\alpha_{obq^{'},k^{'}}^{2})}{4 c}   \left(\sum_{q,k} \frac{\underset{b^{}q,k^{}}{\operatorname{SNR}}  \alpha_{obq,k}^{2}}{2}\right)^{-1}} 
\end{split}
\end{align}
\end{figure*}
\end{lemma}
\begin{proof}
See Appendix \ref{Appendix_lemma_information_loss_FIM_3D_bth_position_3D_bth_veloctiy}.
\end{proof}

\begin{lemma}
\label{lemma:information_loss_FIM_3D_bth_velocity_3D_bth_velocity}
The loss of information about the FIM of $\check{\bm{v}}_{b,0}$ due to uncertainty in the nuisance parameters ${\bm{\kappa}_{2}}$ is given by (\ref{equ_lemma:information_loss_FIM_3D_bth_velocity_3D_bth_velocity_uncertainty}).

\begin{figure*}
\begin{align}
\begin{split}
\label{equ_lemma:information_loss_FIM_3D_bth_velocity_3D_bth_velocity_uncertainty}
\medmath{\bm{G}_{{{y} }}(\bm{y}_{}| \bm{\eta} ;\check{\bm{v}}_{b,0},\check{\bm{v}}_{b,0})} =& \medmath{   \norm{\sum_{k^{},u^{}} \underset{bu^{},k^{}}{\operatorname{SNR}}\bm{\Delta}_{bu^{},k^{}}^{\mathrm{T}} \frac{ (k) \Delta_{t}\omega_{bU,k}}{c} }}^{2}    \medmath{   \left(\sum_{u,k} \underset{bu,k}{\operatorname{SNR}} \omega_{bU,k}\right)^{\mathrm{-1}}  }  +
 \medmath{   \norm{\sum_{k^{},q^{}} \underset{bq^{},k^{}}{\operatorname{SNR}}\bm{\Delta}_{bq^{},k^{}}^{\mathrm{T}} \frac{ (k) \Delta_{t}\omega_{bq,k}}{c} }}^{2}    \medmath{   \left(\sum_{q,k} \underset{bq,k}{\operatorname{SNR}} \omega_{bq,k}\right)^{\mathrm{-1}}  }\\& +  
\medmath{\norm{{\sum_{u^{},k^{} } \underset{bu^{},k^{}}{\operatorname{SNR}} \; \; \bm{\Delta}_{bU^{},k^{}}^{\mathrm{T}}   }  \frac{(f_{c}^{}) (\alpha_{obu,k^{}}^{2})}{2c} }^{2}  \left(\sum_{u,k} \frac{\underset{bu,k}{\operatorname{SNR}}  \alpha_{obu,k}^{2}}{2}\right)^{-1}} 
+ \medmath{\norm{{\sum_{q^{},k^{} } \underset{bq^{},k^{}}{\operatorname{SNR}} \; \; \bm{\Delta}_{bq^{},k^{}}^{\mathrm{T}}   }  \frac{(f_{c}^{}) (\alpha_{obq,k^{}}^{2})}{2c} }^{2}  \left(\sum_{q,k} \frac{\underset{bq,k}{\operatorname{SNR}}  \alpha_{obq,k}^{2}}{2}\right)^{-1}} 
\end{split}
\end{align}
\end{figure*}
\end{lemma}
\begin{proof}
See Appendix \ref{Appendix_lemma_information_loss_FIM_3D_bth_velocity_3D_bth_velocity}.
\end{proof}
The EFIM for the parameters of interest represented by $\mathbf{J}_{\bm{y}|\bm{\kappa_{1}}}^{\mathrm{e}}$ can be obtained by combining the FIM for the parameters of interest, $\mathbf{J}_{\bm{y}|\bm{\kappa_{1}}}$  and the corresponding loss of information represented by $\mathbf{J}_{ \bm{\bm{y}}; \bm{\kappa}_1}^{nu}$. In other words, the entries in $\mathbf{J}_{\bm{y}|\bm{\kappa_{1}}}^{\mathrm{e}}$ can be obtained by appropriately combining Lemmas \ref{lemma:FIM_3D_position} - \ref{lemma:FIM_3D_b_th_velocity_offset_3D_b_th_velocity_offset} with Lemmas \ref{lemma:information_loss_FIM_3D_position} - \ref{lemma:information_loss_FIM_3D_bth_velocity_3D_bth_velocity}. In the next section, simulations will be used to determine the most efficient combinations of $N_B$, $N_K$, $N_Q$, and $N_U$ that produce a positive definite $\mathbf{J}_{\bm{y}|\bm{\kappa_{1}}}^{\mathrm{e}}$. This informs the efficient combinations of $N_B$, $N_K$, $N_Q$, and $N_U$ that allow for $9$D localization and LEO position and velocity estimation.

\section{Numerical Results}
In this section, we use simulations to determine the minimal combinations of $N_B$, $N_K$, $N_Q$, and $N_U$ that produce a positive definite $\mathbf{J}_{\bm{y}|\bm{\kappa_{1}}}^{\mathrm{e}}$. It is important to note that while \cite{Fundamentals_of_LEO_Based_Localization} indicates that using $3$ LEO satellites, $3$ time slots, and $N_U > 1$ is enables the $9$D localization of a receiver, the presence of uncertainty in the LEO ephemeris changes the analysis and the corresponding conclusions. Moreover, the presence of signals from $5$G base stations adds another dimension. Hence, in this section, we investigate the use of signals in the LEO-receiver, LEO-BS, and BS-receiver links for both $9$D localization and ephemeris correction. This investigation is carried out by analyzing the conditions that make $\mathbf{J}_{\bm{y}|\bm{\kappa_{1}}}^{\mathrm{e}}$  positive definite. We notice that $\mathbf{J}_{\bm{y}|\bm{\kappa_{1}}}^{\mathrm{e}}$ is positive definite when $N_B = 1$, if $N_K \geq 3$, $N_U > 1$, and $N_Q \geq 3$. Again, the matrix $\mathbf{J}_{\bm{y}|\bm{\kappa_{1}}}^{\mathrm{e}}$ is positive definite when $N_B = 2$, if $N_K \geq 3$, $N_U > 1$, and $N_Q \geq 3$. Finally, the matrix $\mathbf{J}_{\bm{y}|\bm{\kappa_{1}}}^{\mathrm{e}}$ is positive definite when $N_B = 3$, if $N_K \geq 4$, $N_U > 1$, and $N_Q \geq 3$. These conditions for joint $9$D localization and ephemeris correction are obtained using the following simulation parameters. The 
following frequencies are considered $f_c \in [10,27,40,60] \text{ GHz}$. The following spacings between transmission time slots are considered $\Delta_t \in [25,  50, 100,  1000, 10000,20000,50000] \text{ ms}.$  We consider the following number of LEOs and BSs $N_B \in [1,2,3]$ and $N_Q \in [1,2,3,4]$, respectively. The $3$D coordinates of the LEOs are randomly chosen, but the LEOs are approximately $2000 \text{ km}$ from the receiver. The $3$D coordinates of the receiver and the BSs are also randomly chosen, but their distances are $30 \text{ m}$ and $100 \text{ m}$ from the origin, respectively. The BSs are stationary. However, the $3$D velocity of the LEOs and receiver are randomly chosen, but their speeds are $8000 \text{ m/s}$ and $25 \text{ m/s}$, respectively. The velocity of the LEOs is modeled to change every time slot to capture the acceleration of the LEOs. However, the velocity of the receiver remains constant across all transmission time slots. For all links, the effective baseband bandwidth is $100 \text{ MHz}$, and the BCC is $0 \text{ MHz}$. For the $b^{\text{th}}$ LEO and $q^{\text{th}}$ BS, we assume that the same signal is transmitted across all $N_{K}$ time slots, and the channel gain is constant across all receive antennas and time slots. Hence, the useful SNRs are
$$
\underset{b}{\operatorname{SNR}} \triangleq \frac{8 \pi^2 \left|\beta_{b}\right|^2}{N_{01}} \int_{-\infty}^{\infty}\left|S_{b}[f]\right|^2 d f,
$$
$$
\underset{bq}{\operatorname{SNR}} \triangleq \frac{8 \pi^2 \left|\beta_{bq}\right|^2}{N_{02}} \int_{-\infty}^{\infty}\left|S_{b}[f]\right|^2 d f,
$$
and
$$
\underset{q}{\operatorname{SNR}} \triangleq \frac{8 \pi^2 \left|\beta_{q}\right|^2}{N_{01}} \int_{-\infty}^{\infty}\left|S_{q,k}[f]\right|^2 d f.
$$
The CRLBs for ${\mathbf{p}}_{U,0}$, ${\mathbf{v}}_{U,0}$, ${\mathbf{\Phi}}_{U}$, $\check{\mathbf{p}}_{b,0}$, and  $\check{\mathbf{v}}_{b,0}$ are obtained by inverting $\mathbf{J}_{\bm{y}|\bm{\kappa_{1}}}^{\mathrm{e}}$ and summing the appropriate diagonals. 
\subsection{Observations related to ${\mathbf{p}}_{U,0}$ and ${\mathbf{v}}_{U,0}$. }
The following observations are obtained by examining Figs \ref{Results:PU_NB_1_3_NQ_3_N_K_3_N_U_64_delta_t_index_1_fcIndex_3_SNRIndex_4} - \ref{Results:VU_NB_1_3_NQ_3_N_K_3_N_U_64_delta_t_index_1_fcIndex_3_SNRIndex_4_1}.
\begin{itemize}
    \item The CRLBs decrease with $N_U$.
    \item The CRLBs are more affected by the spacing between the transmission time slots than by any other parameter.
    \item The CRLBs reduce with increasing center frequency.
    \item With $N_B = 1$ and $N_K = 3$, a receiver positioning error of $0.1 \text{ cm}$ is achievable with $N_U = 4$,  $f_c = 40 \text{ GHz}$, SNR of $20 \text{ dB}$ which is constant across all links,  $N_Q = 3$, and $\Delta_t = 1 \text{ s}.$ 
    \item With $N_B = 3$ and $N_K = 4$, a receiver positioning error on the order of $ \text{mm}$ is achievable with $N_U = 4$,  $f_c = 40 \text{ GHz}$, SNR of $20 \text{ dB}$ which is constant across all links,  $N_Q = 3$, and $\Delta_t = 1 \text{ s}.$ 
        \item With $N_B = 1$ and $N_K = 3$, a receiver velocity estimation error on the order of $\text{mm/s}$ is achievable with $N_U = 4$,  $f_c = 40 \text{ GHz}$, SNR of $20 \text{ dB}$ which is constant across all links,  $N_Q = 3$, and $\Delta_t = 1 \text{ s}.$ 
        \item With $N_B = 3$ and $N_K = 4$, a receiver velocity estimation error on the order of $ \text{mm/s}$ is achievable with $N_U = 4$,  $f_c = 40 \text{ GHz}$, SNR of $20 \text{ dB}$ which is constant across all links,  $N_Q = 3$, and $\Delta_t = 1 \text{ s}.$ 
\end{itemize}

\subsection{Observations related to ${\mathbf{\Phi}}_{U}$}
The following observations are obtained by examining Figs \ref{Results:PHIU_NB_1_3_NQ_3_N_K_3_N_U_64_delta_t_index_1_fcIndex_3_SNRIndex_4} - 
 \ref{Results:PHIU_NB_1_3_NQ_3_N_K_3_N_U_64_delta_t_index_1_fcIndex_3_SNRIndex_4_1}. 
\begin{itemize}
    \item The CRLB decreases with $N_U$. This improvement is more substantial than the decrease in the CRLB concerning ${\mathbf{p}}_{U,0}$ and ${\mathbf{v}}_{U,0}$.
    \item The center frequency has a negligible impact on the CRLB.
           \item With $N_B = 1$ and $N_K = 3$, a receiver orientation estimation error of $10^{-3} \text{ rad}$ is achievable with $N_U = 4$,  $f_c = 40 \text{ GHz}$, SNR of $20 \text{ dB}$ which is constant across all links,  $N_Q = 3$, and $\Delta_t = 1 \text{ s}.$ 
        \item With $N_B = 3$ and $N_K = 4$, a receiver orientation estimation error of $10^{-3} \text{ rad}$ is achievable with $N_U = 4$,  $f_c = 40 \text{ GHz}$, SNR of $20 \text{ dB}$ which is constant across all links,  $N_Q = 3$, and $\Delta_t = 1 \text{ s}.$ 
\end{itemize}

\subsection{Observations related to $\check{\mathbf{p}}_{b,0}$}
The following observations are obtained by examining Figs. \ref{Results:CPb_NB_1_3_NQ_3_N_K_3_N_U_64_delta_t_index_1_fcIndex_3_SNRIndex_4} - \ref{Results:CPb_NB_1_3_NQ_3_N_K_3_N_U_64_delta_t_index_1_fcIndex_3_SNRIndex_4_1}.
\begin{itemize}
    \item The CRLBs decrease with $N_U$.
    \item The CRLBs are improved by the spacing between the transmission time slots than by any other parameter.
    \item The CRLBs reduce with increasing center frequency.
           \item With $N_B = 1$ and $N_K = 3$, a LEO position estimation error of $10^{-2} \text{ m}$ is achievable with $N_U = 4$,  $f_c = 40 \text{ GHz}$, SNR of $20 \text{ dB}$ which is constant across all links,  $N_Q = 3$, and $\Delta_t = 20 \text{ s}.$ 
        \item With $N_B = 3$ and $N_K = 4$, a LEO position estimation error of $10^{-3} \text{ m}$ is achievable with $N_U = 4$,  $f_c = 40 \text{ GHz}$, SNR of $20 \text{ dB}$ which is constant across all links,  $N_Q = 3$, and $\Delta_t = 20 \text{ s}.$ 
\end{itemize}

\subsection{Observations related to $\check{\mathbf{v}}_{b,0}$}

The following observations are obtained by examining Figs. \ref{Results:CVb_NB_1_3_NQ_3_N_K_3_N_U_64_delta_t_index_1_fcIndex_3_SNRIndex_4} - \ref{Results:CVb_NB_1_3_NQ_3_N_K_3_N_U_64_delta_t_index_1_fcIndex_3_SNRIndex_4_1}.\begin{itemize}
    \item The number of antennas does not substantially impact the CRLB.
    \item The CRLB is improved by the spacing between the transmission time slots than by any other parameter.
    \item The CRLB reduces with increasing center frequency.
               \item With $N_B = 1$ and $N_K = 3$, a LEO velocity estimation error of $1 \text{ m/s}$ is achievable with $N_U = 4$,  $f_c = 40 \text{ GHz}$, SNR of $20 \text{ dB}$ which is constant across all links,  $N_Q = 3$, and $\Delta_t = 20 \text{ s}.$ 
        \item With $N_B = 3$ and $N_K = 4$, a LEO velocity estimation error of $1 \text{ m/s}$ is achievable with $N_U = 4$,  $f_c = 40 \text{ GHz}$, SNR of $20 \text{ dB}$ which is constant across all links,  $N_Q = 3$, and $\Delta_t = 20 \text{ s}.$ 
\end{itemize}

\begin{figure}[htb!]
\centering
\subfloat[]{\includegraphics[ width= 3.2in]{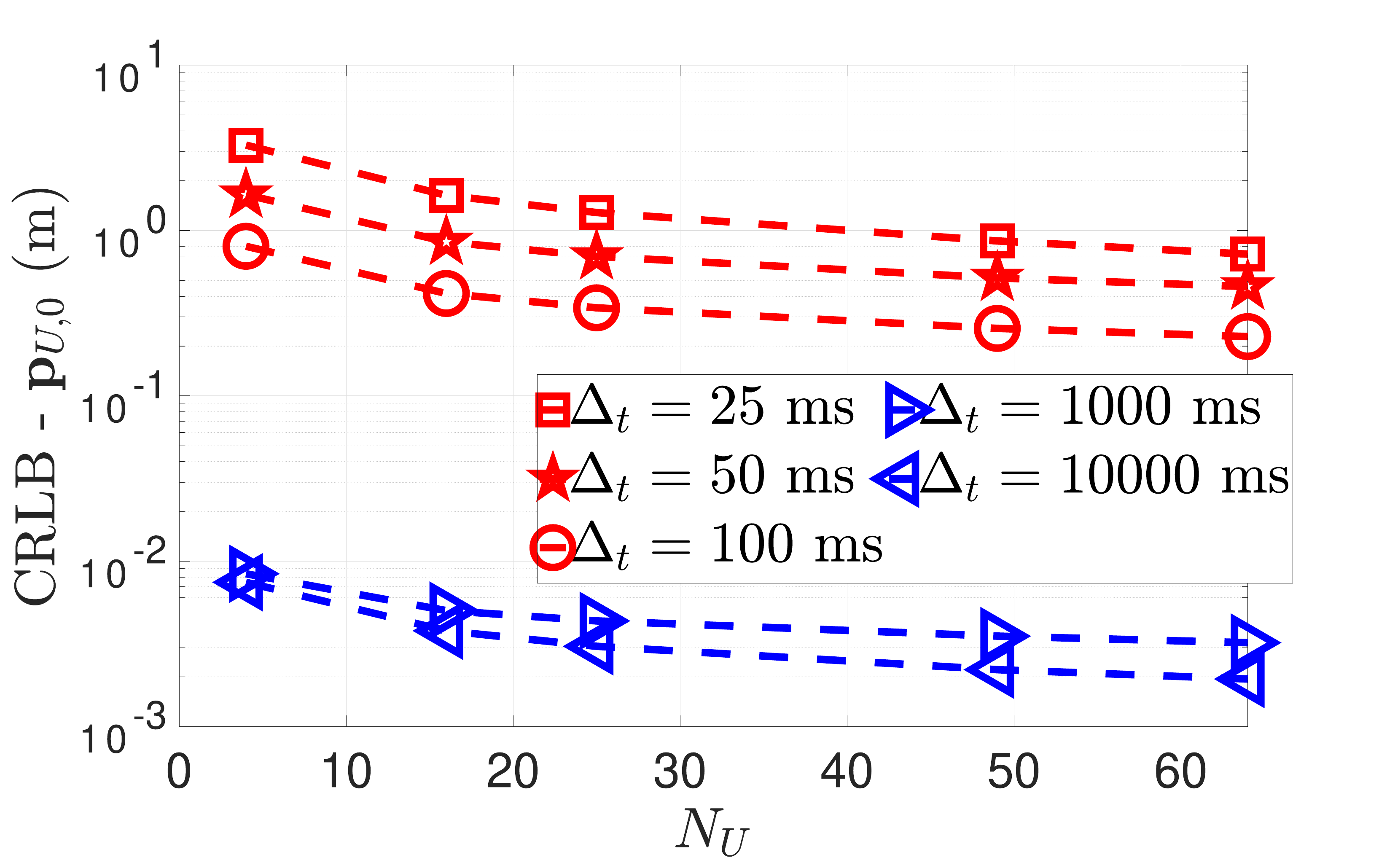}
\label{fig:Results/PU/_NB_1_NQ_3_N_K_3_N_U_64_delta_t_index_1_fcIndex_3_SNRIndex_4}}
\hfil
\subfloat[]{\includegraphics[ width= 3.2in]{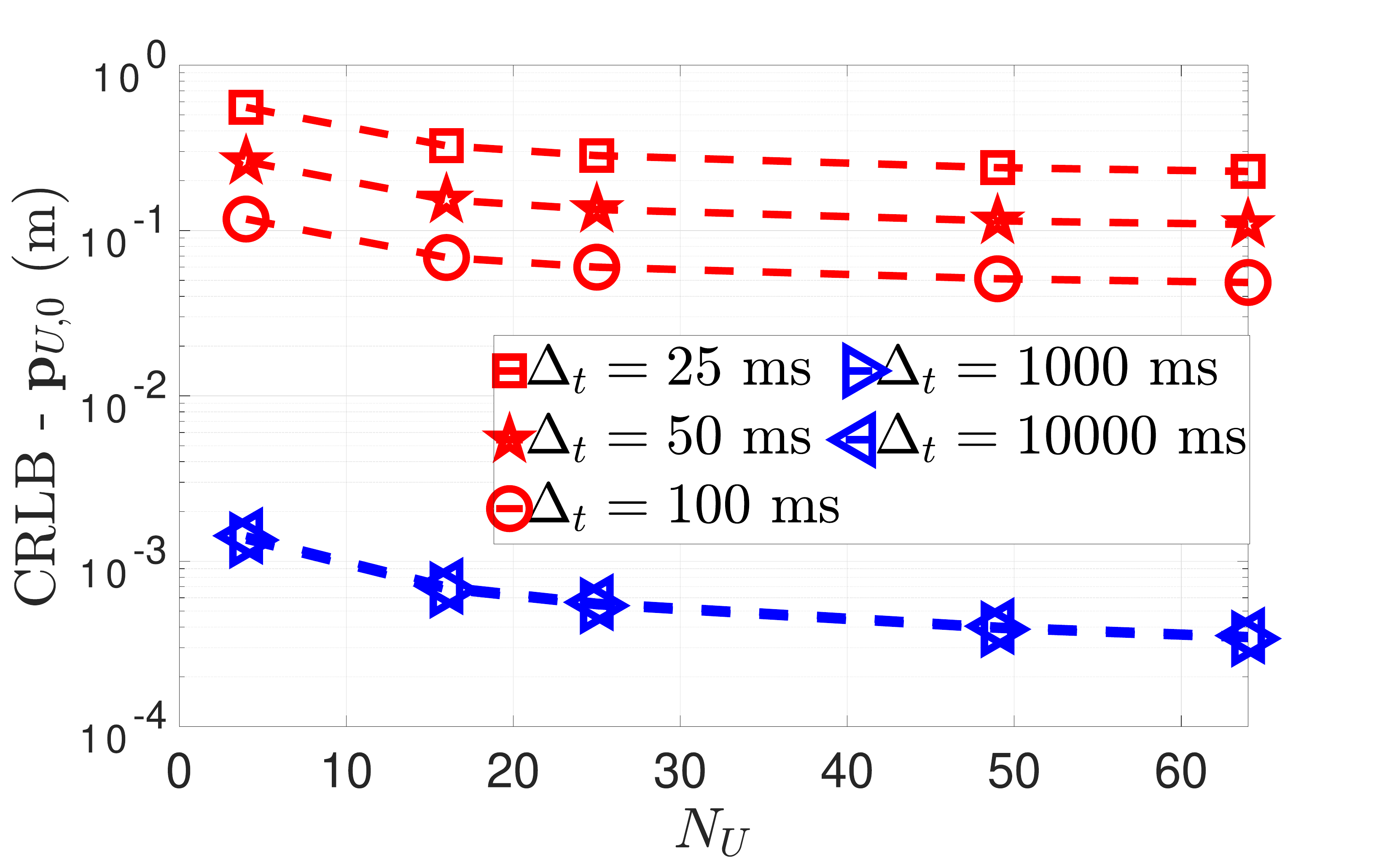}
\label{fig:Results/PU/_NB_3_NQ_3_N_K_4_N_U_64_delta_t_index_1_fcIndex_3_SNRIndex_4}}
\caption{CRLB of ${\mathbf{p}}_{U,0}$ as a function of $N_U$, $f_c = 40 \text{ GHz}$, SNR of $20 \text{ dB}$ which is constant across all links,  and $N_Q = 3$: (a) $N_B = 1$ and $N_K = 3$ and (b) $N_B = 3$ and $N_K = 4$.}
\label{Results:PU_NB_1_3_NQ_3_N_K_3_N_U_64_delta_t_index_1_fcIndex_3_SNRIndex_4}
\end{figure}

\begin{figure}[htb!]
\centering
\subfloat[]{\includegraphics[ width= 3.2in]{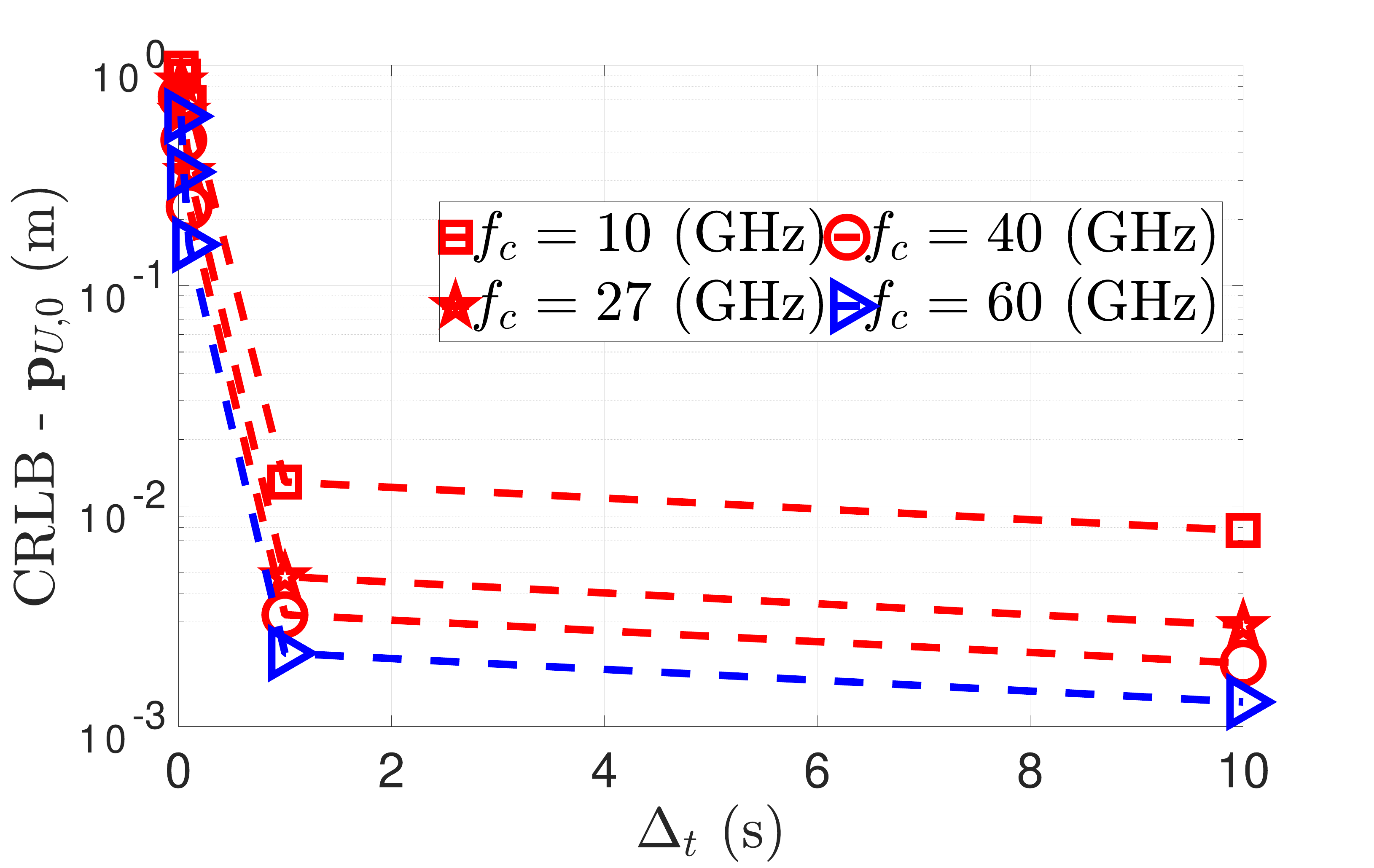}
\label{fig:Results/PU/_NB_1_NQ_3_N_K_3_N_U_64_delta_t_index_1_fcIndex_3_SNRIndex_4_1}}
\hfil
\subfloat[]{\includegraphics[ width= 3.2in]{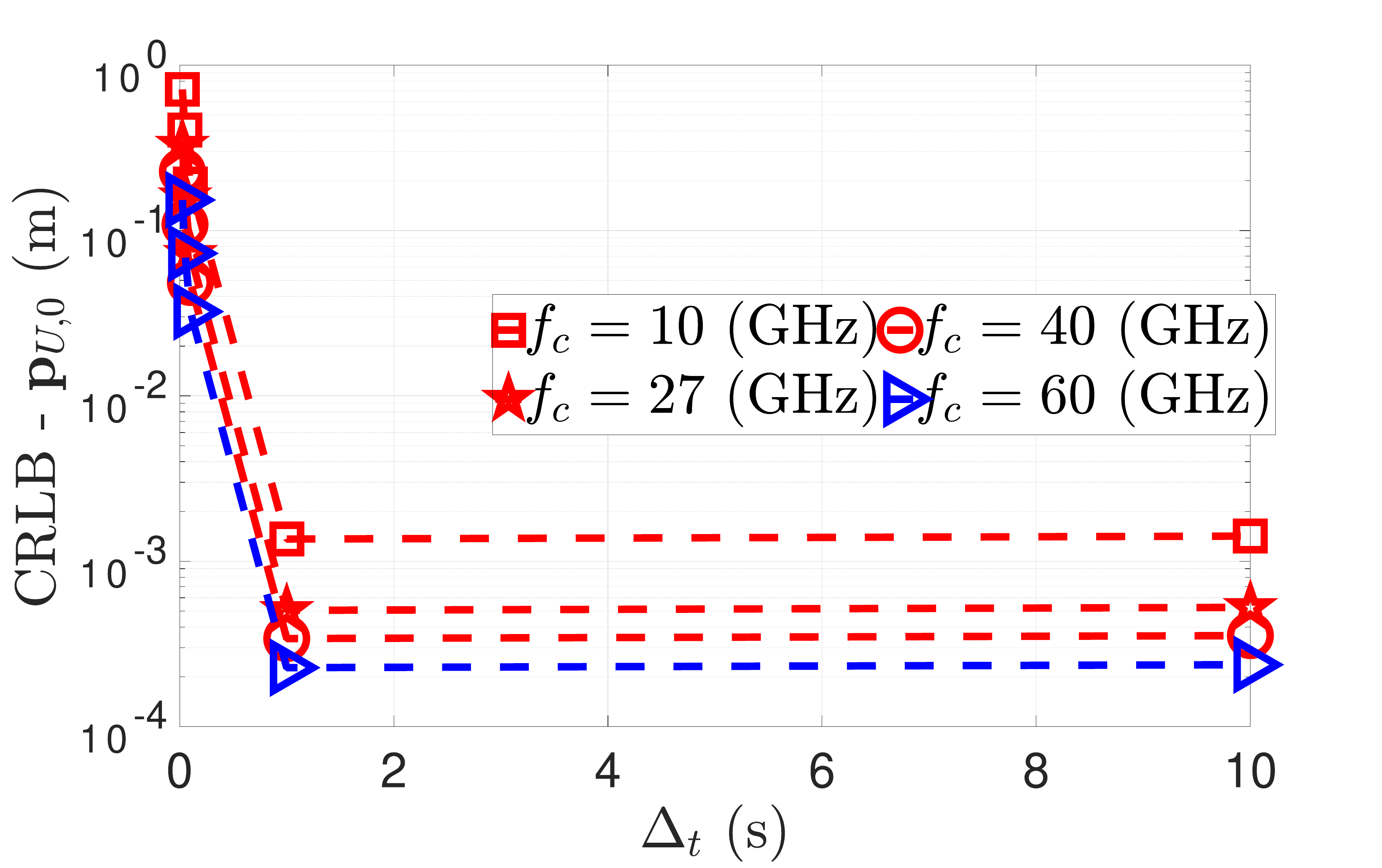}
\label{fig:Results/PU/_NB_3_NQ_3_N_K_4_N_U_64_delta_t_index_1_fcIndex_3_SNRIndex_4_1}}
\caption{CRLB of ${\mathbf{p}}_{U,0}$ as a function of $fc$, $N_U = 64$, SNR of $20 \text{ dB}$ which is constant across all links, and $N_Q = 3$: (a) $N_B = 1$ and $N_K = 3$ and (b) $N_B = 3$ and $N_K = 4$.}
\label{Results:PU_NB_1_3_NQ_3_N_K_3_N_U_64_delta_t_index_1_fcIndex_3_SNRIndex_4_1}
\end{figure}

\begin{figure}[htb!]
\centering
\subfloat[]{\includegraphics[ width= 3.2in]{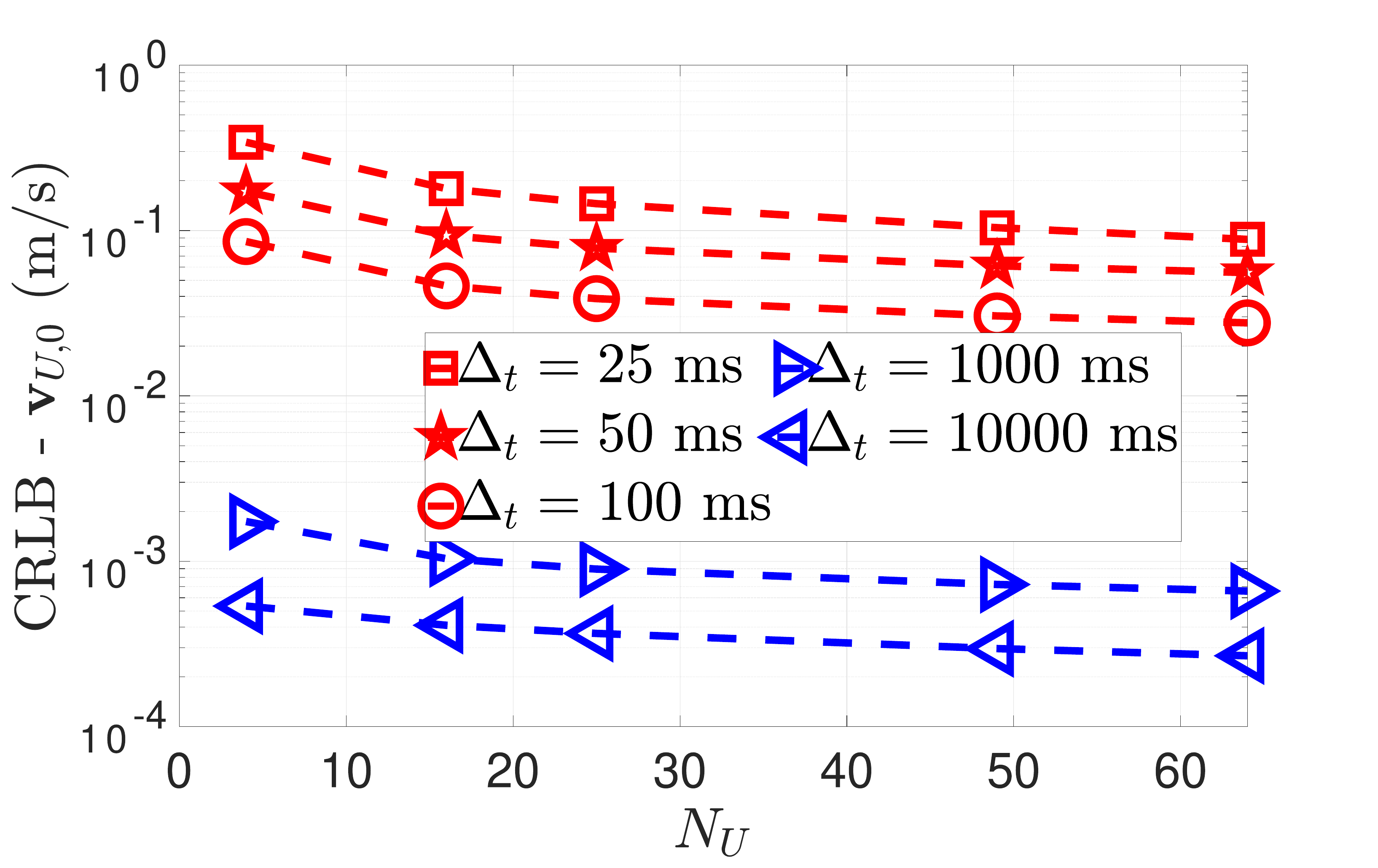}
\label{fig:Results/VU/_NB_1_NQ_3_N_K_3_N_U_64_delta_t_index_1_fcIndex_3_SNRIndex_4}}
\hfil
\subfloat[]{\includegraphics[ width= 3.2in]{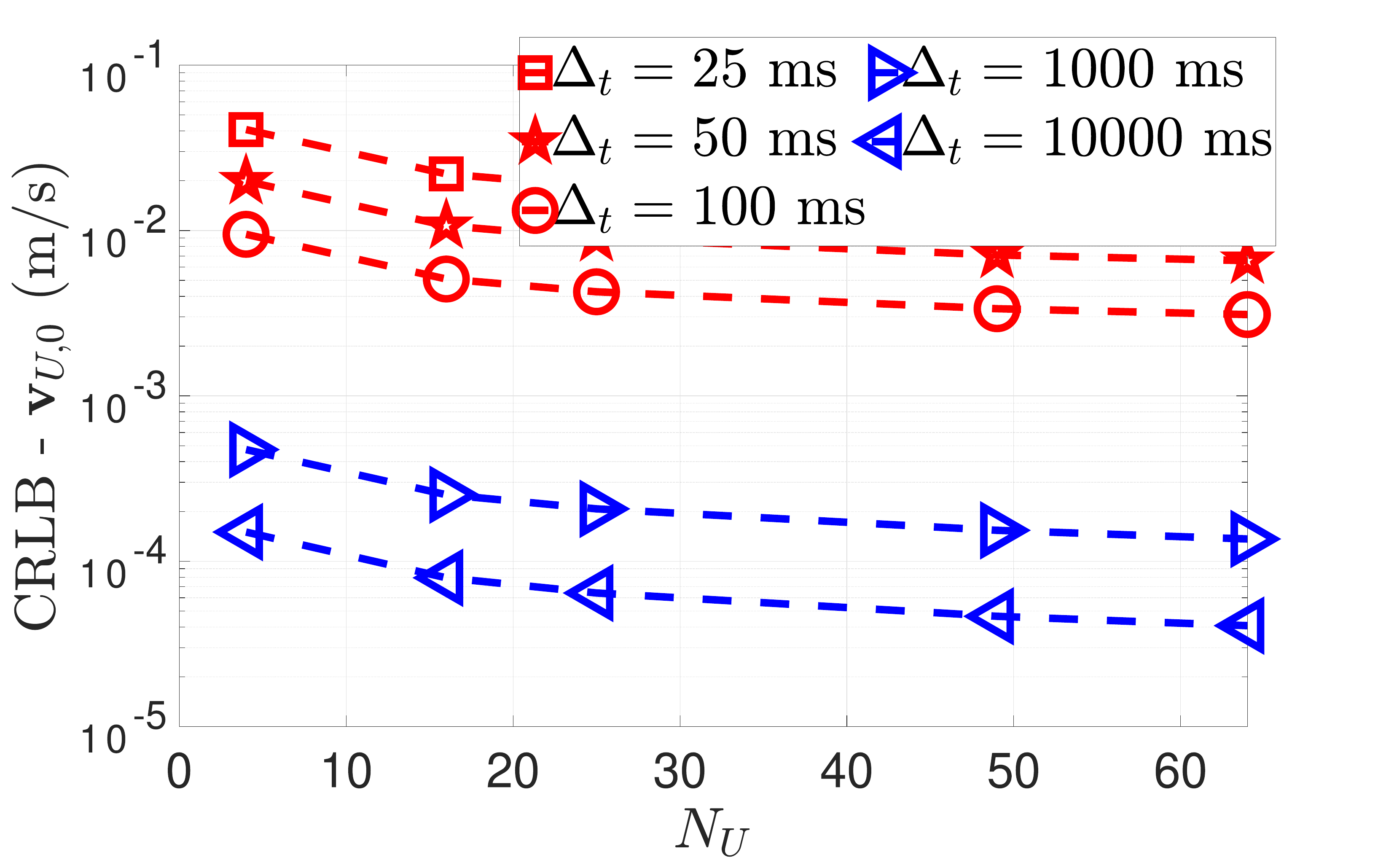}
\label{fig:Results/VU/_NB_3_NQ_3_N_K_4_N_U_64_delta_t_index_1_fcIndex_3_SNRIndex_4}}
\caption{CRLB of ${\mathbf{v}}_{U,0}$ as a function of $N_U$, $f_c = 40 \text{ GHz}$, SNR of $20 \text{ dB}$ which is constant across all links,  and $N_Q = 3$: (a) $N_B = 1$ and $N_K = 3$ and (b) $N_B = 3$ and $N_K = 4$.}
\label{Results:VU_NB_1_3_NQ_3_N_K_3_N_U_64_delta_t_index_1_fcIndex_3_SNRIndex_4}
\end{figure}

\begin{figure}[htb!]
\centering
\subfloat[]{\includegraphics[ width= 3.2in]{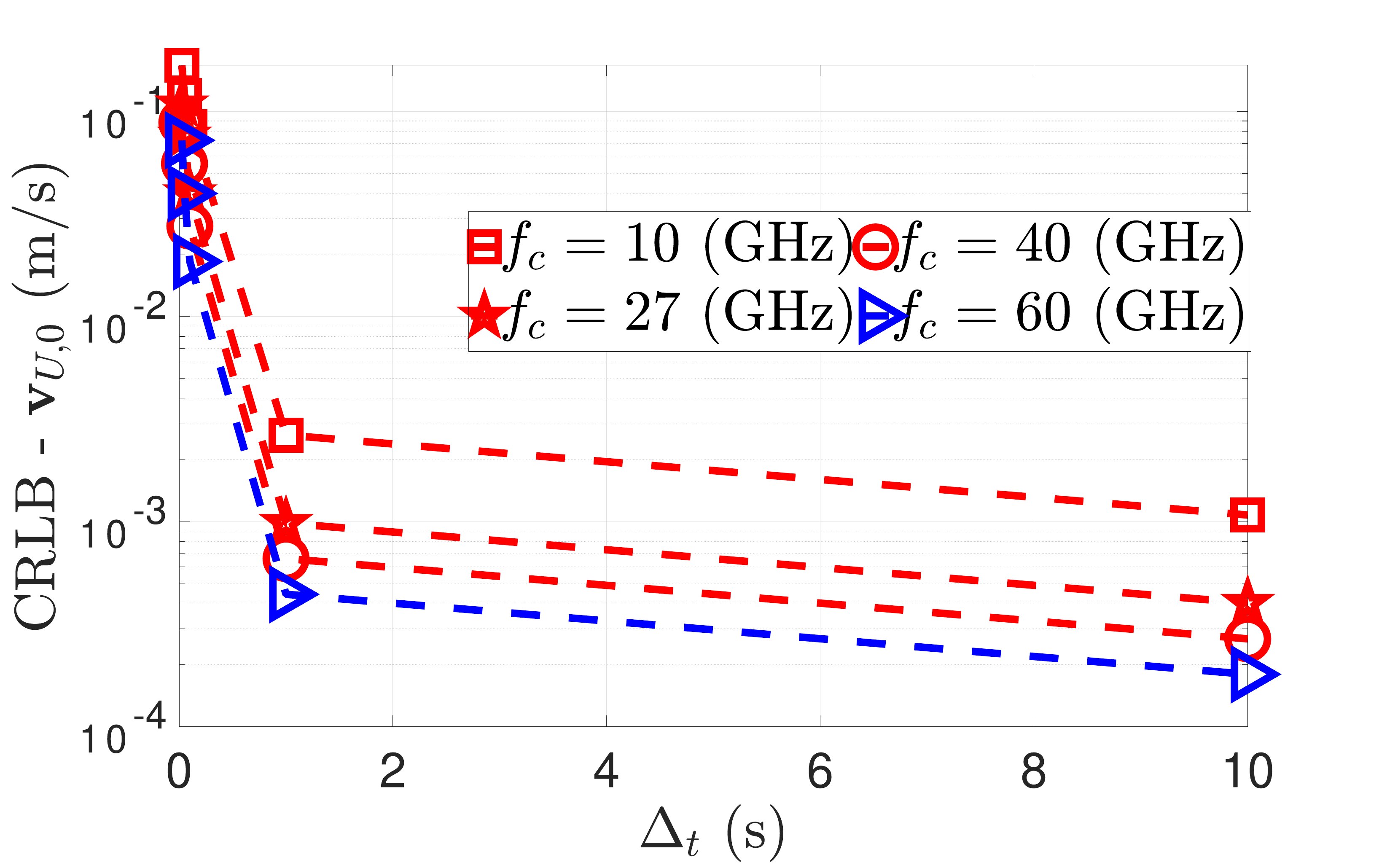}
\label{fig:Results/VU/_NB_1_NQ_3_N_K_3_N_U_64_delta_t_index_1_fcIndex_3_SNRIndex_4_1}}
\hfil
\subfloat[]{\includegraphics[ width= 3.2in]{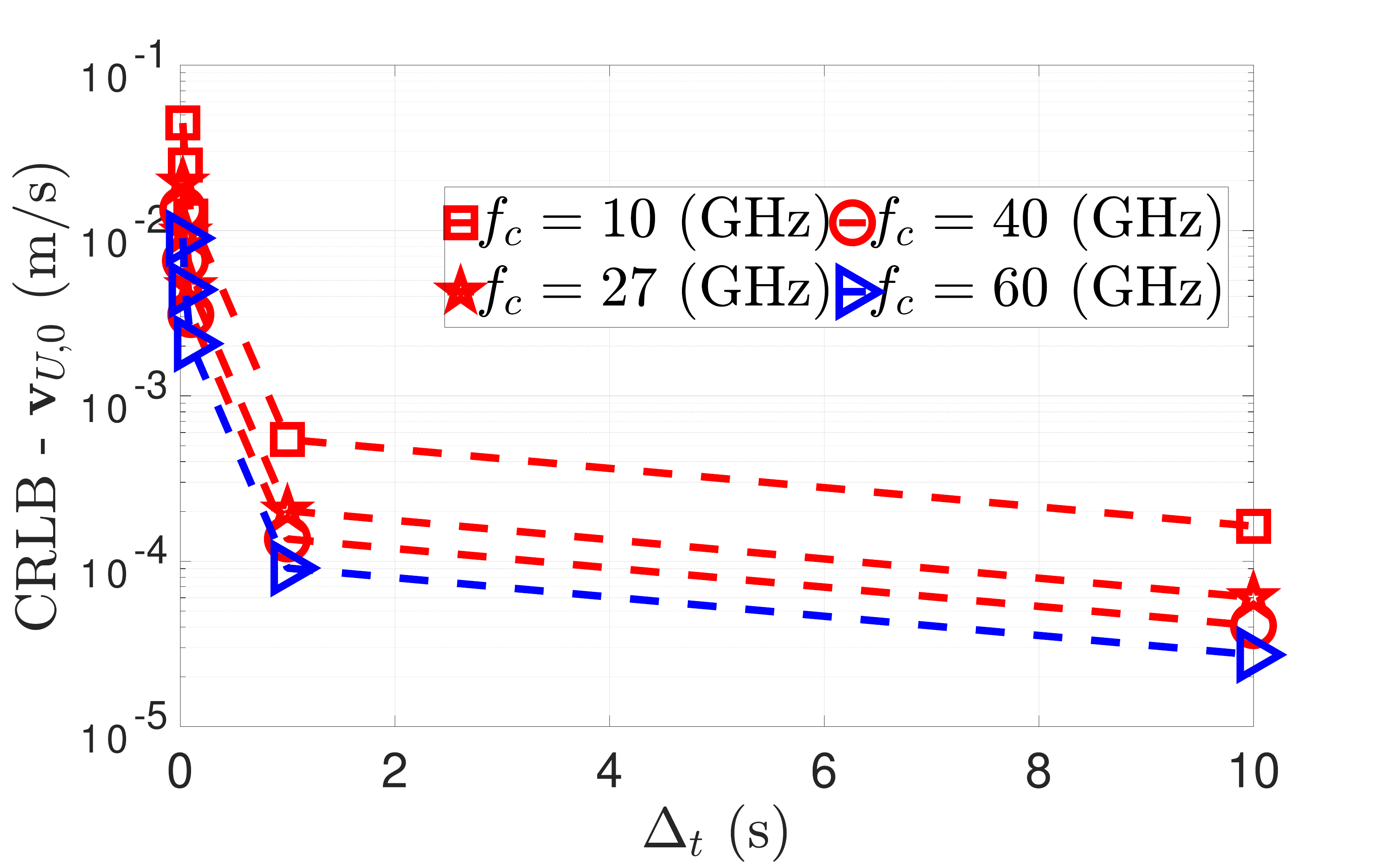}
\label{fig:Results/VU/_NB_3_NQ_3_N_K_4_N_U_64_delta_t_index_1_fcIndex_3_SNRIndex_4_1}}
\caption{CRLB of ${\mathbf{v}}_{U,0}$ as a function of $fc$, $N_U = 64$, SNR of $20 \text{ dB}$ which is constant across all links,  and $N_Q = 3$: (a) $N_B = 1$ and $N_K = 3$ and (b) $N_B = 3$ and $N_K = 4$.}
\label{Results:VU_NB_1_3_NQ_3_N_K_3_N_U_64_delta_t_index_1_fcIndex_3_SNRIndex_4_1}
\end{figure}

\begin{figure}[htb!]
\centering
\subfloat[]{\includegraphics[ width= 3.2in]{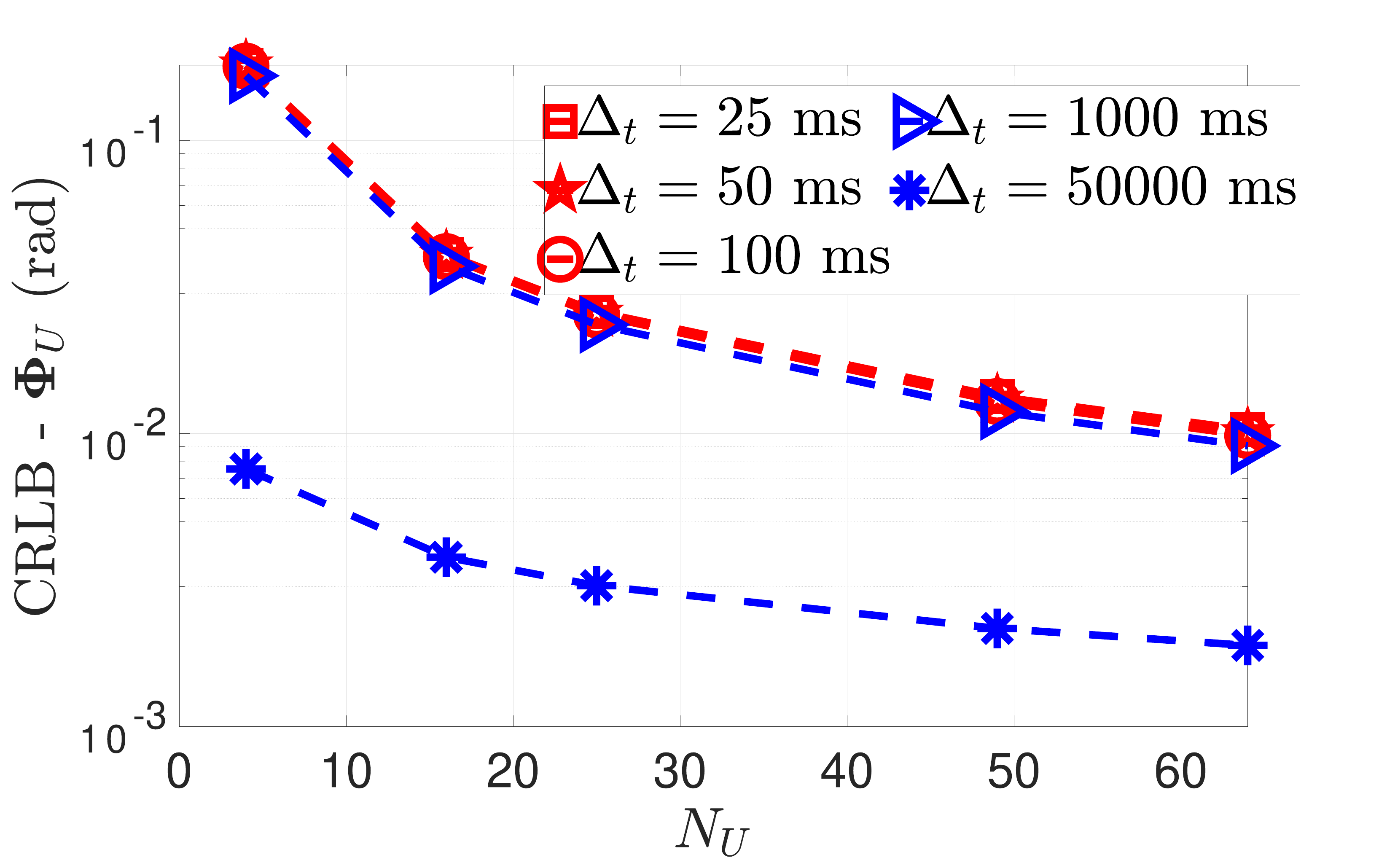}
\label{fig:Results/PHIU/_NB_1_NQ_3_N_K_3_N_U_64_delta_t_index_1_fcIndex_3_SNRIndex_4}}
\hfil
\subfloat[]{\includegraphics[ width= 3.2in]{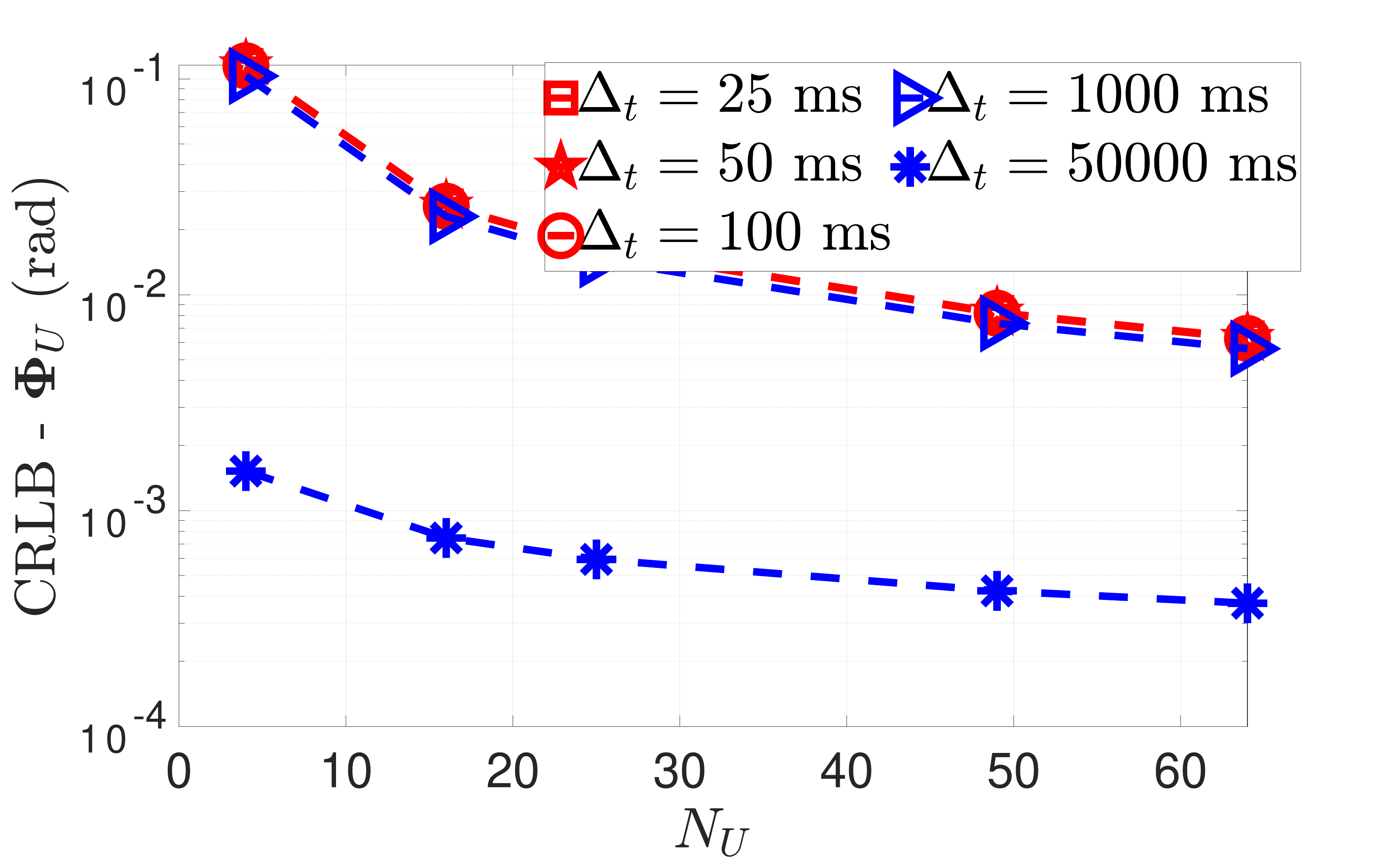}
\label{fig:Results/PHIU/_NB_3_NQ_3_N_K_4_N_U_64_delta_t_index_1_fcIndex_3_SNRIndex_4}}
\caption{CRLB of ${\mathbf{\Phi}}_{U}$ as a function of $N_U$, $f_c = 40 \text{ GHz}$, SNR of $20 \text{ dB}$ which is constant across all links,  and $N_Q = 3$: (a) $N_B = 1$ and $N_K = 3$ and (b) $N_B = 3$ and $N_K = 4$.}
\label{Results:PHIU_NB_1_3_NQ_3_N_K_3_N_U_64_delta_t_index_1_fcIndex_3_SNRIndex_4}
\end{figure}

\begin{figure}[htb!]
\centering
\subfloat[]{\includegraphics[ width= 3.2in]{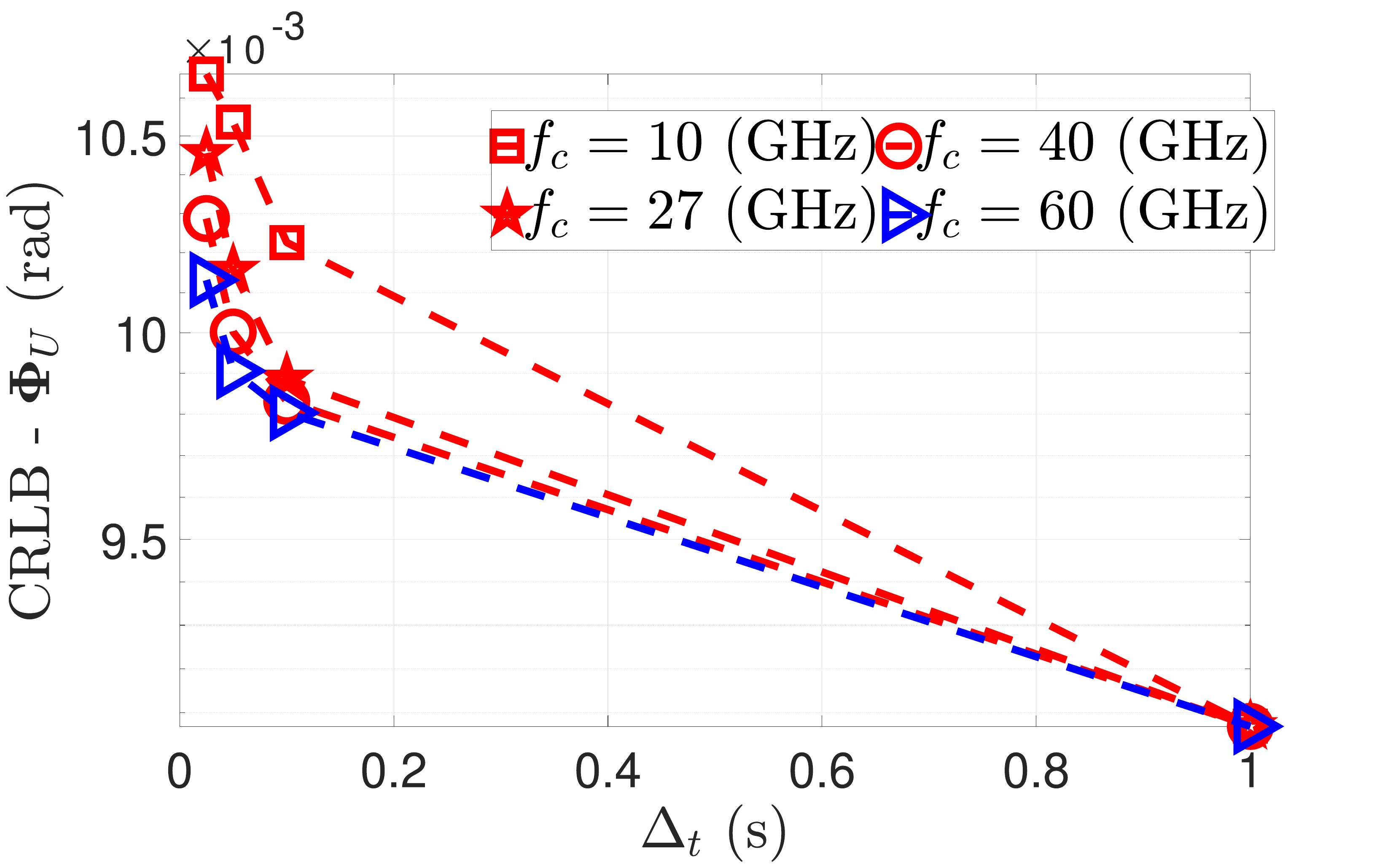}
\label{fig:Results/PHIU/_NB_1_NQ_3_N_K_3_N_U_64_delta_t_index_1_fcIndex_3_SNRIndex_4_1}}
\hfil
\subfloat[]{\includegraphics[ width= 3.2in]{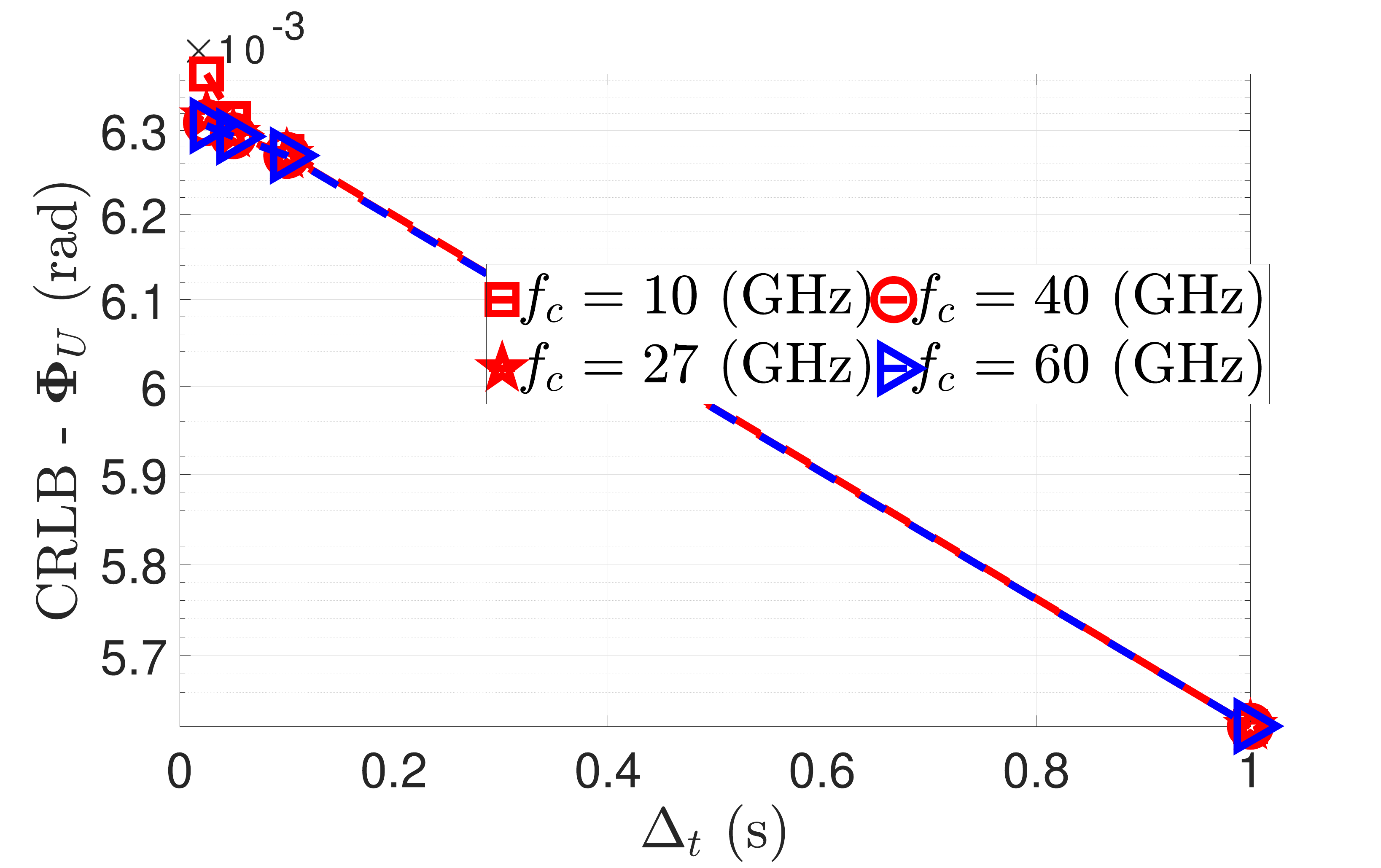}
\label{fig:Results/PHIU/_NB_3_NQ_3_N_K_4_N_U_64_delta_t_index_1_fcIndex_3_SNRIndex_4_1}}
\caption{CRLB of ${\mathbf{\Phi}}_{U}$ as a function of $fc$, $N_U = 64$, SNR of $20 \text{ dB}$ which is constant across all links, and $N_Q = 3$: (a) $N_B = 1$ and $N_K = 3$ and (b) $N_B = 3$ and $N_K = 4$.}
\label{Results:PHIU_NB_1_3_NQ_3_N_K_3_N_U_64_delta_t_index_1_fcIndex_3_SNRIndex_4_1}
\end{figure}

\begin{figure}[htb!]
\centering
\subfloat[]{\includegraphics[ width= 3.2in]{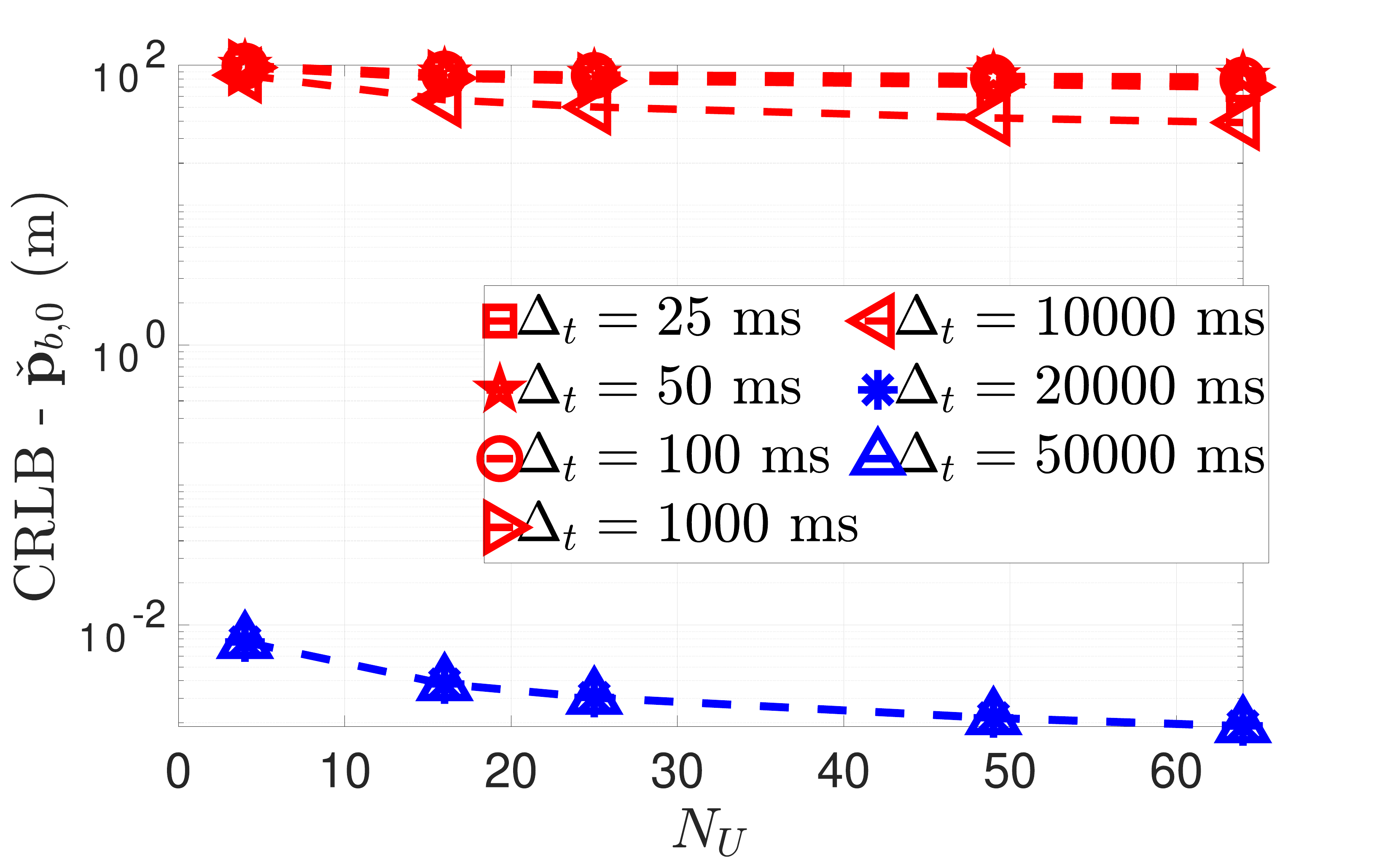}
\label{fig:Results/CPb/_NB_1_NQ_3_N_K_3_N_U_64_delta_t_index_1_fcIndex_3_SNRIndex_4}}
\hfil
\subfloat[]{\includegraphics[ width= 3.2in]{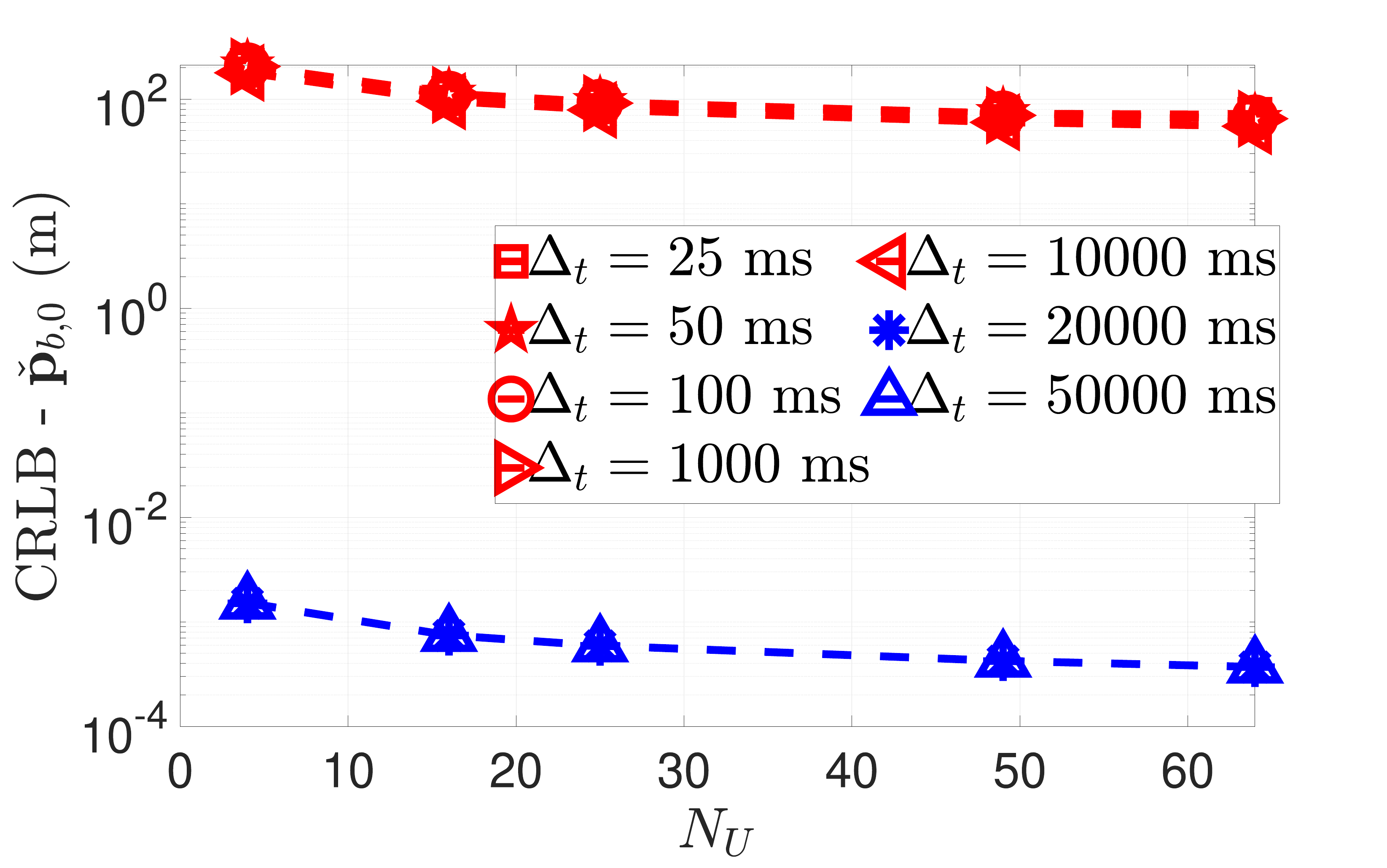}
\label{fig:Results/CPb/_NB_3_NQ_3_N_K_4_N_U_64_delta_t_index_1_fcIndex_3_SNRIndex_4}}
\caption{CRLB of $\check{\mathbf{p}}_{b,0}$ as a function of $N_U$, $f_c = 40 \text{ GHz}$, SNR of $20 \text{ dB}$ which is constant across all links,  and $N_Q = 3$: (a) $N_B = 1$ and $N_K = 3$ and (b) $N_B = 3$ and $N_K = 4$.}
\label{Results:CPb_NB_1_3_NQ_3_N_K_3_N_U_64_delta_t_index_1_fcIndex_3_SNRIndex_4}
\end{figure}

\begin{figure}[htb!]
\centering
\subfloat[]{\includegraphics[ width= 3.2in]{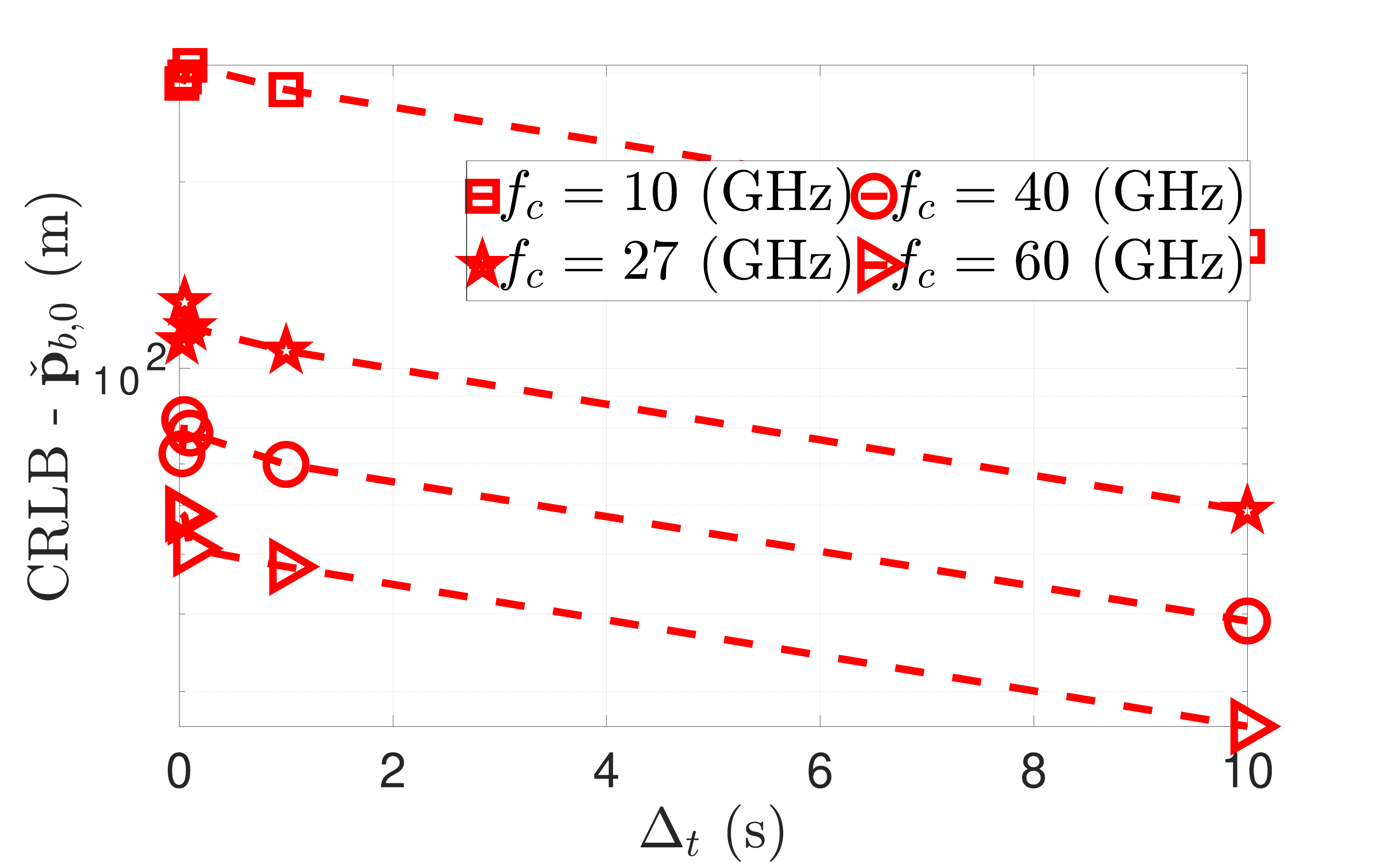}
\label{fig:Results/CPb/_NB_1_NQ_3_N_K_3_N_U_64_delta_t_index_1_fcIndex_3_SNRIndex_4_1}}
\hfil
\subfloat[]{\includegraphics[ width= 3.2in]{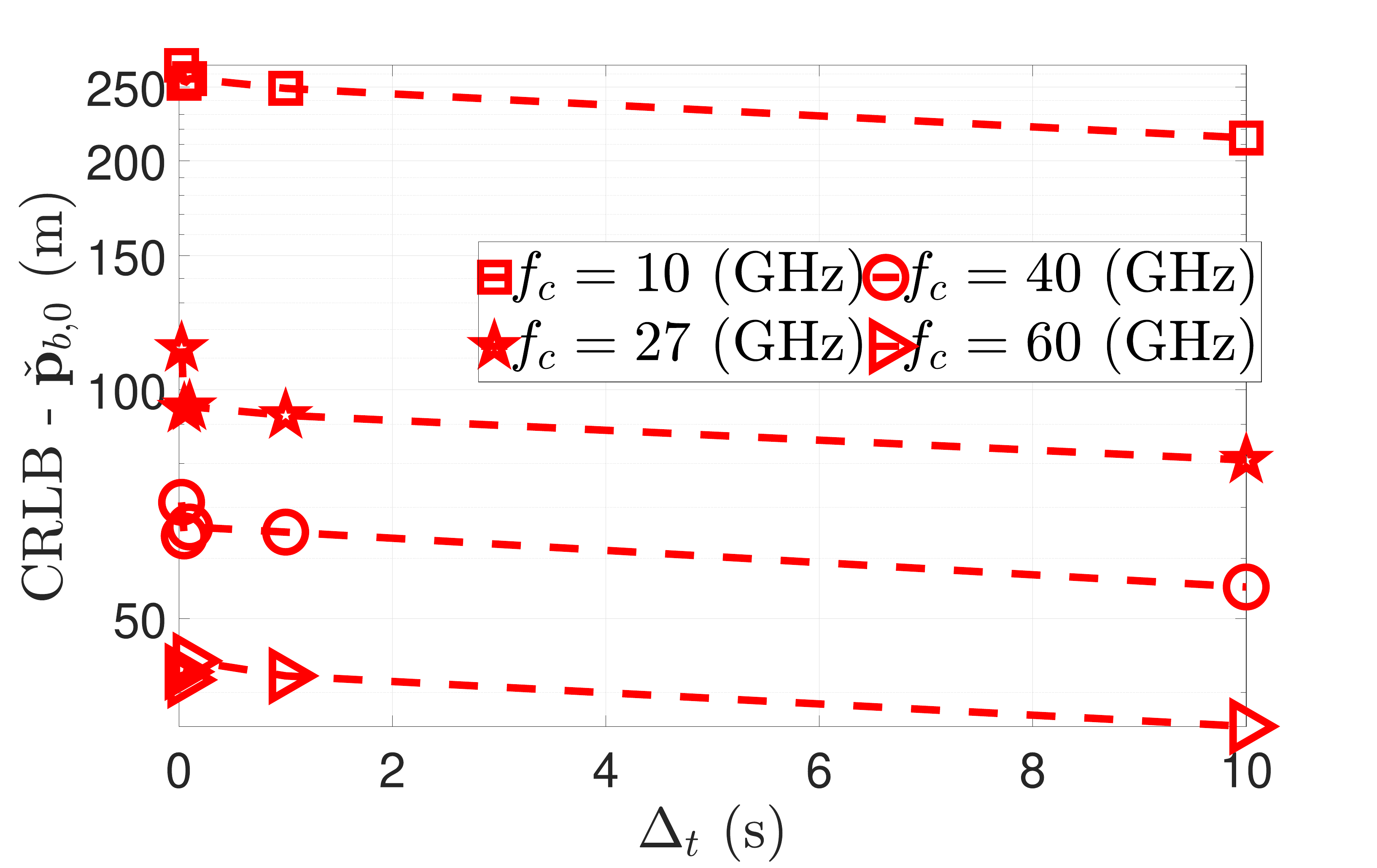}
\label{fig:Results/CPb/_NB_3_NQ_3_N_K_4_N_U_64_delta_t_index_1_fcIndex_3_SNRIndex_4_1}}
\caption{CRLB of $\check{\mathbf{p}}_{b,0}$ as a function of $f_c$, $N_U = 64
$, SNR of $20 \text{ dB}$ which is constant across all links,  and $N_Q = 3$: (a) $N_B = 1$ and $N_K = 3$ and (b) $N_B = 3$ and $N_K = 4$.}
\label{Results:CPb_NB_1_3_NQ_3_N_K_3_N_U_64_delta_t_index_1_fcIndex_3_SNRIndex_4_1}
\end{figure}

\begin{figure}[htb!]
\centering
\subfloat[]{\includegraphics[ width= 3.2in]{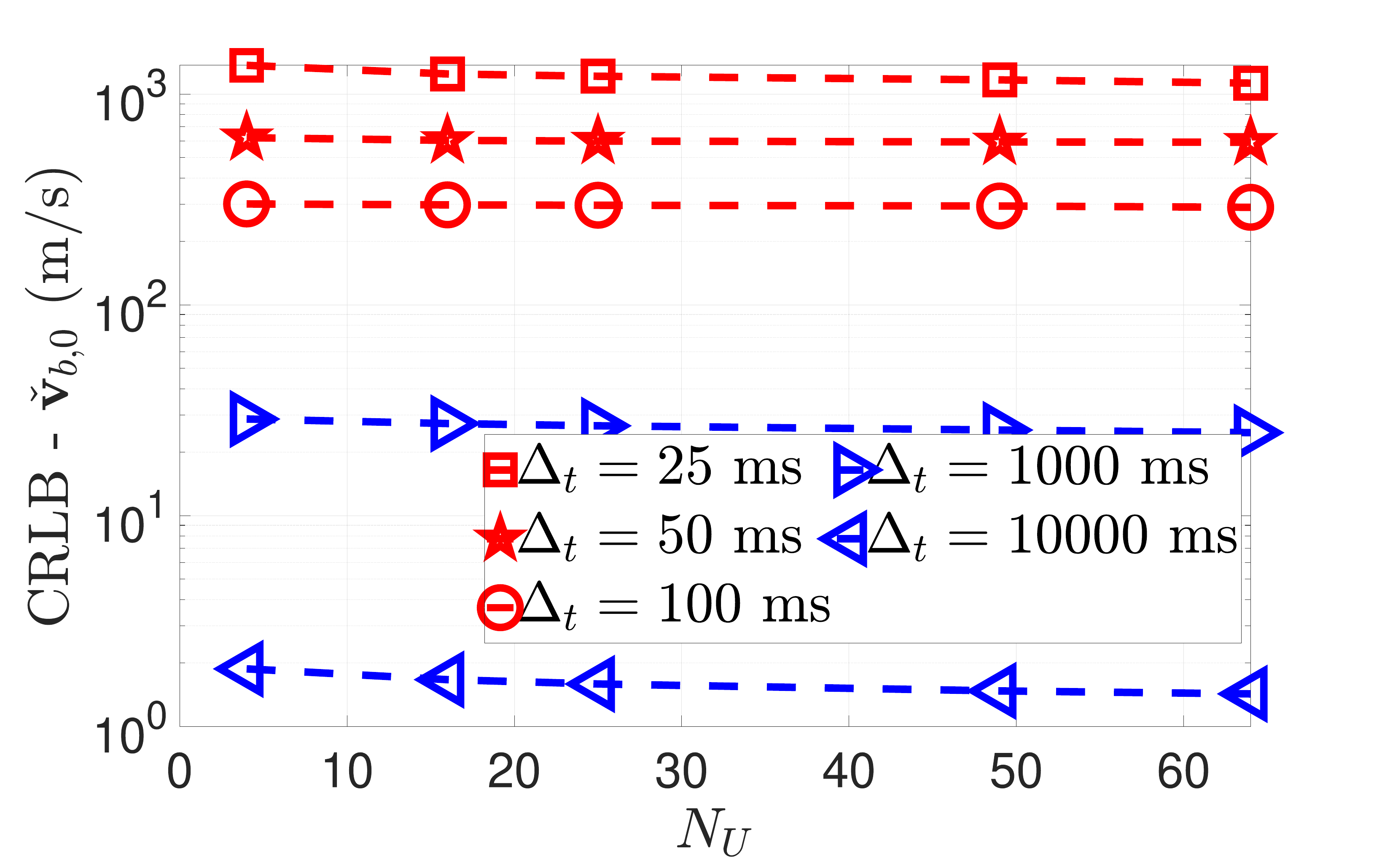}
\label{fig:Results/CVb/_NB_1_NQ_3_N_K_3_N_U_64_delta_t_index_1_fcIndex_3_SNRIndex_4}}
\hfil
\subfloat[]{\includegraphics[ width= 3.2in]{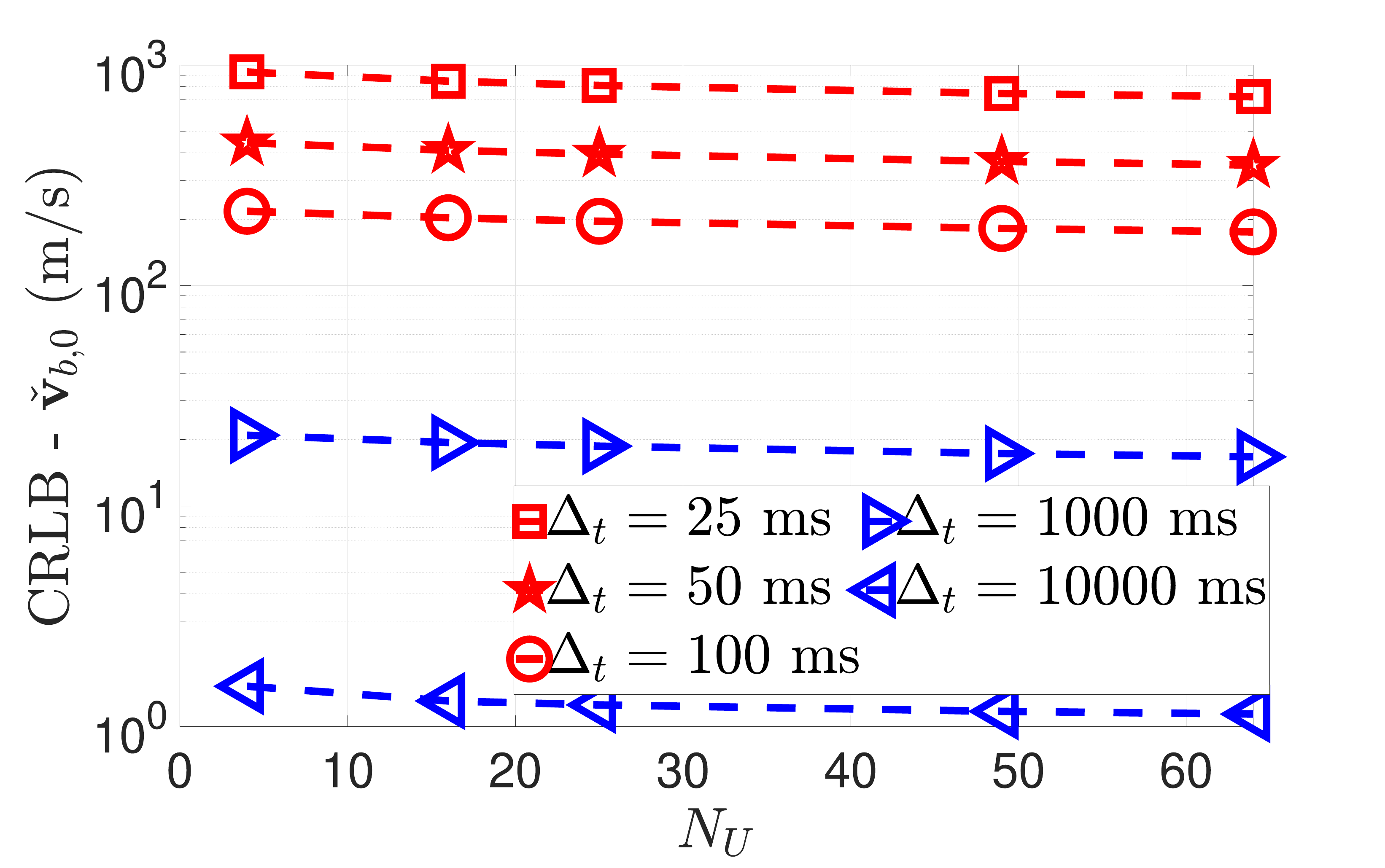}
\label{fig:Results/CVb/_NB_3_NQ_3_N_K_4_N_U_64_delta_t_index_1_fcIndex_3_SNRIndex_4}}
\caption{CRLB of $\check{\mathbf{v}}_{b,0}$ as a function of $N_U$, $f_c = 40 \text{ GHz}$, SNR of $20 \text{ dB}$ which is constant across all links,  and $N_Q = 3$: (a) $N_B = 1$ and $N_K = 3$ and (b) $N_B = 3$ and $N_K = 4$.}
\label{Results:CVb_NB_1_3_NQ_3_N_K_3_N_U_64_delta_t_index_1_fcIndex_3_SNRIndex_4}
\end{figure}

\begin{figure}[htb!]
\centering
\subfloat[]{\includegraphics[ width= 3.2in]{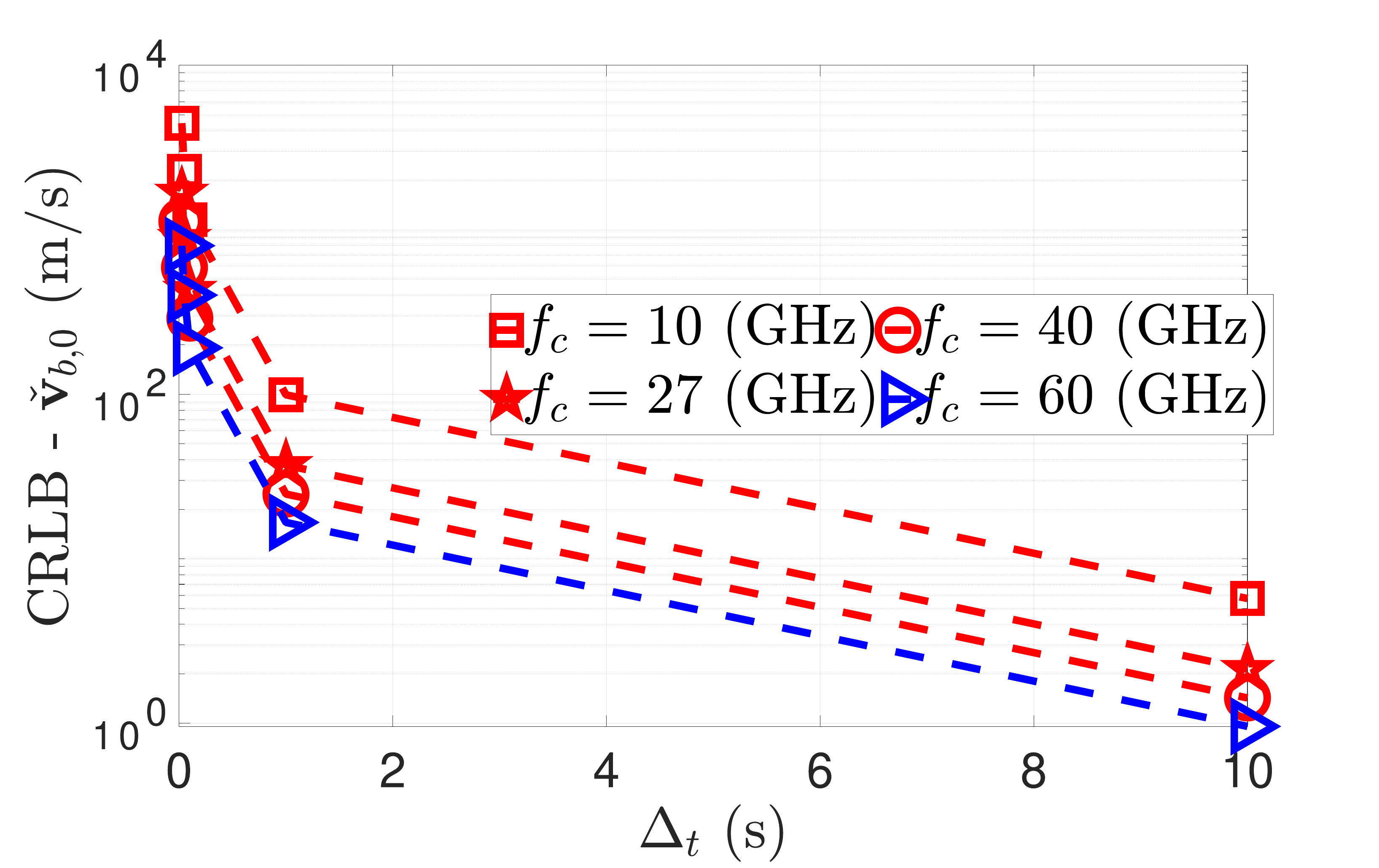}
\label{fig:Results/CVb/_NB_1_NQ_3_N_K_3_N_U_64_delta_t_index_1_fcIndex_3_SNRIndex_4_1}}
\hfil
\subfloat[]{\includegraphics[ width= 3.2in]{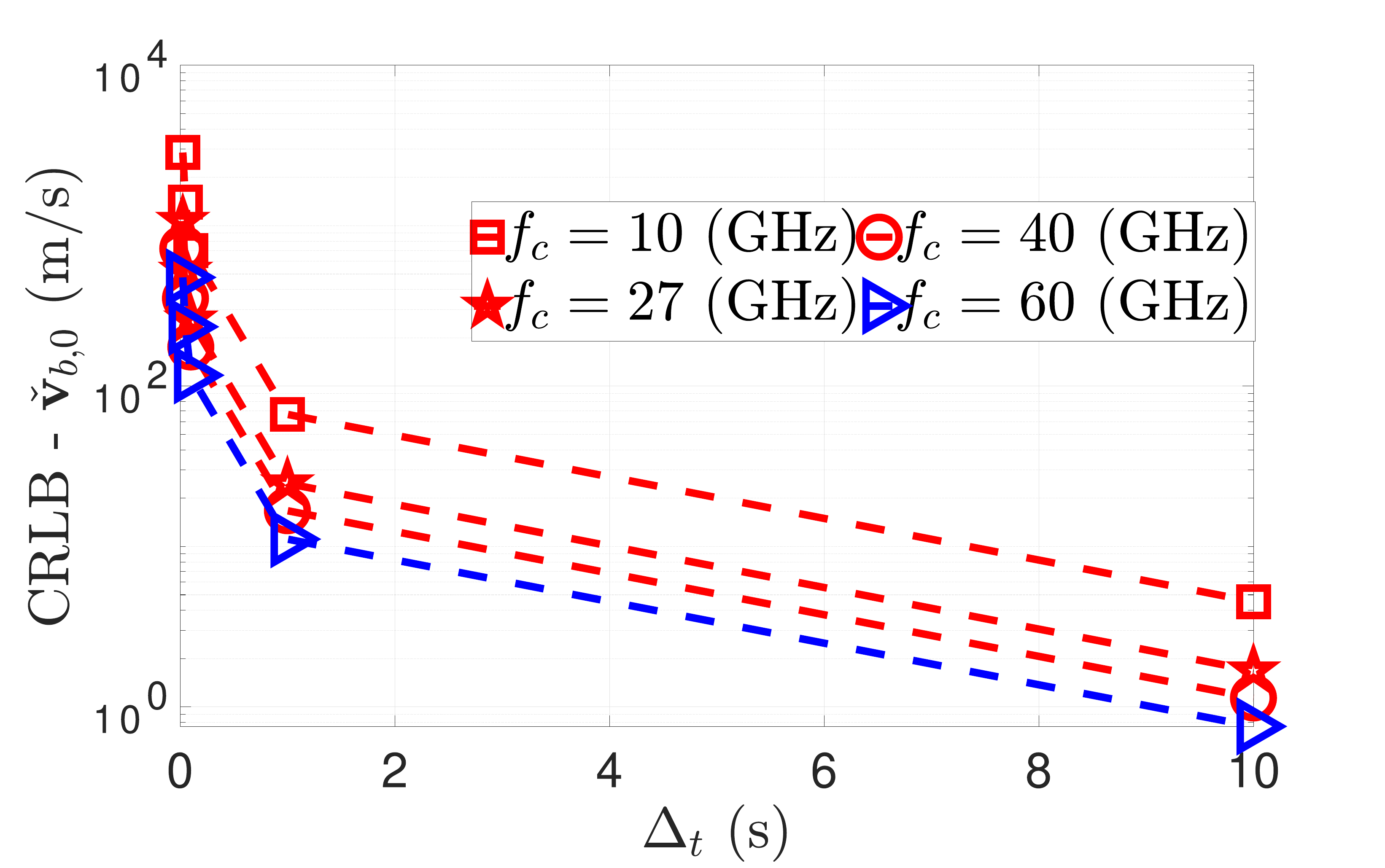}
\label{fig:Results/CVb/_NB_3_NQ_3_N_K_4_N_U_64_delta_t_index_1_fcIndex_3_SNRIndex_4_1}}
\caption{CRLB of $\check{\mathbf{v}}_{b,0}$ as a function of $f_c$, $N_U = 64
$, SNR of $20 \text{ dB}$ which is constant across all links,  and $N_Q = 3$: (a) $N_B = 1$ and $N_K = 3$ and (b) $N_B = 3$ and $N_K = 4$.}
\label{Results:CVb_NB_1_3_NQ_3_N_K_3_N_U_64_delta_t_index_1_fcIndex_3_SNRIndex_4_1}
\end{figure}

\section{Conclusion}

This paper has explored utilizing LEOs and $5$G BSs for both $9$D receiver localization and LEO ephemeris correction. We showed through the FIM that three LEO, three BSs, and four time slots are enough to estimate the $9$D location parameters and correct the LEO position and velocity. We obtained the CRLB and showed that with a single LEO, three time slots, and three BSs, the receiver positioning error, velocity estimation error, orientation error,  LEO position offset estimation error, and LEO velocity offset estimation error is $0.1 \text{ cm}$, $1 \text{ mm/s}$, $10^{-3} \text{ rad}$, $0.01 \text{ m}$, and $1 \text{ m/s}$, respectively. The receiver localization parameters are estimated after $1 \text{ s}$ while the LEO offset parameters are estimated after $20 \text{ s}.$ Our findings illuminate the path forward, revealing that operating frequency has little influence on orientation accuracy and that LEO velocity estimation is unaffected by the number of receiver antennas. This study is a step toward a future where seamless integration of LEO satellites and 5G networks redefines the boundaries of precision and connectivity.

\appendix
	 \label{appendix_channel} 

\subsection{Entries in transformation matrix}
\label{Appendix_Entries_in_transformation_matrix}
The elements in the transformation matrix are presented in this section. We start by presenting the elements related to the link from the $b^{\text{th}}$ LEO to the receiver. The derivative of the delays from one entity to another with respect to the position is the unit vector from that entity to the other. These derivatives are presented next.
$$
\nabla_{\bm{p}_{U,0}}\tau_{bu,k} \triangleq  \nabla_{\bm{p}_{U.0}}\frac{\left\|\mathbf{p}_{u,k}-\mathbf{p}_{b,k}\right\|}{c} = \nabla_{\bm{p}_{U,0}}\frac{\bm{d}_{bu,k}}{c} = \frac{\bm{\Delta}_{bu,k}}{c},
$$

$$
\nabla_{\bm{p}_{U,0}}\tau_{qu,k} \triangleq  \nabla_{\bm{p}_{U.0}}\frac{\left\|\mathbf{p}_{u,k}-\mathbf{p}_{q,k}\right\|}{c} = \nabla_{\bm{p}_{U,0}}\frac{\bm{d}_{qu,k}}{c} = \frac{\bm{\Delta}_{qu,k}}{c},
$$

$$
\nabla_{\check{\bm{p}}_{b,0}}\tau_{bu,k} \triangleq  \nabla_{\check{\bm{p}}_{b.0}}\frac{\left\|\mathbf{p}_{u,k}-\mathbf{p}_{b,k}\right\|}{c} = \nabla_{\check{\bm{p}}_{b,0}}\frac{\bm{d}_{bu,k}}{c} = -\frac{\bm{\Delta}_{bu,k}}{c},
$$

$$
\nabla_{\check{\bm{p}}_{b,0}}\tau_{bq,k} \triangleq  \nabla_{\check{\bm{p}}_{b.0}}\frac{\left\|\mathbf{p}_{q,k}-\mathbf{p}_{b,k}\right\|}{c} = \nabla_{\check{\bm{p}}_{b,0}}\frac{\bm{d}_{bq,k}}{c} = -\frac{\bm{\Delta}_{bq,k}}{c}.
$$

The derivative of the Dopplers from one entity to another with respect to the position is presented next.
$$
\nabla_{\bm{p}_{U,0}}\nu_{bU,k} \triangleq \frac{(\bm{v}_{b,k} - \bm{v}_{U,k}) - \bm{\Delta}_{bU,k}^{\mathrm{T}}(\bm{v}_{b,k} - \bm{v}_{U,k})\bm{\Delta}_{bU,k}}{c\bm{d}_{bU,k}^{-1}},
$$

$$
\nabla_{\bm{p}_{U,0}}\nu_{qU,k} \triangleq \frac{(\bm{v}_{q,k} - \bm{v}_{U,k}) - \bm{\Delta}_{qU,k}^{\mathrm{T}}(\bm{v}_{q,k} - \bm{v}_{U,k})\bm{\Delta}_{qU,k}}{c\bm{d}_{qU,k}^{-1}},
$$

$$
\nabla_{\check{\bm{p}}_{b,0}}\nu_{bU,k} \triangleq \frac{-(\bm{v}_{b,k} - \bm{v}_{U,k}) +\bm{\Delta}_{bU,k}^{\mathrm{T}}(\bm{v}_{b,k} - \bm{v}_{U,k})\bm{\Delta}_{bU,k}}{c\bm{d}_{bU,k}^{-1}},
$$

$$
\nabla_{\check{\bm{p}}_{b,0}}\nu_{bq,k} \triangleq \frac{-(\bm{v}_{b,k} - \bm{v}_{q,k}) + \bm{\Delta}_{bq,k}^{\mathrm{T}}(\bm{v}_{b,k} - \bm{v}_{q,k})\bm{\Delta}_{bq,k}}{c\bm{d}_{bq,k}^{-1}}.
$$
The derivatives with respect to the receiver's orientation are presented next.
$$
\nabla_{{\alpha}_{U}}\tau_{bu,k} \triangleq   \frac{\bm{\Delta}_{bu,k}^{\mathrm{T}} \nabla_{\alpha_{U}}\bm{Q}_{U}\Tilde{\bm{s}}_{u}}{c},
$$
$$
\nabla_{{\alpha}_{U}}\tau_{qu,k} \triangleq   \frac{\bm{\Delta}_{qu,k}^{\mathrm{T}} \nabla_{\alpha_{U}}\bm{Q}_{U}\Tilde{\bm{s}}_{u}}{c},
$$

$$
\nabla_{{\psi}_{U}}\tau_{bu,k} \triangleq   \frac{\bm{\Delta}_{bu,k}^{\mathrm{T}} \nabla_{\psi_{U}}\bm{Q}_{U}\Tilde{\bm{s}}_{u}}{c},
$$
$$
\nabla_{{\psi}_{U}}\tau_{qu,k} \triangleq   \frac{\bm{\Delta}_{qu,k}^{\mathrm{T}} \nabla_{\psi_{U}}\bm{Q}_{U}\Tilde{\bm{s}}_{u}}{c},
$$

$$
\nabla_{{\varphi}_{U}}\tau_{bu,k} \triangleq   \frac{\bm{\Delta}_{bu,k}^{\mathrm{T}} \nabla_{\varphi_{U}}\bm{Q}_{U}\Tilde{\bm{s}}_{u}}{c},
$$
$$
\nabla_{{\varphi}_{U}}\tau_{qu,k} \triangleq   \frac{\bm{\Delta}_{qu,k}^{\mathrm{T}} \nabla_{\varphi_{U}}\bm{Q}_{U}\Tilde{\bm{s}}_{u}}{c},
$$
$$
\nabla_{\bm{\Phi}_{U}}\tau_{bu,k} \triangleq   
\left[\begin{array}{c}
{\bm{\Delta}_{bu,k}^{\mathrm{T}} \nabla_{\alpha_{U}}\bm{Q}_{U}\Tilde{\bm{s}}_{u}} \\
{\bm{\Delta}_{bu,k}^{\mathrm{T}} \nabla_{\psi_{U}}\bm{Q}_{U}\Tilde{\bm{s}}_{u}} \\
{\bm{\Delta}_{bu,k}^{\mathrm{T}} \nabla_{\varphi_{U}}\bm{Q}_{U}\Tilde{\bm{s}}_{u}}
\end{array}\right].
$$
$$
\nabla_{\bm{\Phi}_{U}}\tau_{qu,k} \triangleq   
\left[\begin{array}{c}
{\bm{\Delta}_{qu,k}^{\mathrm{T}} \nabla_{\alpha_{U}}\bm{Q}_{U}\Tilde{\bm{s}}_{u}} \\
{\bm{\Delta}_{qu,k}^{\mathrm{T}} \nabla_{\psi_{U}}\bm{Q}_{U}\Tilde{\bm{s}}_{u}} \\
{\bm{\Delta}_{qu,k}^{\mathrm{T}} \nabla_{\varphi_{U}}\bm{Q}_{U}\Tilde{\bm{s}}_{u}}
\end{array}\right]
$$
The derivatives of the delays with respect to the velocities are presented next.

$$
\nabla_{\bm{v}_{U,0}}\tau_{bu,k} \triangleq  \nabla_{\bm{v}_{U.0}}\frac{\left\|\mathbf{p}_{u,k}-\mathbf{p}_{b,k}\right\|}{c} = (k) \Delta_{t}\frac{\bm{\Delta}_{bu,k}}{c},
$$

$$
\nabla_{\bm{p}_{U,0}}\tau_{qu,k} \triangleq  \nabla_{\bm{p}_{U.0}}\frac{\left\|\mathbf{p}_{u,k}-\mathbf{p}_{q,k}\right\|}{c} = (k) \Delta_{t}\frac{\bm{\Delta}_{qu,k}}{c},
$$

$$
\nabla_{\check{\bm{p}}_{b,0}}\tau_{bu,k} \triangleq  \nabla_{\check{\bm{p}}_{b.0}}\frac{\left\|\mathbf{p}_{u,k}-\mathbf{p}_{b,k}\right\|}{c} = -(k) \Delta_{t}\frac{\bm{\Delta}_{bu,k}}{c},
$$

$$
\nabla_{\check{\bm{p}}_{b,0}}\tau_{bq,k} \triangleq  \nabla_{\check{\bm{p}}_{b.0}}\frac{\left\|\mathbf{p}_{q,k}-\mathbf{p}_{b,k}\right\|}{c} =  -(k) \Delta_{t}\frac{\bm{\Delta}_{bq,k}}{c}.
$$
The derivatives of the Dopplers with respect to the velocities are presented next.
$$
\nabla_{\bm{v}_{U,0}}\nu_{bU,k} \triangleq   -\frac{\bm{\Delta}_{bU,k}}{c},
$$
$$
\nabla_{\bm{v}_{U,0}}\nu_{qU,k} \triangleq   -\frac{\bm{\Delta}_{qU,k}}{c},
$$
$$
\nabla_{\check{\bm{v}}_{b,0}}\nu_{bU,k} \triangleq  \frac{\bm{\Delta}_{bU,k}}{c},
$$
$$
\nabla_{\check{\bm{v}}_{b,0}}\nu_{bq,k} \triangleq  \frac{\bm{\Delta}_{bq,k}}{c}.
$$
\subsection{Proof of elements in $\mathbf{J}_{ \bm{\bm{y}}; \bm{\kappa}_1}^{}$}
\label{Appendix_Entries_in_the_EFIM_Location_1}
The proof of the elements in $\mathbf{J}_{ \bm{\bm{y}}; \bm{\kappa}_1}^{}$ are presented in this section. We start with the elements that are related to the $3D$ position of the receiver.

\subsubsection{Proof of the FIM related to $\bm{p}_{U,0}$}
\label{Appendix_lemma_FIM_3D_position}
The FIM of the $3$D position of the receiver can be written as (\ref{equ_lemma:FIM_3D_position_1}), which  can be simplified to (\ref{equ_lemma:FIM_3D_position_2}). 
Finally, substituting the FIM for the channel parameters
gives us (\ref{equ_lemma:FIM_3D_position}).

\begin{figure*}
\begin{align}
\begin{split}
\label{equ_lemma:FIM_3D_position_1}
\medmath{\bm{F}_{{{y} }}(\bm{y}_{}| \bm{\eta} ;\bm{p}_{U,0},\bm{p}_{U,0}) }
&= \medmath{\sum_{b,k^{'},u^{'},k^{},u^{}} \frac{1}{c^2}\bm{\Delta}_{bu,k} \bm{F}_{{{y} }}({y}_{bu,k}| {\eta} ;{\tau}_{bu,k},{\tau}_{bu^{'},k^{'}}) \bm{\Delta}_{bu^{'},k^{'}}^{\mathrm{T}} +  \frac{1}{c}\bm{\Delta}_{bu^{},k^{}} \bm{F}_{{{y} }}({y}_{bu,k}| {\eta} ;{\tau}_{bu^{},k^{}},{\nu}_{bU,k^{'}}) \nabla_{\bm{p}_{U,0}}^{\mathrm{T}} \nu_{bU,k^{'}}} \\ & \medmath{+ \frac{1}{c}\nabla_{\bm{p}_{U,0}} \nu_{bU,k} \bm{F}_{{{y} }}({y}_{bu,k}| {\eta} ;{\nu}_{bU,k},{\tau}_{bu^{'},k^{'}}) \bm{\Delta}_{bu^{'},k^{'}}^{\mathrm{T}}
+ \nabla_{\bm{p}_{U,0}} \nu_{bU,k} \bm{F}_{{{y} }}({y}_{bu,k}| {\eta} ;{\nu}_{bU,k},{\nu}_{b^{}U,k^{'}}) \nabla_{\bm{p}_{U,0}}^{\mathrm{T}} \nu_{b^{}U,k^{'}}} \\
& \medmath{+ \sum_{q,k^{'},u^{'},k^{},u^{}}\frac{1}{c^2}\bm{\Delta}_{qu,k} \bm{F}_{{{y} }}({y}_{qu,k}| {\eta} ;{\tau}_{qu,k},{\tau}_{qu^{'},k^{'}}) \bm{\Delta}_{qu^{'},k^{'}}^{\mathrm{T}} +  \frac{1}{c}\bm{\Delta}_{qu^{},k^{}} \bm{F}_{{{y} }}({y}_{qu,k}| {\eta} ;{\tau}_{qu^{},k^{}},{\nu}_{qU,k^{'}}) \nabla_{\bm{p}_{U,0}}^{\mathrm{T}} \nu_{qU,k^{'}}} \\ & \medmath{+ \frac{1}{c} \nabla_{\bm{p}_{U,0}} \nu_{qU,k} \bm{F}_{{{y} }}({y}_{qu,k}| {\eta} ;{\nu}_{qU,k},{\tau}_{qu^{'},k^{'}}) \bm{\Delta}_{qu^{'},k^{'}}^{\mathrm{T}}
+ \nabla_{\bm{p}_{U,0}} \nu_{qU,k} \bm{F}_{{{y} }}({y}_{qu,k}| {\eta} ;{\nu}_{qU,k},{\nu}_{q^{}U,k^{'}}) \nabla_{\bm{p}_{U,0}}^{\mathrm{T}} \nu_{q^{}U,k^{'}}},
\end{split}
\end{align}
\end{figure*}
\begin{figure*}
\begin{align}
\begin{split}
\label{equ_lemma:FIM_3D_position_2}
\medmath{\bm{F}_{{{y} }}(\bm{y}_{}| \bm{\eta} ;\bm{p}_{U,0},\bm{p}_{U,0})} 
&= \medmath{\sum_{b,k^{},u^{}} \frac{1}{c^2}\bm{\Delta}_{bu,k} \bm{F}_{{{y} }}({y}_{bu,k}| {\eta} ;{\tau}_{bu,k},{\tau}_{bu^{},k^{}}) \bm{\Delta}_{bu^{},k^{}}^{\mathrm{T}} +  \frac{1}{c}\bm{\Delta}_{bu^{},k^{}} \bm{F}_{{{y} }}({y}_{bu,k}| {\eta} ;{\tau}_{bu^{},k^{}},{\nu}_{bU,k^{}}) \nabla_{\bm{p}_{U,0}}^{\mathrm{T}} \nu_{bU,k^{}}} \\ & \medmath{+ \frac{1}{c}\nabla_{\bm{p}_{U,0}} \nu_{bU,k} \bm{F}_{{{y} }}({y}_{bu,k}| {\eta} ;{\nu}_{bU,k},{\tau}_{bu^{},k^{}}) \bm{\Delta}_{bu^{},k^{}}^{\mathrm{T}}
+ \nabla_{\bm{p}_{U,0}} \nu_{bU,k} \bm{F}_{{{y} }}({y}_{bu,k}| {\eta} ;{\nu}_{bU,k},{\nu}_{b^{}U,k^{}}) \nabla_{\bm{p}_{U,0}}^{\mathrm{T}} \nu_{b^{}U,k^{}}} \\
&+ \medmath{\sum_{q,k^{},u^{}}\frac{1}{c^2}\bm{\Delta}_{qu,k} \bm{F}_{{{y} }}({y}_{qu,k}| {\eta} ;{\tau}_{qu,k},{\tau}_{qu^{},k^{}}) \Delta_{qu^{},k^{}}^{\mathrm{T}} +  \frac{1}{c}\bm{\Delta}_{qu^{},k^{}} \bm{F}_{{{y} }}({y}_{qu,k}| {\eta} ;{\tau}_{qu^{},k^{}},{\nu}_{qU,k^{}}) \nabla_{\bm{p}_{U,0}}^{\mathrm{T}} \nu_{qU,k^{}}} \\ & \medmath{+ \frac{1}{c} \nabla_{\bm{p}_{U,0}} \nu_{qU,k} \bm{F}_{{{y} }}({y}_{qu,k}| {\eta} ;{\nu}_{qU,k},{\tau}_{qu^{},k^{}}) \bm{\Delta}_{qu^{},k^{}}^{\mathrm{T}}
+ \nabla_{\bm{p}_{U,0}} \nu_{qU,k} \bm{F}_{{{y} }}({y}_{qu,k}| {\eta} ;{\nu}_{qU,k},{\nu}_{q^{}U,k^{}}) \nabla_{\bm{p}_{U,0}}^{\mathrm{T}} \nu_{q^{}U,k^{}}},
\end{split}
\end{align}
\end{figure*}

\subsubsection{Proof of the FIM related to $\bm{p}_{U,0}$ and $\bm{v}_{U,0}$}
\label{Appendix_lemma_FIM_3D_position_3D_velocity}
The FIM related to $\bm{p}_{U,0}$ and $\bm{v}_{U,0}$ can be written as (\ref{equ_lemma:FIM_3D_position_3D_velocity_1}), which  can be simplified to (\ref{equ_lemma:FIM_3D_position_3D_velocity_2}). 
Finally, substituting the FIM for the channel parameters
gives us (\ref{equ_lemma:FIM_3D_position_3D_velocity}).

\begin{figure*}
\begin{align}
\begin{split}
\label{equ_lemma:FIM_3D_position_3D_velocity_1}
\medmath{\bm{F}_{{{y} }}(\bm{y}_{}| \bm{\eta} ;\bm{p}_{U,0},\bm{v}_{U,0})} 
&= \medmath{\sum_{b,k^{'},u^{'},k^{},u^{}} \frac{(k^{'}) \Delta_{t}}{c^2}\bm{\Delta}_{bu,k} \bm{F}_{{{y} }}({y}_{bu,k}| {\eta} ;{\tau}_{bu,k},{\tau}_{bu^{'},k^{'}}) \bm{\Delta}_{bu^{'},k^{'}}^{\mathrm{T}} -  \frac{1}{c^2}\bm{\Delta}_{bu^{},k^{}} \bm{F}_{{{y} }}({y}_{bu,k}| {\eta} ;{\tau}_{bu^{},k^{}},{\nu}_{bU,k^{'}}) \bm{\Delta}_{bU,k^{'}}^{\mathrm{T}}} \\ &+\medmath{\frac{(k^{'})\Delta_{t}}{c}\nabla_{\bm{p}_{U,0}} \nu_{bU,k} \bm{F}_{{{y} }}({y}_{bu,k}| {\eta} ;{\nu}_{bU,k},{\tau}_{bu^{'},k^{'}}) \Delta_{bu^{'},k^{'}}^{\mathrm{T}}
-\frac{1}{c}\nabla_{\bm{p}_{U,0}} \nu_{bU,k} \bm{F}_{{{y} }}({y}_{bu,k}| {\eta} ;{\nu}_{bU,k},{\nu}_{b^{}U,k^{'}}) \bm{\Delta}_{bU,k^{'}}^{\mathrm{T}}} \\
&+ \medmath{\sum_{q,,k^{'},u^{'},k^{},u^{}}\frac{(k^{'})\Delta_t^{}}{c^2}\bm{\Delta}_{qu,k} \bm{F}_{{{y} }}({y}_{qu,k}| {\eta} ;{\tau}_{qu,k},{\tau}_{qu^{'},k^{'}}) \bm{\Delta}_{qu^{'},k^{'}}^{\mathrm{T}} -  \frac{1}{c^2}\bm{\Delta}_{qu^{},k^{}} \bm{F}_{{{y} }}({y}_{qu,k}| {\eta} ;{\tau}_{qu^{},k^{}},{\nu}_{qU,k^{'}}) \bm{\Delta}_{qU,k^{'}}^{\mathrm{T}}} \\ &+ \medmath{\frac{(k^{'})\Delta_{t}}{c}\bm{\Delta}_{qU,k^{}} \bm{F}_{{{y} }}({y}_{qu,k}| {\eta} ;{\nu}_{qU,k},{\tau}_{qu^{'},k^{'}}) \bm{\Delta}_{qu^{'},k^{'}}^{\mathrm{T}}
-  \frac{1}{c}\nabla_{\bm{p}_{U,0}} \nu_{qU,k^{'}} \bm{F}_{{{y} }}({y}_{qu,k}| {\eta} ;{\nu}_{qU,k},{\nu}_{q^{}U,k^{'}}) \bm{\Delta}_{qU,k^{'}}^{\mathrm{T}}},
\end{split}
\end{align}
\end{figure*}
\begin{figure*}
\begin{align}
\begin{split}
\label{equ_lemma:FIM_3D_position_3D_velocity_2}
\medmath{\bm{F}_{{{y} }}(\bm{y}_{}| \bm{\eta} ;\bm{p}_{U,0},\bm{v}_{U,0})} 
&= \medmath{\sum_{b,k^{},u^{}} \frac{(k^{}) \Delta_{t}}{c^2}\bm{\Delta}_{bu,k} \bm{F}_{{{y} }}({y}_{bu,k}| {\eta} ;{\tau}_{bu,k},{\tau}_{bu^{},k^{}}) \bm{\Delta}_{bu^{},k^{}}^{\mathrm{T}} -  \frac{1}{c^2}\bm{\Delta}_{bu^{},k^{}} \bm{F}_{{{y} }}({y}_{bu,k}| {\eta} ;{\tau}_{bu^{},k^{}},{\nu}_{bU,k^{}}) \bm{\Delta}_{bU,k^{}}^{\mathrm{T}}} \\ &+\medmath{\frac{(k^{})\Delta_{t}}{c}\nabla_{\bm{p}_{U,0}} \nu_{bU,k} \bm{F}_{{{y} }}({y}_{bu,k}| {\eta} ;{\nu}_{bU,k},{\tau}_{bu^{},k^{}}) \Delta_{bu^{},k^{}}^{\mathrm{T}}
-\frac{1}{c}\nabla_{\bm{p}_{U,0}} \nu_{bU,k} \bm{F}_{{{y} }}({y}_{bu,k}| {\eta} ;{\nu}_{bU,k},{\nu}_{b^{}U,k^{}}) \bm{\Delta}_{bU,k^{}}^{\mathrm{T}}} \\
&+ \medmath{\sum_{q,k^{},u^{}}\frac{(k^{})\Delta_t^{}}{c^2}\bm{\Delta}_{qu,k} \bm{F}_{{{y} }}({y}_{qu,k}| {\eta} ;{\tau}_{qu,k},{\tau}_{qu^{},k^{}}) \bm{\Delta}_{qu^{},k^{}}^{\mathrm{T}} -  \frac{1}{c^2}\bm{\Delta}_{qu^{},k^{}} \bm{F}_{{{y} }}({y}_{qu,k}| {\eta} ;{\tau}_{qu^{},k^{}},{\nu}_{qU,k^{}}) \bm{\Delta}_{qU,k^{}}^{\mathrm{T}}} \\ &+ \medmath{\frac{(k^{})\Delta_{t}}{c}\bm{\Delta}_{qU,k^{}} \bm{F}_{{{y} }}({y}_{qu,k}| {\eta} ;{\nu}_{qU,k},{\tau}_{qu^{},k^{}}) \bm{\Delta}_{qu^{},k^{}}^{\mathrm{T}}
-  \frac{1}{c}\nabla_{\bm{p}_{U,0}} \nu_{qU,k^{}} \bm{F}_{{{y} }}({y}_{qu,k}| {\eta} ;{\nu}_{qU,k},{\nu}_{q^{}U,k^{}})\bm{\Delta}_{qU,k^{}}^{\mathrm{T}}},
\end{split}
\end{align}
\end{figure*}

\subsubsection{Proof of the FIM related to $\bm{p}_{U,0}$ and $\bm{\Phi}_{U}$}
\label{Appendix_lemma_FIM_3D_position_3D_orientation}
The FIM related to $\bm{p}_{U,0}$ and $\bm{\Phi}_{U}$ can be written as (\ref{equ_lemma:FIM_3D_position_3D_orientation_1}), which  can be simplified to (\ref{equ_lemma:FIM_3D_position_3D_orientation_2}). 
Finally, substituting the FIM for the channel parameters
gives us (\ref{equ_lemma:FIM_3D_position_3D_orientation}).

\begin{figure*}
\begin{align}
\begin{split}
\label{equ_lemma:FIM_3D_position_3D_orientation_1}
\medmath{\bm{F}_{{{y} }}(\bm{y}_{}| \bm{\eta} ;\bm{p}_{U,0},\bm{\Phi}_{U}) 
= \sum_{b,k^{'},u^{'},k^{},u^{}} \frac{1}{c}\bm{\Delta}_{bu,k} \bm{F}_{{{y} }}({y}_{bu,k}| {\eta} ;{\tau}_{bu,k},{\tau}_{bu^{'},k^{'}}) \nabla_{\bm{\Phi}_{U}}^{\mathrm{T}} \tau_{bu^{'},k^{'}} + \sum_{q,k^{'},u^{'},k^{},u^{}}\frac{1}{c}\bm{\Delta}_{qu,k} \bm{F}_{{{y} }}({y}_{qu,k}| {\eta} ;{\tau}_{qu,k},{\tau}_{qu^{'},k^{'}}) \nabla_{\bm{\Phi}_{U}}^{\mathrm{T}} \tau_{qu^{'},k^{'}}},
\end{split}
\end{align}
\end{figure*}
\begin{figure*}
\begin{align}
\begin{split}
\label{equ_lemma:FIM_3D_position_3D_orientation_2}
\medmath{\bm{F}_{{{y} }}(\bm{y}_{}| \bm{\eta} ;\bm{p}_{U,0},\bm{\Phi}_{U}) 
= \sum_{b,k^{},u^{}} \frac{1}{c}\bm{\Delta}_{bu,k} \bm{F}_{{{y} }}({y}_{bu,k}| {\eta} ;{\tau}_{bu,k},{\tau}_{bu^{},k^{}}) \nabla_{\bm{\Phi}_{U}}^{\mathrm{T}} \tau_{bu^{},k^{}} + \sum_{q,k^{},u^{}}\frac{1}{c}\bm{\Delta}_{qu,k} \bm{F}_{{{y} }}({y}_{qu,k}| {\eta} ;{\tau}_{qu,k},{\tau}_{qu^{},k^{}}) \nabla_{\bm{\Phi}_{U}}^{\mathrm{T}} \tau_{qu^{},k^{}}},
\end{split}
\end{align}
\end{figure*}

\subsubsection{Proof of the FIM related to $\bm{p}_{U,0}$ and $\check{\bm{p}}_{b,0}$}
\label{Appendix_lemma_FIM_3D_position_3D_b_th_position_offset}
The FIM relating the $3$D position and the uncertainty in position associated with the $b^{\text{th}}$ LEO can be written as (\ref{equ_lemma:FIM_3D_position_bth_position_1}), which  can be simplified to (\ref{equ_lemma:FIM_3D_position_bth_position_2}). 
Finally, substituting the FIM for the channel parameters
gives us (\ref{equ_lemma:FIM_3D_position_3D_b_th_position_offset}).

\begin{figure*}
\begin{align}
\begin{split}
\label{equ_lemma:FIM_3D_position_bth_position_1}
\medmath{\bm{F}_{{{y} }}(\bm{y}_{}| \bm{\eta} ;\bm{p}_{U,0},\check{\bm{p}}_{b,0})}&= \medmath{\sum_{k^{'},u^{'},k^{},u^{}} \frac{-1}{c^2}\bm{\Delta}_{bu,k} \bm{F}_{{{y} }}({y}_{bu,k}| {\eta} ;{\tau}_{bu,k},{\tau}_{bu^{'},k^{'}}) \bm{\Delta}_{bu^{'},k^{'}}^{\mathrm{T}} +  \frac{1}{c}\bm{\Delta}_{bu^{},k^{}} \bm{F}_{{{y} }}({y}_{bu,k}| {\eta} ;{\tau}_{bu^{},k^{}},{\nu}_{bU,k^{'}}) \nabla_{\check{\bm{p}}_{b,0}}^{\mathrm{T}} \nu_{bU,k^{'}}} \\ &+ \medmath{\frac{-1}{c}\nabla_{\bm{p}_{U,0}} \nu_{bU,k} \bm{F}_{{{y} }}({y}_{bu,k}| {\eta} ;{\nu}_{bU,k},{\tau}_{bu^{'},k^{'}}) \bm{\Delta}_{bu^{'},k^{'}}^{\mathrm{T}}
+ \nabla_{\bm{p}_{U,0}} \nu_{bU,k} \bm{F}_{{{y} }}({y}_{bu,k}| {\eta} ;{\nu}_{bU,k},{\nu}_{b^{}U,k^{'}}) \nabla_{\check{\bm{p}}_{b,0}}^{\mathrm{T}} \nu_{b^{}U,k^{'}}},
\end{split}
\end{align}
\end{figure*}
\begin{figure*}
\begin{align}
\begin{split}
\label{equ_lemma:FIM_3D_position_bth_position_2}
\medmath{\bm{F}_{{{y} }}(\bm{y}_{}| \bm{\eta} ;\bm{p}_{U,0},\check{\bm{p}}_{b,0})}
&= \medmath{\sum_{k^{},u^{}} \frac{-1}{c^2}\bm{\Delta}_{bu,k} \bm{F}_{{{y} }}({y}_{bu,k}| {\eta} ;{\tau}_{bu,k},{\tau}_{bu^{},k^{}}) \bm{\Delta}_{bu^{},k^{}}^{\mathrm{T}} +  \frac{1}{c}\bm{\Delta}_{bu^{},k^{}} \bm{F}_{{{y} }}({y}_{bu,k}| {\eta} ;{\tau}_{bu^{},k^{}},{\nu}_{bU,k^{}}) \nabla_{\check{\bm{p}}_{b,0}}^{\mathrm{T}} \nu_{bU,k^{}}} \\ & \medmath{- \frac{1}{c}\nabla_{\bm{p}_{U,0}} \nu_{bU,k} \bm{F}_{{{y} }}({y}_{bu,k}| {\eta} ;{\nu}_{bU,k},{\tau}_{bu^{},k^{}}) \bm{\Delta}_{bu^{},k^{}}^{\mathrm{T}}
+ \nabla_{\bm{p}_{U,0}} \nu_{bU,k} \bm{F}_{{{y} }}({y}_{bu,k}| {\eta} ;{\nu}_{bU,k},{\nu}_{b^{}U,k^{}}) \nabla_{\check{\bm{p}}_{b,0}}^{\mathrm{T}}} \nu_{b^{}U,k^{}},
\end{split}
\end{align}
\end{figure*}

\subsubsection{Proof of the FIM related to $\bm{p}_{U,0}$ and $\check{\bm{v}}_{b,0}$}
\label{Appendix_lemma_FIM_3D_position_3D_b_th_velocity_offset}
The FIM relating the $3$D position and the uncertainty in velocity associated with the $b^{\text{th}}$ LEO can be written as (\ref{equ_lemma:FIM_3D_position_bth_velocity_1}), which  can be simplified to (\ref{equ_lemma:FIM_3D_position_bth_velocity_2}). 
Finally, substituting the FIM for the channel parameters
gives us (\ref{equ_lemma:FIM_3D_position_3D_b_th_velocity_offset}).

\begin{figure*}
\begin{align}
\begin{split}
\label{equ_lemma:FIM_3D_position_bth_velocity_1}
\medmath{\bm{F}_{{{y} }}(\bm{y}_{}| \bm{\eta} ;\bm{p}_{U,0},\check{\bm{v}}_{b,0})} 
&= \medmath{\sum_{k^{'},u^{'},k^{},u^{}} \frac{-(k^{'}) \Delta_{t}}{c^2}\bm{\Delta}_{bu,k} \bm{F}_{{{y} }}({y}_{bu,k}| {\eta} ;{\tau}_{bu,k},{\tau}_{bu^{'},k^{'}}) \bm{\Delta}_{bu^{'},k^{'}}^{\mathrm{T}} +  \frac{1}{c^2}\bm{\Delta}_{bu^{},k^{}} \bm{F}_{{{y} }}({y}_{bu,k}| {\eta} ;{\tau}_{bu^{},k^{}},{\nu}_{bU,k^{'}}) \bm{\Delta}_{bU,k^{'}}^{\mathrm{T}}} \\ &-\medmath{\frac{(k^{'})\Delta_{t}}{c}\nabla_{\bm{p}_{U,0}} \nu_{bU,k} \bm{F}_{{{y} }}({y}_{bu,k}| {\eta} ;{\nu}_{bU,k},{\tau}_{bu^{'},k^{'}}) \Delta_{bu^{'},k^{'}}^{\mathrm{T}}
+ \frac{1}{c}\nabla_{\bm{p}_{U,0}} \nu_{bU,k} \bm{F}_{{{y} }}({y}_{bu,k}| {\eta} ;{\nu}_{bU,k},{\nu}_{b^{}U,k^{'}}) \bm{\Delta}_{bU,k^{'}}^{\mathrm{T}}} 
\end{split}
\end{align}
\end{figure*}
\begin{figure*}
\begin{align}
\begin{split}
\label{equ_lemma:FIM_3D_position_bth_velocity_2}
\medmath{\bm{F}_{{{y} }}(\bm{y}_{}| \bm{\eta} ;\bm{p}_{U,0},\check{\bm{v}}_{b,0})} 
&= \medmath{\sum_{k^{},u^{}} \frac{-(k^{}) \Delta_{t}}{c^2}\bm{\Delta}_{bu,k} \bm{F}_{{{y} }}({y}_{bu,k}| {\eta} ;{\tau}_{bu,k},{\tau}_{bu^{},k^{}}) \bm{\Delta}_{bu^{},k^{}}^{\mathrm{T}} +  \frac{1}{c^2}\bm{\Delta}_{bu^{},k^{}} \bm{F}_{{{y} }}({y}_{bu,k}| {\eta} ;{\tau}_{bu^{},k^{}},{\nu}_{bU,k^{}}) \bm{\Delta}_{bU,k^{}}^{\mathrm{T}}} \\ &-\medmath{\frac{(k^{})\Delta_{t}}{c}\nabla_{\bm{p}_{U,0}} \nu_{bU,k} \bm{F}_{{{y} }}({y}_{bu,k}| {\eta} ;{\nu}_{bU,k},{\tau}_{bu^{},k^{}}) \Delta_{bu^{},k^{}}^{\mathrm{T}}
+ \frac{1}{c}\nabla_{\bm{p}_{U,0}} \nu_{bU,k} \bm{F}_{{{y} }}({y}_{bu,k}| {\eta} ;{\nu}_{bU,k},{\nu}_{b^{}U,k^{}}) \bm{\Delta}_{bU,k^{}}^{\mathrm{T}}} 
\end{split}
\end{align}
\end{figure*}

\subsubsection{Proof of the FIM related to $\bm{v}_{U,0}$}
\label{Appendix_lemma_FIM_3D_velocity_3D_velocity}
The FIM of the $3$D velocity of the receiver can be written as (\ref{equ_lemma:FIM_3D_velocity_1}), which can be simplified to (\ref{equ_lemma:FIM_3D_velocity_2}). 
Finally, substituting the FIM for the channel parameters
gives us (\ref{equ_lemma:FIM_3D_velocity_3D_velocity}).

\begin{figure*}
\begin{align}
\begin{split}
\label{equ_lemma:FIM_3D_velocity_1}
\medmath{\bm{F}_{{{y} }}(\bm{y}_{}| \bm{\eta} ;\bm{v}_{U,0},\bm{v}_{U,0})} 
&= \medmath{\sum_{b,k^{'},u^{'},k^{},u^{}} \frac{(k)(k^{'}) \Delta_{t}^2}{c^2}\bm{\Delta}_{bu,k} \bm{F}_{{{y} }}({y}_{bu,k}| {\eta} ;{\tau}_{bu,k},{\tau}_{bu^{'},k^{'}}) \bm{\Delta}_{bu^{'},k^{'}}^{\mathrm{T}} -  \frac{(k) \Delta_{t}}{c^2}\bm{\Delta}_{bu^{},k^{}} \bm{F}_{{{y} }}({y}_{bu,k}| {\eta} ;{\tau}_{bu^{},k^{}},{\nu}_{bU,k^{'}}) \bm{\Delta}_{bU,k^{'}}^{\mathrm{T}}} \\ &-\medmath{\frac{(k^{'}) \Delta_{t}}{c^2}\bm{\Delta}_{bU,k} \bm{F}_{{{y} }}({y}_{bu,k}| {\eta} ;{\nu}_{bU,k},{\tau}_{bu^{'},k^{'}}) \Delta_{bu^{'},k^{'}}^{\mathrm{T}}
+ \frac{1}{c^2}\bm{\Delta}_{bU,k^{}} \bm{F}_{{{y} }}({y}_{bu,k}| {\eta} ;{\nu}_{bU,k},{\nu}_{b^{}U,k^{'}}) \bm{\Delta}_{bU,k^{'}}^{\mathrm{T}}} \\
&+ \medmath{\sum_{q,k^{'},u^{'},k^{},u^{}}\frac{(k)(k^{'})\Delta_t^{2}}{c^2}\bm{\Delta}_{qu,k} \bm{F}_{{{y} }}({y}_{qu,k}| {\eta} ;{\tau}_{qu,k},{\tau}_{qu^{'},k^{'}}) \bm{\Delta}_{qu^{'},k^{'}}^{\mathrm{T}} -  \frac{(k) \Delta_{t}}{c^2}\bm{\Delta}_{qu^{},k^{}} \bm{F}_{{{y} }}({y}_{qu,k}| {\eta} ;{\tau}_{qu^{},k^{}},{\nu}_{qU,k^{'}}) \bm{\Delta}_{qU,k^{'}}^{\mathrm{T}}} \\ &- \medmath{\frac{(k^{'})\Delta_{t}}{c^2}\bm{\Delta}_{qU,k^{}} \bm{F}_{{{y} }}({y}_{qu,k}| {\eta} ;{\nu}_{qU,k},{\tau}_{qu^{'},k^{'}}) \bm{\Delta}_{qu^{'},k^{'}}^{\mathrm{T}}
+  \frac{1}{c^2}\bm{\Delta}_{qU,k^{}} \bm{F}_{{{y} }}({y}_{qu,k}| {\eta} ;{\nu}_{qU,k},{\nu}_{q^{}U,k^{'}}) \bm{\Delta}_{qU,k^{'}}^{\mathrm{T}}},
\end{split}
\end{align}
\end{figure*}
\begin{figure*}
\begin{align}
\begin{split}
\label{equ_lemma:FIM_3D_velocity_2}
\medmath{\bm{F}_{{{y} }}(\bm{y}_{}| \bm{\eta} ;\bm{v}_{U,0},\bm{v}_{U,0})} 
&= \medmath{\sum_{b,k^{},u^{}} \frac{(k^{})^{2} \Delta_{t}^{2}}{c^2}\bm{\Delta}_{bu,k} \bm{F}_{{{y} }}({y}_{bu,k}| {\eta} ;{\tau}_{bu,k},{\tau}_{bu^{},k^{}}) \bm{\Delta}_{bu^{},k^{}}^{\mathrm{T}} -  \frac{(k) \Delta_{t}}{c^2}\bm{\Delta}_{bu^{},k^{}} \bm{F}_{{{y} }}({y}_{bu,k}| {\eta} ;{\tau}_{bu^{},k^{}},{\nu}_{bU,k^{}}) \bm{\Delta}_{bU,k^{}}^{\mathrm{T}}} \\ &-\medmath{\frac{(k^{}) \Delta_{t}}{c^2}\bm{\Delta}_{bU,k} \bm{F}_{{{y} }}({y}_{bu,k}| {\eta} ;{\nu}_{bU,k},{\tau}_{bu^{},k^{}}) \bm{\Delta}_{bu^{},k^{}}^{\mathrm{T}}
+ \frac{1}{c^2}\bm{\Delta}_{bU,k^{}} \bm{F}_{{{y} }}({y}_{bu,k}| {\eta} ;{\nu}_{bU,k},{\nu}_{b^{}U,k^{}}) \bm{\Delta}_{bU,k^{}}^{\mathrm{T}}} \\
&+ \medmath{\sum_{q,k^{},u^{}}\frac{(k^{})^2\Delta_t^{2}}{c^2}\bm{\Delta}_{qu,k} \bm{F}_{{{y} }}({y}_{qu,k}| {\eta} ;{\tau}_{qu,k},{\tau}_{qu^{},k^{}}) \bm{\Delta}_{qu^{},k^{}}^{\mathrm{T}} -  \frac{(k) \Delta_t}{c^2}\bm{\Delta}_{qu^{},k^{}} \bm{F}_{{{y} }}({y}_{qu,k}| {\eta} ;{\tau}_{qu^{},k^{}},{\nu}_{qU,k^{}}) \bm{\Delta}_{qU,k^{}}^{\mathrm{T}}} \\ &- \medmath{\frac{(k) \Delta_t}{c^2}\bm{\Delta}_{qU,k^{}} \bm{F}_{{{y} }}({y}_{qu,k}| {\eta} ;{\nu}_{qU,k},{\tau}_{qu^{},k^{}}) \bm{\Delta}_{qu^{},k^{}}^{\mathrm{T}}
+  \frac{1}{c^2}\bm{\Delta}_{qU,k^{}} \bm{F}_{{{y} }}({y}_{qu,k}| {\eta} ;{\nu}_{qU,k},{\nu}_{q^{}U,k^{}}) \bm{\Delta}_{qU,k^{}}^{\mathrm{T}}},
\end{split}
\end{align}
\end{figure*}

\subsubsection{Proof of the FIM related to $\bm{v}_{U,0}$ and $\bm{\Phi}_{U}$}
\label{Appendix_lemma_FIM_3D_velocity_3D_orientation}
The FIM related to $\bm{v}_{U,0}$ and $\bm{\Phi}_{U}$ can be written as (\ref{equ_lemma:FIM_3D_velocity_3D_orientation_1}), which  can be simplified to (\ref{equ_lemma:FIM_3D_velocity_3D_orientation_2}). 
Finally, substituting the FIM for the channel parameters
gives us (\ref{equ_lemma:FIM_3D_velocity_3D_orientation}).

\begin{figure*}
\begin{align}
\begin{split}
\label{equ_lemma:FIM_3D_velocity_3D_orientation_1}
\medmath{\bm{F}_{{{y} }}(\bm{y}_{}| \bm{\eta} ;\bm{v}_{U,0},\bm{\Phi}_{U})} 
&= \medmath{\sum_{b,k^{'},u^{'},k^{},u^{}} \frac{(k)\Delta_{t}}{c}\bm{\Delta}_{bu,k} \bm{F}_{{{y} }}({y}_{bu,k}| {\eta} ;{\tau}_{bu,k},{\tau}_{bu^{'},k^{'}}) \nabla_{\bm{\Phi}_{U}}^{\mathrm{T}} \tau_{bu^{'},k^{'}}} \\ & + \medmath{\sum_{q,k^{'},u^{'},k^{},u^{}}\frac{(k)\Delta_{t}}{c}\bm{\Delta}_{qu,k} \bm{F}_{{{y} }}({y}_{qu,k}| {\eta} ;{\tau}_{qu,k},{\tau}_{qu^{'},k^{'}}) \nabla_{\bm{\Phi}_{U}}^{\mathrm{T}} \tau_{qu^{'},k^{'}}},
\end{split}
\end{align}
\end{figure*}
\begin{figure*}
\begin{align}
\begin{split}
\label{equ_lemma:FIM_3D_velocity_3D_orientation_2}
\medmath{\bm{F}_{{{y} }}(\bm{y}_{}| \bm{\eta} ;\bm{v}_{U,0},\bm{\Phi}_{U}) 
= \sum_{b,k^{},u^{}} \frac{(k)\Delta_{t}}{c}\bm{\Delta}_{bu,k} \bm{F}_{{{y} }}({y}_{bu,k}| {\eta} ;{\tau}_{bu,k},{\tau}_{bu^{},k^{}}) \nabla_{\bm{\Phi}_{U}}^{\mathrm{T}} \tau_{bu^{},k^{}} + \sum_{q,k^{},u^{}}\frac{(k)\Delta_{t}}{c}\bm{\Delta}_{qu,k} \bm{F}_{{{y} }}({y}_{qu,k}| {\eta} ;{\tau}_{qu,k},{\tau}_{qu^{},k^{}}) \nabla_{\bm{\Phi}_{U}}^{\mathrm{T}} \tau_{qu^{},k^{}}},
\end{split}
\end{align}
\end{figure*}

\subsubsection{Proof of the FIM related to $\bm{v}_{U,0}$ and $\check{\bm{p}}_{b,0}$}
\label{Appendix_lemma_FIM_3D_velocity_3D_b_th_position_offset}
The FIM relating the $3$D velocity and the uncertainty in position associated with the $b^{\text{th}}$ LEO can be written as (\ref{equ_lemma:FIM_3D_velocity_bth_position_1}), which  can be simplified to (\ref{equ_lemma:FIM_3D_velocity_bth_position_2}). 
Finally, substituting the FIM for the channel parameters
gives us (\ref{equ_lemma:FIM_3D_velocity_3D_b_th_position_offset}).

\begin{figure*}
\begin{align}
\begin{split}
\label{equ_lemma:FIM_3D_velocity_bth_position_1}
\medmath{\bm{F}_{{{y} }}(\bm{y}_{}| \bm{\eta} ;\bm{v}_{U,0},\check{\bm{p}}_{b,0})}&= \medmath{\sum_{k^{'},u^{'},k^{},u^{}} \frac{-(k)\Delta_{t}}{c^2}\bm{\Delta}_{bu,k} \bm{F}_{{{y} }}({y}_{bu,k}| {\eta} ;{\tau}_{bu,k},{\tau}_{bu^{'},k^{'}}) \bm{\Delta}_{bu^{'},k^{'}}^{\mathrm{T}} +  \frac{(k)\Delta_{t}}{c}\bm{\Delta}_{bu^{},k^{}} \bm{F}_{{{y} }}({y}_{bu,k}| {\eta} ;{\tau}_{bu^{},k^{}},{\nu}_{bU,k^{'}}) \nabla_{\check{\bm{p}}_{b,0}}^{\mathrm{T}} \nu_{bU,k^{'}}} \\ &+ \medmath{\frac{1}{c^{2}}\bm{\Delta}_{bU,k} \bm{F}_{{{y} }}({y}_{bu,k}| {\eta} ;{\nu}_{bU,k},{\tau}_{bu^{'},k^{'}}) \bm{\Delta}_{bu^{'},k^{'}}^{\mathrm{T}}
- \frac{1}{c}\bm{\Delta}_{bU,k} \bm{F}_{{{y} }}({y}_{bu,k}| {\eta} ;{\nu}_{bU,k},{\nu}_{b^{'}U,k^{'}}) \nabla_{\check{\bm{p}}_{b,0}}^{\mathrm{T}} \nu_{b^{}U,k^{'}}},
\end{split}
\end{align}
\end{figure*}
\begin{figure*}
\begin{align}
\begin{split}
\label{equ_lemma:FIM_3D_velocity_bth_position_2}
\medmath{\bm{F}_{{{y} }}(\bm{y}_{}| \bm{\eta} ;\bm{v}_{U,0},\check{\bm{p}}_{b,0})}
&= \medmath{\sum_{k^{},u^{}} \frac{-(k)\Delta_{t}}{c^2}\bm{\Delta}_{bu,k} \bm{F}_{{{y} }}({y}_{bu,k}| {\eta} ;{\tau}_{bu,k},{\tau}_{bu^{},k^{}}) \bm{\Delta}_{bu^{},k^{}}^{\mathrm{T}} +  \frac{(k)\Delta_{t}}{c}\bm{\Delta}_{bu^{},k^{}} \bm{F}_{{{y} }}({y}_{bu,k}| {\eta} ;{\tau}_{bu^{},k^{}},{\nu}_{bU,k^{}}) \nabla_{\check{\bm{p}}_{b,0}}^{\mathrm{T}} \nu_{bU,k^{}}} \\ & \medmath{+ \frac{1}{c^{2}}\nabla_{\bm{p}_{U,0}} \nu_{bU,k} \bm{F}_{{{y} }}({y}_{bu,k}| {\eta} ;{\nu}_{bU,k},{\tau}_{bu^{},k^{}}) \bm{\Delta}_{bu^{},k^{}}^{\mathrm{T}}
- \frac{1}{c} \bm{\Delta}_{bU,k} \bm{F}_{{{y} }}({y}_{bu,k}| {\eta} ;{\nu}_{bU,k},{\nu}_{b^{}U,k^{}}) \nabla_{\check{\bm{p}}_{b,0}}^{\mathrm{T}}} \nu_{b^{}U,k^{}},
\end{split}
\end{align}
\end{figure*}

\subsubsection{Proof of the FIM related to $\bm{v}_{U,0}$ and $\check{\bm{v}}_{b,0}$}
\label{Appendix_lemma_FIM_3D_velocity_3D_b_th_velocity_offset}
The FIM relating the $3$D velocity and the uncertainty in velocity associated with the $b^{\text{th}}$ LEO can be written as (\ref{equ_lemma:FIM_3D_velocity_bth_velocity_1}), which  can be simplified to (\ref{equ_lemma:FIM_3D_velocity_bth_velocity_2}). 
Finally, substituting the FIM for the channel parameters
gives us (\ref{equ_lemma:FIM_3D_velocity_3D_b_th_velocity_offset}).

\begin{figure*}
\begin{align}
\begin{split}
\label{equ_lemma:FIM_3D_velocity_bth_velocity_1}
\medmath{\bm{F}_{{{y} }}(\bm{y}_{}| \bm{\eta} ;\bm{v}_{U,0},\check{\bm{v}}_{b,0})} 
&= \medmath{\sum_{k^{'},u^{'},k^{},u^{}} \frac{-(k)(k^{'}) \Delta_{t}^{2}}{c^2}\bm{\Delta}_{bu,k} \bm{F}_{{{y} }}({y}_{bu,k}| {\eta} ;{\tau}_{bu,k},{\tau}_{bu^{'},k^{'}}) \bm{\Delta}_{bu^{'},k^{'}}^{\mathrm{T}} +  \frac{(k)\Delta_{t}}{c^2}\bm{\Delta}_{bu^{},k^{}} \bm{F}_{{{y} }}({y}_{bu,k}| {\eta} ;{\tau}_{bu^{},k^{}},{\nu}_{bU,k^{'}}) \bm{\Delta}_{bU,k^{'}}^{\mathrm{T}}} \\ &-\medmath{\frac{(k^{'})\Delta_{t}}{c^2}\bm{\Delta}_{bU^{},k^{}} \bm{F}_{{{y} }}({y}_{bu,k}| {\eta} ;{\nu}_{bU,k},{\tau}_{bu^{'},k^{'}}) \Delta_{bu^{'},k^{'}}^{\mathrm{T}}
- \frac{1}{c^2}\bm{\Delta}_{bU,k^{}} \bm{F}_{{{y} }}({y}_{bu,k}| {\eta} ;{\nu}_{bU,k},{\nu}_{b^{}U,k^{'}}) \bm{\Delta}_{bU,k^{'}}^{\mathrm{T}}} 
\end{split}
\end{align}
\end{figure*}
\begin{figure*}
\begin{align}
\begin{split}
\label{equ_lemma:FIM_3D_velocity_bth_velocity_2}
\medmath{\bm{F}_{{{y} }}(\bm{y}_{}| \bm{\eta} ;\bm{v}_{U,0},\check{\bm{v}}_{b,0})} 
&= \medmath{\sum_{k^{},u^{}} \frac{-(k^{})^{2} \Delta_{t}}{c^2}\bm{\Delta}_{bu,k} \bm{F}_{{{y} }}({y}_{bu,k}| {\eta} ;{\tau}_{bu,k},{\tau}_{bu^{},k^{}}) \bm{\Delta}_{bu^{},k^{}}^{\mathrm{T}} +  \frac{(k) \Delta_{t}}{c^2}\bm{\Delta}_{bu^{},k^{}} \bm{F}_{{{y} }}({y}_{bu,k}| {\eta} ;{\tau}_{bu^{},k^{}},{\nu}_{bU,k^{}}) \bm{\Delta}_{bU,k^{}}^{\mathrm{T}}} \\ &\medmath{-\frac{(k)\Delta_{t}}{c^2}\bm{\Delta}_{bU^{},k^{}}\bm{F}_{{{y} }}({y}_{bu,k}| {\eta} ;{\nu}_{bU,k},{\tau}_{bu^{},k^{}}) \bm{\Delta}_{bu^{},k^{}}^{\mathrm{T}}
- \frac{1}{c^2}\bm{\Delta}_{bU^{},k^{}}  \bm{F}_{{{y} }}({y}_{bu,k}| {\eta} ;{\nu}_{bU,k},{\nu}_{b^{}U,k^{}}) \bm{\Delta}_{bU,k^{}}^{\mathrm{T}}}
\end{split}
\end{align}
\end{figure*}

\subsubsection{Proof of the FIM related to $\bm{\Phi}_{U}$}
\label{Appendix_lemma_FIM_3D_orientation_3D_orientation}
The FIM related to $\bm{\Phi}_{U}$ can be written as (\ref{equ_lemma:FIM_3D_orientation_3D_orientation_1}), which  can be simplified to (\ref{equ_lemma:FIM_3D_orientation_3D_orientation_2}). 
Finally, substituting the FIM for the channel parameters
gives us (\ref{equ_lemma:FIM_3D_orientation_3D_orientation}).

\begin{figure*}
\begin{align}
\begin{split}
\label{equ_lemma:FIM_3D_orientation_3D_orientation_1}
\medmath{\bm{F}_{{{y} }}(\bm{y}_{}| \bm{\eta} ;\bm{\Phi}_{U},\bm{\Phi}_{U})} 
&= \medmath{\sum_{b,k^{'},u^{'},k^{},u^{}} \nabla_{\bm{\Phi}_{U}} \tau_{bu^{},k^{}} \bm{F}_{{{y} }}({y}_{bu,k}| {\eta} ;{\tau}_{bu,k},{\tau}_{bu^{'},k^{'}}) \nabla_{\bm{\Phi}_{U}}^{\mathrm{T}} \tau_{bu^{'},k^{'}}} + \medmath{\sum_{q,k^{'},u^{'},k^{},u^{}}\nabla_{\bm{\Phi}_{U}} \tau_{qu^{},k^{}}\bm{\Delta}_{qu,k} \bm{F}_{{{y} }}({y}_{qu,k}| {\eta} ;{\tau}_{qu,k},{\tau}_{qu^{'},k^{'}}) \nabla_{\bm{\Phi}_{U}}^{\mathrm{T}} \tau_{qu^{'},k^{'}}},
\end{split}
\end{align}
\end{figure*}
\begin{figure*}
\begin{align}
\begin{split}
\label{equ_lemma:FIM_3D_orientation_3D_orientation_2}
\medmath{\bm{F}_{{{y} }}(\bm{y}_{}| \bm{\eta} ;\bm{\Phi}_{U},\bm{\Phi}_{U})} 
&= \medmath{\sum_{b,k^{},u^{}} \nabla_{\bm{\Phi}_{U}} \tau_{bu^{},k^{}} \bm{F}_{{{y} }}({y}_{bu,k}| {\eta} ;{\tau}_{bu,k},{\tau}_{bu^{},k^{}}) \nabla_{\bm{\Phi}_{U}}^{\mathrm{T}} \tau_{bu^{},k^{}}} + \medmath{\sum_{q,k^{},u^{}}\nabla_{\bm{\Phi}_{U}} \tau_{qu^{},k^{}}\bm{\Delta}_{qu,k} \bm{F}_{{{y} }}({y}_{qu,k}| {\eta} ;{\tau}_{qu,k},{\tau}_{qu^{},k^{}}) \nabla_{\bm{\Phi}_{U}}^{\mathrm{T}} \tau_{qu^{},k^{}}},
\end{split}
\end{align}
\end{figure*}

\subsubsection{Proof of the FIM related to $\bm{\Phi}_{U}$ and $\check{\bm{p}}_{b,0}$}
\label{Appendix_lemma_FIM_3D_orientation_3D_b_th_position_offset}
The FIM relating the $3$D orientation and the uncertainty in position associated with the $b^{\text{th}}$ LEO can be written as (\ref{equ_lemma:FIM_3D_orientation_bth_position_1}), which  can be simplified to (\ref{equ_lemma:FIM_3D_orientation_bth_position_2}). 
Finally, substituting the FIM for the channel parameters
gives us (\ref{equ_lemma:FIM_3D_orientation_3D_b_th_position_offset}).

\begin{figure*}
\begin{align}
\begin{split}
\label{equ_lemma:FIM_3D_orientation_bth_position_1}
\medmath{\bm{F}_{{{y} }}(\bm{y}_{}| \bm{\eta} ;\bm{\Phi}_{U},\check{\bm{p}}_{b,0})}&= \medmath{\sum_{k^{'},u^{'},k^{},u^{}} -\frac{1}{c}\nabla_{\bm{\Phi}_{U}} \tau_{bu^{},k^{}} \bm{F}_{{{y} }}({y}_{bu,k}| {\eta} ;{\tau}_{bu,k},{\tau}_{bu^{'},k^{'}}) \bm{\Delta}_{bu^{'},k^{'}}^{\mathrm{T}}},
\end{split}
\end{align}
\end{figure*}
\begin{figure*}
\begin{align}
\begin{split}
\label{equ_lemma:FIM_3D_orientation_bth_position_2}
\medmath{\bm{F}_{{{y} }}(\bm{y}_{}| \bm{\eta} ;\bm{\Phi}_{U},\check{\bm{p}}_{b,0})}&= \medmath{\sum_{k^{},u^{}} -\frac{1}{c}\nabla_{\bm{\Phi}_{U}} \tau_{bu^{},k^{}} \bm{F}_{{{y} }}({y}_{bu,k}| {\eta} ;{\tau}_{bu,k},{\tau}_{bu^{},k^{}}) \bm{\Delta}_{bu^{},k^{}}^{\mathrm{T}}},
\end{split}
\end{align}
\end{figure*}

\subsubsection{Proof of the FIM related to $\bm{\Phi}_{U}$ and $\check{\bm{v}}_{b,0}$}
\label{Appendix_lemma_FIM_3D_orientation_3D_b_th_velocity_offset}
The FIM relating the $3$D orientation and the uncertainty in velocity associated with the $b^{\text{th}}$ LEO can be written as (\ref{equ_lemma:FIM_3D_orientation_bth_velocity_1}), which  can be simplified to (\ref{equ_lemma:FIM_3D_orientation_bth_velocity_2}). 
Finally, substituting the FIM for the channel parameters
gives us (\ref{equ_lemma:FIM_3D_orientation_3D_b_th_velocity_offset}).

\begin{figure*}
\begin{align}
\begin{split}
\label{equ_lemma:FIM_3D_orientation_bth_velocity_1}
\medmath{\bm{F}_{{{y} }}(\bm{y}_{}| \bm{\eta} ;\bm{\Phi}_{U},\check{\bm{v}}_{b,0})}&= \medmath{\sum_{k^{'},u^{'},k^{},u^{}} -\frac{(k^{'}) \Delta_{t}}{c}\nabla_{\bm{\Phi}_{U}} \tau_{bu^{},k^{}} \bm{F}_{{{y} }}({y}_{bu,k}| {\eta} ;{\tau}_{bu,k},{\tau}_{bu^{'},k^{'}}) \bm{\Delta}_{bu^{'},k^{'}}^{\mathrm{T}}},
\end{split}
\end{align}
\end{figure*}
\begin{figure*}
\begin{align}
\begin{split}
\label{equ_lemma:FIM_3D_orientation_bth_velocity_2}
\medmath{\bm{F}_{{{y} }}(\bm{y}_{}| \bm{\eta} ;\bm{\Phi}_{U},\check{\bm{v}}_{b,0})}&= \medmath{\sum_{k^{},u^{}} -\frac{(k^{}) \Delta_{t}}{c}\nabla_{\bm{\Phi}_{U}} \tau_{bu^{},k^{}} \bm{F}_{{{y} }}({y}_{bu,k}| {\eta} ;{\tau}_{bu,k},{\tau}_{bu^{},k^{}}) \bm{\Delta}_{bu^{},k^{}}^{\mathrm{T}}},
\end{split}
\end{align}
\end{figure*}

\subsubsection{Proof of the FIM related to $\check{\bm{p}}_{b,0}$}
\label{Appendix_lemma_FIM_3D_b_th_position_offset_3D_b_th_position_offset}
The FIM of the $3$D position of the receiver can be written as (\ref{equ_lemma:FIM_3D_bth_position_1}), which  can be simplified to (\ref{equ_lemma:FIM_3D_bth_position_2}). 
Finally, substituting the FIM for the channel parameters
gives us (\ref{equ_lemma:FIM_3D_b_th_position_offset_3D_b_th_position_offset}).

\begin{figure*}
\begin{align}
\begin{split}
\label{equ_lemma:FIM_3D_bth_position_1}
\medmath{\bm{F}_{{{y} }}(\bm{y}_{}| \bm{\eta} ;\check{\bm{p}}_{b,0},\check{\bm{p}}_{b,0}) }
&= \medmath{\sum_{k^{'},u^{'},k^{},u^{}} \frac{1}{c^2}\bm{\Delta}_{bu,k} \bm{F}_{{{y} }}({y}_{bu,k}| {\eta} ;{\tau}_{bu,k},{\tau}_{bu^{'},k^{'}}) \bm{\Delta}_{bu^{'},k^{'}}^{\mathrm{T}} -  \frac{1}{c}\bm{\Delta}_{bu^{},k^{}} \bm{F}_{{{y} }}({y}_{bu,k}| {\eta} ;{\tau}_{bu^{},k^{}},{\nu}_{bU,k^{'}}) \nabla_{\check{\bm{p}}_{b,0}}^{\mathrm{T}} \nu_{bU,k^{'}}} \\ & \medmath{- \frac{1}{c}\nabla_{\check{\bm{p}}_{b,0}} \nu_{bU,k} \bm{F}_{{{y} }}({y}_{bu,k}| {\eta} ;{\nu}_{bU,k},{\tau}_{bu^{'},k^{'}}) \bm{\Delta}_{bu^{'},k^{'}}^{\mathrm{T}}
+ \nabla_{\check{\bm{p}}_{b,0}} \nu_{bU,k} \bm{F}_{{{y} }}({y}_{bu,k}| {\eta} ;{\nu}_{bU,k},{\nu}_{b^{}U,k^{'}}) \nabla_{\check{\bm{p}}_{b,0}}^{\mathrm{T}} \nu_{b^{}U,k^{'}}} \\
& \medmath{+ \sum_{q^{'}k^{'},q^{},k^{}}\frac{1}{c^2}\bm{\Delta}_{bq,k} \bm{F}_{{{y} }}({y}_{bq,k}| {\eta} ;{\tau}_{bq,k},{\tau}_{bq^{'},k^{'}}) \bm{\Delta}_{bq^{'},k^{'}}^{\mathrm{T}} -  \frac{1}{c}\bm{\Delta}_{bq^{},k^{}} \bm{F}_{{{y} }}({y}_{bq,k}| {\eta} ;{\tau}_{bq^{},k^{}},{\nu}_{bq^{'},k^{'}}) \nabla_{\check{\bm{p}}_{b,0}}^{\mathrm{T}} \nu_{bq^{'},k^{'}}} \\ & \medmath{- \frac{1}{c} \nabla_{\check{\bm{p}}_{b,0}} \nu_{bq,k} \bm{F}_{{{y} }}({y}_{bq,k}| {\eta} ;{\nu}_{bq,k},{\tau}_{bq^{'},k^{'}}) \bm{\Delta}_{bq^{'},k^{'}}^{\mathrm{T}}
+ \nabla_{\check{\bm{p}}_{b,0}} \nu_{bq,k} \bm{F}_{{{y} }}({y}_{bq,k}| {\eta} ;{\nu}_{bq,k},{\nu}_{b^{}q^{'},k^{'}}) \nabla_{\check{\bm{p}}_{b,0}}^{\mathrm{T}} \nu_{b^{}q^{'},k^{'}}},
\end{split}
\end{align}
\end{figure*}
\begin{figure*}
\begin{align}
\begin{split}
\label{equ_lemma:FIM_3D_bth_position_2}
\medmath{\bm{F}_{{{y} }}(\bm{y}_{}| \bm{\eta} ;\check{\bm{p}}_{b,0},\check{\bm{p}}_{b,0})} 
&= \medmath{\sum_{k^{},u^{}} \frac{1}{c^2}\bm{\Delta}_{bu,k} \bm{F}_{{{y} }}({y}_{bu,k}| {\eta} ;{\tau}_{bu,k},{\tau}_{bu^{},k^{}}) \bm{\Delta}_{bu^{},k^{}}^{\mathrm{T}} -  \frac{1}{c}\bm{\Delta}_{bu^{},k^{}} \bm{F}_{{{y} }}({y}_{bu,k}| {\eta} ;{\tau}_{bu^{},k^{}},{\nu}_{bU,k^{}}) \nabla_{\check{\bm{p}}_{b,0}}^{\mathrm{T}} \nu_{bU,k^{}}} \\ & \medmath{- \frac{1}{c}\nabla_{\check{\bm{p}}_{b,0}} \nu_{bU,k} \bm{F}_{{{y} }}({y}_{bu,k}| {\eta} ;{\nu}_{bU,k},{\tau}_{bu^{},k^{}}) \bm{\Delta}_{bu^{},k^{}}^{\mathrm{T}}
+ \nabla_{\check{\bm{p}}_{b,0}} \nu_{bU,k} \bm{F}_{{{y} }}({y}_{bu,k}| {\eta} ;{\nu}_{bU,k},{\nu}_{b^{}U,k^{}}) \nabla_{\check{\bm{p}}_{b,0}}^{\mathrm{T}} \nu_{b^{}U,k^{}}} \\
& \medmath{+ \sum_{q^{}k^{},q^{},k^{}}\frac{1}{c^2}\bm{\Delta}_{bq,k} \bm{F}_{{{y} }}({y}_{bq,k}| {\eta} ;{\tau}_{bq,k},{\tau}_{bq^{},k^{}}) \bm{\Delta}_{bq^{},k^{}}^{\mathrm{T}} -  \frac{1}{c}\bm{\Delta}_{bq^{},k^{}} \bm{F}_{{{y} }}({y}_{bq,k}| {\eta} ;{\tau}_{bq^{},k^{}},{\nu}_{bq^{},k^{}}) \nabla_{\check{\bm{p}}_{b,0}}^{\mathrm{T}} \nu_{bq^{},k^{}}} \\ & \medmath{- \frac{1}{c} \nabla_{\check{\bm{p}}_{b,0}} \nu_{bq,k} \bm{F}_{{{y} }}({y}_{bq,k}| {\eta} ;{\nu}_{bq,k},{\tau}_{bq^{},k^{}}) \bm{\Delta}_{bq^{},k^{}}^{\mathrm{T}}
+ \nabla_{\check{\bm{p}}_{b,0}} \nu_{bq,k} \bm{F}_{{{y} }}({y}_{bq,k}| {\eta} ;{\nu}_{bq,k},{\nu}_{b^{}q^{},k^{}}) \nabla_{\check{\bm{p}}_{b,0}}^{\mathrm{T}} \nu_{b^{}q^{},k^{}}},
\end{split}
\end{align}
\end{figure*}

\subsubsection{Proof of the FIM related to $\check{\bm{p}}_{b,0}$ and $\check{\bm{v}}_{b,0}$}
\label{Appendix_lemma_FIM_3D_b_th_position_offset_3D_b_th_velocity_offset}
The FIM related to $\check{\bm{p}}_{b,0}$ and $\check{\bm{v}}_{b,0}$ can be written as (\ref{equ_lemma:FIM_3D_bth_position_3D_bth_velocity_1}), which  can be simplified to (\ref{equ_lemma:FIM_3D_bth_position_3D_bth_velocity_2}). 
Finally, substituting the FIM for the channel parameters
gives us (\ref{equ_lemma:FIM_3D_b_th_position_offset_3D_b_th_velocity_offset}).

\begin{figure*}
\begin{align}
\begin{split}
\label{equ_lemma:FIM_3D_bth_position_3D_bth_velocity_1}
\medmath{\bm{F}_{{{y} }}(\bm{y}_{}| \bm{\eta} ;\check{\bm{p}}_{b,0},\check{\bm{v}}_{b,0})} 
&= \medmath{\sum_{k^{'},u^{'},k^{},u^{}} \frac{-(k^{'}) \Delta_{t}}{c^2}\bm{\Delta}_{bu,k} \bm{F}_{{{y} }}({y}_{bu,k}| {\eta} ;{\tau}_{bu,k},{\tau}_{bu^{'},k^{'}}) \bm{\Delta}_{bu^{'},k^{'}}^{\mathrm{T}} -  \frac{1}{c^2}\bm{\Delta}_{bu^{},k^{}} \bm{F}_{{{y} }}({y}_{bu,k}| {\eta} ;{\tau}_{bu^{},k^{}},{\nu}_{bU,k^{'}}) \bm{\Delta}_{bU,k^{'}}^{\mathrm{T}}} \\ &-\medmath{\frac{(k^{'})\Delta_{t}}{c}\nabla_{\check{\bm{p}}_{b,0}} \nu_{bU,k} \bm{F}_{{{y} }}({y}_{bu,k}| {\eta} ;{\nu}_{bU,k},{\tau}_{bu^{'},k^{'}}) \bm{\Delta}_{bu^{'},k^{'}}^{\mathrm{T}}
+ \frac{1}{c}\nabla_{\check{\bm{p}}_{b,0}} \nu_{bU,k} \bm{F}_{{{y} }}({y}_{bu,k}| {\eta} ;{\nu}_{bU,k},{\nu}_{b^{}U,k^{'}}) \bm{\Delta}_{bU,k^{'}}^{\mathrm{T}}}
\\ &+\medmath{\sum_{q^{'},k^{'},q^{},k^{}} \frac{-(k^{'}) \Delta_{t}}{c^2}\bm{\Delta}_{bq,k} \bm{F}_{{{y} }}({y}_{bq,k}| {\eta} ;{\tau}_{bq,k},{\tau}_{bq^{'},k^{'}}) \bm{\Delta}_{bq^{'},k^{'}}^{\mathrm{T}} -  \frac{1}{c^2}\bm{\Delta}_{bq^{},k^{}} \bm{F}_{{{y} }}({y}_{bq,k}| {\eta} ;{\tau}_{bq^{},k^{}},{\nu}_{bq^{'},k^{'}}) \bm{\Delta}_{bq^{'},k^{'}}^{\mathrm{T}}} \\ &-\medmath{\frac{(k^{'})\Delta_{t}}{c}\nabla_{\check{\bm{p}}_{b,0}} \nu_{bq,k} \bm{F}_{{{y} }}({y}_{bq,k}| {\eta} ;{\nu}_{bq,k},{\tau}_{bq^{'},k^{'}}) \bm{\Delta}_{bq^{'},k^{'}}^{\mathrm{T}}
+ \frac{1}{c}\nabla_{\check{\bm{p}}_{b,0}} \nu_{bq,k} \bm{F}_{{{y} }}({y}_{bq,k}| {\eta} ;{\nu}_{bq,k},{\nu}_{bq^{'},k^{'}}) \bm{\Delta}_{bq,k^{'}}^{\mathrm{T}}}
\end{split}
\end{align}
\end{figure*}
\begin{figure*}
\begin{align}
\begin{split}
\label{equ_lemma:FIM_3D_bth_position_3D_bth_velocity_2}
\medmath{\bm{F}_{{{y} }}(\bm{y}_{}| \bm{\eta} ;\check{\bm{p}}_{b,0},\check{\bm{v}}_{b,0})} 
&= \medmath{\sum_{k^{},u^{}} \frac{-(k^{}) \Delta_{t}}{c^2}\bm{\Delta}_{bu,k} \bm{F}_{{{y} }}({y}_{bu,k}| {\eta} ;{\tau}_{bu,k},{\tau}_{bu^{},k^{}}) \bm{\Delta}_{bu^{},k^{}}^{\mathrm{T}} -  \frac{1}{c^2}\bm{\Delta}_{bu^{},k^{}} \bm{F}_{{{y} }}({y}_{bu,k}| {\eta} ;{\tau}_{bu^{},k^{}},{\nu}_{bU,k^{}}) \bm{\Delta}_{bU,k^{}}^{\mathrm{T}}} \\ &-\medmath{\frac{(k^{})\Delta_{t}}{c}\nabla_{\check{\bm{p}}_{b,0}} \nu_{bU,k} \bm{F}_{{{y} }}({y}_{bu,k}| {\eta} ;{\nu}_{bU,k},{\tau}_{bu^{},k^{}}) \bm{\Delta}_{bu^{},k^{}}^{\mathrm{T}}
+ \frac{1}{c}\nabla_{\check{\bm{p}}_{b,0}} \nu_{bU,k} \bm{F}_{{{y} }}({y}_{bu,k}| {\eta} ;{\nu}_{bU,k},{\nu}_{b^{}U,k^{}}) \bm{\Delta}_{bU,k^{}}^{\mathrm{T}}}
\\ &+\medmath{\sum_{q^{},k^{}} \frac{-(k^{}) \Delta_{t}}{c^2}\bm{\Delta}_{bq,k} \bm{F}_{{{y} }}({y}_{bq,k}| {\eta} ;{\tau}_{bq,k},{\tau}_{bq^{},k^{}}) \bm{\Delta}_{bq^{},k^{}}^{\mathrm{T}} -  \frac{1}{c^2}\bm{\Delta}_{bq^{},k^{}} \bm{F}_{{{y} }}({y}_{bq,k}| {\eta} ;{\tau}_{bq^{},k^{}},{\nu}_{bq^{},k^{}}) \bm{\Delta}_{bq^{},k^{}}^{\mathrm{T}}} \\ &-\medmath{\frac{(k^{})\Delta_{t}}{c}\nabla_{\check{\bm{p}}_{b,0}} \nu_{bq,k} \bm{F}_{{{y} }}({y}_{bq,k}| {\eta} ;{\nu}_{bq,k},{\tau}_{bq^{},k^{}}) \bm{\Delta}_{bq^{},k^{}}^{\mathrm{T}}
+ \frac{1}{c}\nabla_{\check{\bm{p}}_{b,0}} \nu_{bq,k} \bm{F}_{{{y} }}({y}_{bq,k}| {\eta} ;{\nu}_{bq,k},{\nu}_{bq^{},k^{}}) \bm{\Delta}_{bq,k^{}}^{\mathrm{T}}}
\end{split}
\end{align}
\end{figure*}

\subsubsection{Proof of the FIM related to $\check{\bm{v}}_{b,0}$}
\label{Appendix_lemma_FIM_3D_b_th_velocity_offset_3D_b_th_velocity_offset}
The FIM related to $\check{\bm{p}}_{b,0}$ and $\check{\bm{v}}_{b,0}$ can be written as (\ref{equ_lemma:FIM_3D_bth_velocity_3D_bth_velocity_1}), which  can be simplified to (\ref{equ_lemma:FIM_3D_bth_velocity_3D_bth_velocity_2}). 
Finally, substituting the FIM for the channel parameters
gives us (\ref{equ_lemma:FIM_3D_b_th_velocity_offset_3D_b_th_velocity_offset}).

\begin{figure*}
\begin{align}
\begin{split}
\label{equ_lemma:FIM_3D_bth_velocity_3D_bth_velocity_1}
\medmath{\bm{F}_{{{y} }}(\bm{y}_{}| \bm{\eta} ;\check{\bm{v}}_{b,0},\check{\bm{v}}_{b,0})} 
&= \medmath{\sum_{k^{'},u^{'},k^{},u^{}} \frac{(k^{})(k^{'}) \Delta_{t}^{2}}{c^2}\bm{\Delta}_{bu,k} \bm{F}_{{{y} }}({y}_{bu,k}| {\eta} ;{\tau}_{bu,k},{\tau}_{bu^{'},k^{'}}) \bm{\Delta}_{bu^{'},k^{'}}^{\mathrm{T}} -  \frac{(k^{}) \Delta_{t}^{}}{c^2}\bm{\Delta}_{bu^{},k^{}} \bm{F}_{{{y} }}({y}_{bu,k}| {\eta} ;{\tau}_{bu^{},k^{}},{\nu}_{bU,k^{'}}) \bm{\Delta}_{bU,k^{'}}^{\mathrm{T}}} \\ &-\medmath{\frac{(k^{'})\Delta_{t}}{c^2}\bm{\Delta}_{bU,k} \bm{F}_{{{y} }}({y}_{bu,k}| {\eta} ;{\nu}_{bU,k},{\tau}_{bu^{'},k^{'}}) \bm{\Delta}_{bu^{'},k^{'}}^{\mathrm{T}}
+ \frac{1}{c^2}\bm{\Delta}_{bU,k} \bm{F}_{{{y} }}({y}_{bu,k}| {\eta} ;{\nu}_{bU,k},{\nu}_{b^{}U,k^{'}}) \bm{\Delta}_{bU,k^{'}}^{\mathrm{T}}}
\\ &+\medmath{\sum_{q^{'},k^{'},q^{},k^{}}\frac{(k^{})(k^{'}) \Delta_{t}^{2}}{c^2}\bm{\Delta}_{bq,k} \bm{F}_{{{y} }}({y}_{bq,k}| {\eta} ;{\tau}_{bq,k},{\tau}_{bq^{'},k^{'}}) \bm{\Delta}_{bq^{'},k^{'}}^{\mathrm{T}} -  \frac{(k) \Delta_{t}}{c^2}\bm{\Delta}_{bq^{},k^{}} \bm{F}_{{{y} }}({y}_{bq,k}| {\eta} ;{\tau}_{bq^{},k^{}},{\nu}_{bq^{'},k^{'}}) \bm{\Delta}_{bq^{'},k^{'}}^{\mathrm{T}}} \\ &-\medmath{\frac{(k^{'})\Delta_{t}}{c^2}\bm{\Delta}_{bq,k} \bm{F}_{{{y} }}({y}_{bq,k}| {\eta} ;{\nu}_{bq,k},{\tau}_{bq^{'},k^{'}}) \bm{\Delta}_{bq^{'},k^{'}}^{\mathrm{T}}
+ \frac{1}{c^2}\bm{\Delta}_{bq,k} \bm{F}_{{{y} }}({y}_{bq,k}| {\eta} ;{\nu}_{bq,k},{\nu}_{bq^{'},k^{'}}) \bm{\Delta}_{bq,k^{'}}^{\mathrm{T}}}
\end{split}
\end{align}
\end{figure*}
\begin{figure*}
\begin{align}
\begin{split}
\label{equ_lemma:FIM_3D_bth_velocity_3D_bth_velocity_2}
\medmath{\bm{F}_{{{y} }}(\bm{y}_{}| \bm{\eta} ;\check{\bm{v}}_{b,0},\check{\bm{v}}_{b,0})} 
&= \medmath{\sum_{k^{},u^{}} \frac{(k^{})^2 \Delta_{t}^{2}}{c^2}\bm{\Delta}_{bu,k} \bm{F}_{{{y} }}({y}_{bu,k}| {\eta} ;{\tau}_{bu,k},{\tau}_{bu^{},k^{}}) \bm{\Delta}_{bu^{},k^{}}^{\mathrm{T}} -  \frac{(k^{}) \Delta_{t}^{}}{c^2}\bm{\Delta}_{bu^{},k^{}} \bm{F}_{{{y} }}({y}_{bu,k}| {\eta} ;{\tau}_{bu^{},k^{}},{\nu}_{bU,k^{}}) \bm{\Delta}_{bU,k^{}}^{\mathrm{T}}} \\ &-\medmath{\frac{(k^{})\Delta_{t}}{c^2}\bm{\Delta}_{bU,k} \bm{F}_{{{y} }}({y}_{bu,k}| {\eta} ;{\nu}_{bU,k},{\tau}_{bu^{},k^{}}) \bm{\Delta}_{bu^{},k^{}}^{\mathrm{T}}
+ \frac{1}{c^2}\bm{\Delta}_{bU,k} \bm{F}_{{{y} }}({y}_{bu,k}| {\eta} ;{\nu}_{bU,k},{\nu}_{b^{}U,k^{}}) \bm{\Delta}_{bU,k^{}}^{\mathrm{T}}}
\\ &+\medmath{\sum_{q^{},k^{}}\frac{(k^{})^{2} \Delta_{t}^{2}}{c^2}\bm{\Delta}_{bq,k} \bm{F}_{{{y} }}({y}_{bq,k}| {\eta} ;{\tau}_{bq,k},{\tau}_{bq^{},k^{}}) \bm{\Delta}_{bq^{},k^{}}^{\mathrm{T}} -  \frac{(k) \Delta_{t}}{c^2}\bm{\Delta}_{bq^{},k^{}} \bm{F}_{{{y} }}({y}_{bq,k}| {\eta} ;{\tau}_{bq^{},k^{}},{\nu}_{bq^{},k^{}}) \bm{\Delta}_{bq^{},k^{}}^{\mathrm{T}}} \\ &-\medmath{\frac{(k^{})\Delta_{t}}{c^2}\bm{\Delta}_{bq,k} \bm{F}_{{{y} }}({y}_{bq,k}| {\eta} ;{\nu}_{bq,k},{\tau}_{bq^{},k^{}}) \bm{\Delta}_{bq^{},k^{}}^{\mathrm{T}}
+ \frac{1}{c^2}\bm{\Delta}_{bq,k} \bm{F}_{{{y} }}({y}_{bq,k}| {\eta} ;{\nu}_{bq,k},{\nu}_{bq^{},k^{}}) \bm{\Delta}_{bq,k^{}}^{\mathrm{T}}}
\end{split}
\end{align}
\end{figure*}

\subsection{Proof for the elements in $\mathbf{J}_{ \bm{\bm{y}}; \bm{\kappa}_1}^{nu}$}
Proofs for the elements in $\mathbf{J}_{ \bm{\bm{y}}; \bm{\kappa}_1}^{nu}$ are presented in this section.
\subsubsection{Proof of Lemma \ref{lemma:information_loss_FIM_3D_position}}
\label{Appendix_lemma_information_loss_FIM_3D_position}
Using the definition of the information loss terms in Definition \ref{definition_EFIM}, we get the first equality (\ref{equ_lemma:information_loss_FIM_3D_position_1}) and after substituting the appropriate FIM for channel parameters, we get the second equality.

\begin{figure*}
\begin{align}
\begin{split}
\label{equ_lemma:information_loss_FIM_3D_position_1}
&\medmath{\bm{G}_{{{y} }}(\bm{y}_{}| \bm{\eta} ;\bm{p}_{U,0},\bm{p}_{U,0})  }= \\   &\medmath{\sum_{b}\Bigg[
 \Bigg[\sum_{k^{},u^{}} \bm{F}_{{{y} }}({y}_{bu,k}| {\eta} ;{\delta}_{bU},{\delta}_{bU}) \Bigg]^{-1} \norm{ \sum_{k^{},u^{}}   \bm{F}_{{{y} }}({y}_{bu,k}| {\eta} ;{\delta}_{bU},{\tau}_{bu,k}) \nabla_{\bm{p}_{U,0}}^{\mathrm{T}} \tau_{bu,k^{}}  }^{2}} \\ &+  \medmath{\Bigg[ \sum_{k,u}  \bm{F}_{{{y} }}({y}_{bu,k}| {\eta} ;{\epsilon}_{bU},{\epsilon}_{bU})\Bigg]^{-1}\norm{\sum_{k^{},u^{}}   \bm{F}_{{{y} }}({y}_{bu,k}| {\eta} ;{\epsilon}_{bU},{\nu}_{bU,k}) \nabla_{\bm{p}_{U,0}}^{\mathrm{T}} \nu_{bU,k^{}} }^2 }
 \\ &\medmath{+\Bigg[
 \Bigg[\sum_{q,k^{},u^{}} \bm{F}_{{{y} }}({y}_{qu,k}| {\eta} ;{\delta}_{QU},{\delta}_{QU}) \Bigg]^{-1} \norm{ \sum_{q,k^{},u^{}}   \bm{F}_{{{y} }}({y}_{qu,k}| {\eta} ;{\delta}_{QU},{\tau}_{qu,k}) \nabla_{\bm{p}_{U,0}}^{\mathrm{T}} \tau_{qu,k^{}}  }^{2}} \\ & \medmath{+  \Bigg[ \sum_{q,k,u}  \bm{F}_{{{y} }}({y}_{qu,k}| {\eta} ;{\epsilon}_{QU},{\epsilon}_{QU})\Bigg]^{-1}\norm{\sum_{q,k^{},u^{}}   \bm{F}_{{{y} }}({y}_{qu,k}| {\eta} ;{\epsilon}_{QU},{\nu}_{qU,k}) \nabla_{\bm{p}_{U,0}}^{\mathrm{T}} \nu_{qU,k^{}} }^2}
 \\ &
 =\medmath{\sum_{b}   \norm{\sum_{k^{},u^{}} \underset{bu^{},k^{}}{\operatorname{SNR}}\bm{\Delta}_{bu^{},k^{}}^{\mathrm{T}} \frac{ \omega_{bU,k}}{c} }}^{2}    \medmath{   \left(\sum_{u,k} \underset{bu,k}{\operatorname{SNR}} \omega_{bU,k}\right)^{\mathrm{-1}}  }  +
 \medmath{   \norm{\sum_{q,u^{},k^{}} \underset{qu^{},k^{}}{\operatorname{SNR}}\bm{\Delta}_{qu^{},k^{}}^{\mathrm{T}} \frac{ \omega_{qU,k}}{c} }}^{2}    \medmath{   \left(\sum_{q,u,k} \underset{qu,k}{\operatorname{SNR}} \omega_{qU,k}\right)^{\mathrm{-1}}  } \\& +  
\medmath{\sum_{b}\norm{{\sum_{u^{},k^{} } \underset{bu^{},k^{}}{\operatorname{SNR}} \; \; \nabla_{\bm{p}_{U,0}}^{\mathrm{T}} \nu_{bU,k^{}}   }  \frac{(f_{c}^{}) (\alpha_{obu,k^{}}^{2})}{2} }^{2}  \left(\sum_{u,k} \frac{\underset{bu,k}{\operatorname{SNR}}  \alpha_{obu,k}^{2}}{2}\right)^{-1}}  + \medmath{\norm{{\sum_{q^{},u,k^{} } \underset{q^{},u,k^{}}{\operatorname{SNR}} \; \; \nabla_{\bm{p}_{U,0}}^{\mathrm{T}} \nu_{qU,k^{}}   }  \frac{(f_{c}^{}) (\alpha_{oqu,k^{}}^{2})}{2} }^{2}  \left(\sum_{q,u,k} \frac{\underset{qu,k}{\operatorname{SNR}}  \alpha_{oqu,k}^{2}}{2}\right)^{-1}}
\end{split}
\end{align}
\end{figure*}

\subsubsection{Proof of Lemma \ref{lemma:information_loss_FIM_3D_position_3D_veloctiy}}
\label{Appendix_lemma_information_loss_FIM_3D_position_3D_velocity}
Using the definition of the information loss terms in Definition \ref{definition_EFIM}, we get the first equality (\ref{equ_lemma:information_loss_FIM_3D_position_3D_velocity_1}) and after substituting the appropriate FIM for channel parameters, we get the second equality.

\begin{figure*}
\begin{align}
\begin{split}
\label{equ_lemma:information_loss_FIM_3D_position_3D_velocity_1}
&\medmath{\bm{G}_{{{y} }}(\bm{y}_{}| \bm{\eta} ;\bm{p}_{U,0},\bm{v}_{U,0})  }= \\   &\medmath{\sum_{b}\Bigg[
 \Bigg[\sum_{k^{},u^{}} \bm{F}_{{{y} }}({y}_{bu,k}| {\eta} ;{\delta}_{bU},{\delta}_{bU}) \Bigg]^{-1} \sum_{k^{'},u^{'}
 } \sum_{k^{},u^{}}   \bm{F}_{{{y} }}({y}_{bu,k}| {\eta} ;{\delta}_{bU},{\tau}_{bu,k}) \bm{F}_{{{y} }}({y}_{bu,k}| {\eta} ;{\delta}_{bU},{\tau}_{bu^{'},k^{'}}) \nabla_{\bm{p}_{U,0}} \tau_{bu,k^{}} \nabla_{\bm{v}_{U,0}}^{\mathrm{T}} \tau_{bu^{'},k^{'}}  } \\ &+  \medmath{\Bigg[ \sum_{k,u}  \bm{F}_{{{y} }}({y}_{bu,k}| {\eta} ;{\epsilon}_{bU},{\epsilon}_{bU})\Bigg]^{-1}\sum_{k^{'},u^{'}}\sum_{k^{},u^{}}   \bm{F}_{{{y} }}({y}_{bu,k}| {\eta} ;{\epsilon}_{bU},{\nu}_{bU,k}) \bm{F}_{{{y} }}({y}_{bu,k}| {\eta} ;{\epsilon}_{bU},{\nu}_{bU,k^{'}}) \nabla_{\bm{p}_{U,0}} \nu_{bU,k^{}} \nabla_{\bm{v}_{U,0}}^{\mathrm{T}} \nu_{bU,k^{'}} }
 \\ &\medmath{+\Bigg[
 \Bigg[\sum_{q,k^{},u^{}} \bm{F}_{{{y} }}({y}_{qu,k}| {\eta} ;{\delta}_{QU},{\delta}_{QU}) \Bigg]^{-1}  \sum_{q^{'}u^{'},k^{'}}\sum_{q,k^{},u^{}}   \bm{F}_{{{y} }}({y}_{qu,k}| {\eta} ;{\delta}_{QU},{\tau}_{qu,k}) \bm{F}_{{{y} }}({y}_{qu,k}| {\eta} ;{\delta}_{QU},{\tau}_{q^{'}u^{'},k^{'}}) \nabla_{\bm{p}_{U,0}} \tau_{qu,k^{}} \nabla_{\bm{v}_{U,0}}^{\mathrm{T}} \tau_{q^{'}u^{'},k^{'}}  } \\ & \medmath{+  \Bigg[ \sum_{q,k,u}  \bm{F}_{{{y} }}({y}_{qu,k}| {\eta} ;{\epsilon}_{QU},{\epsilon}_{QU})\Bigg]^{-1}\sum_{q^{'}u^{'},k^{'}}\sum_{q,k^{},u^{}}   \bm{F}_{{{y} }}({y}_{qu,k}| {\eta} ;{\epsilon}_{QU},{\nu}_{qU,k}) \bm{F}_{{{y} }}({y}_{qu,k}| {\eta} ;{\epsilon}_{QU},{\nu}_{q^{'}U,k^{'}}) \nabla_{\bm{p}_{U,0}} \nu_{qU,k^{}} \nabla_{\bm{v}_{U,0}}^{\mathrm{T}} \nu_{q^{'}U,k^{'}}}
 \\ &
= \medmath{\frac{\Delta_{t}}{c^2} \sum_{b,k^{},u^{}k^{'},u^{'}} \underset{bu^{},k^{}}{\operatorname{SNR}} \underset{bu^{'},k^{'}}{\operatorname{SNR}} } \medmath{   \bm{\Delta}_{bu^{},k^{}}  (k^{'}) {\bm{\Delta}_{bu^{'},k^{'}}^{\mathrm{T}}} \omega_{bU,k} \omega_{bU,k^{'}}}    \medmath{   \left(\sum_{u,k} \underset{bu,k}{\operatorname{SNR}} \omega_{bU,k}\right)^{\mathrm{-1}}  + \frac{\Delta_{t}}{c^2} \sum_{q,q^{'},u^{},k^{}u^{'},k^{'}} \underset{q^{},u,k^{}}{\operatorname{SNR}} \underset{q^{'},u^{'},k^{'}}{\operatorname{SNR}} } \medmath{   \bm{\Delta}_{q^{}u^{},k^{}}}  \\& \medmath{(k^{'})} \medmath{{\bm{\Delta}_{q^{'}u^{'},k^{'}}^{\mathrm{T}}} \omega_{qU,k} \omega_{q^{'}U,k^{'}}}    \medmath{   \left(\sum_{qu,k} \underset{qu,k}{\operatorname{SNR}} \omega_{qU,k}\right)^{\mathrm{-1}} - \frac{1}{c} \sum_{b,k^{},u^{}k^{'},u^{'}}}  \medmath{  \underset{bu^{},k^{}}{\operatorname{SNR}} \underset{bu^{'},k^{'}}{\operatorname{SNR}}  }  \medmath{{ \nabla_{\bm{p}_{U,0}} \nu_{bU,k^{}} \bm{\Delta}_{bU,k^{'}}^{\mathrm{T}}   }  \frac{(f_{c}^{2}) (\alpha_{obu,k^{}}^{2}\alpha_{obu^{'},k^{'}}^{2})}{4}   \left(\sum_{u,k} \frac{\underset{b^{}u,k^{}}{\operatorname{SNR}}  \alpha_{obu,k}^{2}}{2}\right)^{-1}}  \\& - \frac{1}{c}\sum_{q^{},u^{},q^{'},u^{'},k^{},k^{'}} \medmath{  \underset{qu^{},k^{}}{\operatorname{SNR}} \underset{q^{'}u^{'},k^{'}}{\operatorname{SNR}}  }  \medmath{{ \nabla_{\bm{p}_{U,0}} \nu_{qU,k^{}} \bm{\Delta}_{q^{'}U^{},k^{'}}^{\mathrm{T}}   }  \frac{(f_{c}^{2}) (\alpha_{oqu,k^{}}^{2}\alpha_{oq^{'}u^{'},k^{'}}^{2})}{4}  }  \medmath{ \left(\sum_{q,u,k} \frac{\underset{q^{}u,k^{}}{\operatorname{SNR}}  \alpha_{oqu,k}^{2}}{2}\right)^{-1}}
\end{split}
\end{align}
\end{figure*}

\subsubsection{Proof of Lemma \ref{lemma:information_loss_FIM_3D_position_3D_orientation}}
\label{Appendix_lemma_information_loss_FIM_3D_position_3D_orientation}
Using the definition of the information loss terms in Definition \ref{definition_EFIM}, we get the first equality (\ref{equ_lemma:information_loss_FIM_3D_position_3D_orientation_1}) and after substituting the appropriate FIM for channel parameters, we get the second equality.

\begin{figure*}
\begin{align}
\begin{split}
\label{equ_lemma:information_loss_FIM_3D_position_3D_orientation_1}
&\medmath{\bm{G}_{{{y} }}(\bm{y}_{}| \bm{\eta} ;\bm{p}_{U,0},\bm{\Phi}_{U})  }= \\   &\medmath{\sum_{b}\Bigg[
 \Bigg[\sum_{k^{},u^{}} \bm{F}_{{{y} }}({y}_{bu,k}| {\eta} ;{\delta}_{bU},{\delta}_{bU}) \Bigg]^{-1} \sum_{k^{'},u^{'}
 } \sum_{k^{},u^{}}   \bm{F}_{{{y} }}({y}_{bu,k}| {\eta} ;{\delta}_{bU},{\tau}_{bu,k}) \bm{F}_{{{y} }}({y}_{bu,k}| {\eta} ;{\delta}_{bU},{\tau}_{bu^{'},k^{'}}) \nabla_{\bm{p}_{U,0}} \tau_{bu,k^{}} \nabla_{\bm{\Phi}_{U}}^{\mathrm{T}} \tau_{bu^{'},k^{'}}  } 
 \\ &\medmath{+\Bigg[
 \Bigg[\sum_{q,k^{},u^{}} \bm{F}_{{{y} }}({y}_{qu,k}| {\eta} ;{\delta}_{QU},{\delta}_{QU}) \Bigg]^{-1}  \sum_{q^{'},k^{'},u^{'}}\sum_{q,k^{},u^{}}   \bm{F}_{{{y} }}({y}_{qu,k}| {\eta} ;{\delta}_{QU},{\tau}_{qu,k}) \bm{F}_{{{y} }}({y}_{qu,k}| {\eta} ;{\delta}_{QU},{\tau}_{q^{'}u^{'},k^{'}}) \nabla_{\bm{p}_{U,0}} \tau_{qu,k^{}} \nabla_{\bm{\Phi}_{U}}^{\mathrm{T}} \tau_{q^{'}u^{'},k^{'}}  } \\&
= \medmath{\frac{1}{c}\sum_{b,k^{},u^{},k^{'},u^{'}} \underset{bu^{},k^{}}{\operatorname{SNR}} \underset{bu^{'},k^{'}}{\operatorname{SNR}} } \medmath{   \bm{\Delta}_{bu^{},k^{}}  \nabla_{\bm{\Phi}_{U}}^{\mathrm{T}} \tau_{bu^{'},k^{'}}  \omega_{bU,k} \omega_{bU,k^{'}}}    \medmath{   \left(\sum_{u,k} \underset{bu,k}{\operatorname{SNR}} \omega_{bU,k}\right)^{\mathrm{-1}}} \\ &+  \medmath{\frac{1}{c}\sum_{q,q^{'},u^{},k^{}u^{'},k^{'}} \underset{q^{},u,k^{}}{\operatorname{SNR}} \underset{q^{'},u^{'},k^{'}}{\operatorname{SNR}} } \medmath{   \bm{\Delta}_{q^{}u^{},k^{}}   \nabla_{\bm{\Phi}_{U}}^{\mathrm{T}} \tau_{q^{'}u^{'},k^{'}}} \medmath{ \omega_{qU,k} \omega_{q^{'}U,k^{'}}}    \medmath{   \left(\sum_{qu,k} \underset{qu,k}{\operatorname{SNR}} \omega_{qU,k}\right)^{\mathrm{-1}}} 
\end{split}
\end{align}
\end{figure*}

\subsubsection{Proof of Lemma \ref{lemma:information_loss_FIM_3D_position_3D_bth_position_uncertainty}}
\label{Appendix_lemma_information_loss_FIM_3D_position_3D_bth_position_uncertainty}
Using the definition of the information loss terms in Definition \ref{definition_EFIM}, we get the first equality (\ref{equ_lemma:information_loss_FIM_3D_position_3D_bth_position_uncertainty_1}) and after substituting the appropriate FIM for channel parameters, we get the second equality.

\begin{figure*}
\begin{align}
\begin{split}
\label{equ_lemma:information_loss_FIM_3D_position_3D_bth_position_uncertainty_1}
&\medmath{\bm{G}_{{{y} }}(\bm{y}_{}| \bm{\eta} ;\bm{p}_{U,0},\check{\bm{p}}_{b,0})  }= \\   &\medmath{\Bigg[
 \Bigg[\sum_{k^{},u^{}} \bm{F}_{{{y} }}({y}_{bu,k}| {\eta} ;{\delta}_{bU},{\delta}_{bU}) \Bigg]^{-1} \sum_{k^{'},u^{'}
 } \sum_{k^{},u^{}}   \bm{F}_{{{y} }}({y}_{bu,k}| {\eta} ;{\delta}_{bU},{\tau}_{bu,k}) \bm{F}_{{{y} }}({y}_{bu,k}| {\eta} ;{\delta}_{bU},{\tau}_{bu^{'},k^{'}}) \nabla_{\bm{p}_{U,0}} \tau_{bu,k^{}} \nabla_{\check{\bm{p}}_{b,0}}^{\mathrm{T}} \tau_{bu^{'},k^{'}}  } \\ &+  \medmath{\Bigg[ \sum_{k,u}  \bm{F}_{{{y} }}({y}_{bu,k}| {\eta} ;{\epsilon}_{bU},{\epsilon}_{bU})\Bigg]^{-1}\sum_{k^{'},u^{'}}\sum_{k^{},u^{}}   \bm{F}_{{{y} }}({y}_{bu,k}| {\eta} ;{\epsilon}_{bU},{\nu}_{bU,k}) \bm{F}_{{{y} }}({y}_{bu,k}| {\eta} ;{\epsilon}_{bU},{\nu}_{bU,k^{'}}) \nabla_{\bm{p}_{U,0}} \nu_{bU,k^{}} \nabla_{\check{\bm{p}}_{b,0}}^{\mathrm{T}} \nu_{bU,k^{'}} }
 \\ &
= \medmath{ \frac{-1}{c^2} \sum_{k^{},u^{},k^{'},u^{'}} \underset{bu^{},k^{}}{\operatorname{SNR}} \underset{bu^{'},k^{'}}{\operatorname{SNR}} } \medmath{   \bm{\Delta}_{bu^{},k^{}}  {\bm{\Delta}_{bu^{'},k^{'}}^{\mathrm{T}}} \omega_{bU,k} \omega_{bU,k^{'}}}    \medmath{   \left(\sum_{u,k} \underset{bu,k}{\operatorname{SNR}} \omega_{bU,k}\right)^{\mathrm{-1}}}  \\ &+ \sum_{k^{},u^{},k^{'},u^{'}}  \medmath{  \underset{bu^{},k^{}}{\operatorname{SNR}} \underset{bu^{'},k^{'}}{\operatorname{SNR}}  }  \medmath{{ \nabla_{\bm{p}_{U,0}} \nu_{bU,k^{}} \nabla_{\check{\bm{p}}_{b,0}}^{\mathrm{T}} \nu_{bU,k^{'}}   }  \frac{(f_{c}^{2}) (\alpha_{obu,k^{}}^{2}\alpha_{obu^{'},k^{'}}^{2})}{4}   \left(\sum_{u,k} \frac{\underset{b^{}u,k^{}}{\operatorname{SNR}}  \alpha_{obu,k}^{2}}{2}\right)^{-1}}
\end{split}
\end{align}
\end{figure*}

\subsubsection{Proof of Lemma \ref{lemma:information_loss_FIM_3D_position_3D_bth_velocity_uncertainty}}
\label{Appendix_lemma_information_loss_FIM_3D_position_3D_bth_velocity_uncertainty}
Using the definition of the information loss terms in Definition \ref{definition_EFIM}, we get the first equality (\ref{equ_lemma:information_loss_FIM_3D_position_3D_bth_velocity_uncertainty_1}), and after substituting the appropriate FIM for channel parameters, we get the second equality.

\begin{figure*}
\begin{align}
\begin{split}
\label{equ_lemma:information_loss_FIM_3D_position_3D_bth_velocity_uncertainty_1}
&\medmath{\bm{G}_{{{y} }}(\bm{y}_{}| \bm{\eta} ;\bm{p}_{U,0},\check{\bm{v}}_{b,0})  }= \\   &\medmath{\Bigg[
 \Bigg[\sum_{k^{},u^{}} \bm{F}_{{{y} }}({y}_{bu,k}| {\eta} ;{\delta}_{bU},{\delta}_{bU}) \Bigg]^{-1} \sum_{k^{'},u^{'}
 } \sum_{k^{},u^{}}   \bm{F}_{{{y} }}({y}_{bu,k}| {\eta} ;{\delta}_{bU},{\tau}_{bu,k}) \bm{F}_{{{y} }}({y}_{bu,k}| {\eta} ;{\delta}_{bU},{\tau}_{bu^{'},k^{'}}) \nabla_{\bm{p}_{U,0}} \tau_{bu,k^{}} \nabla_{\check{\bm{v}}_{b,0}}^{\mathrm{T}} \tau_{bu^{'},k^{'}}  } \\ &+  \medmath{\Bigg[ \sum_{k,u}  \bm{F}_{{{y} }}({y}_{bu,k}| {\eta} ;{\epsilon}_{bU},{\epsilon}_{bU})\Bigg]^{-1}\sum_{k^{'},u^{'}}\sum_{k^{},u^{}}   \bm{F}_{{{y} }}({y}_{bu,k}| {\eta} ;{\epsilon}_{bU},{\nu}_{bU,k}) \bm{F}_{{{y} }}({y}_{bu,k}| {\eta} ;{\epsilon}_{bU},{\nu}_{bU,k^{'}}) \nabla_{\bm{p}_{U,0}} \nu_{bU,k^{}} \nabla_{\check{\bm{v}}_{b,0}}^{\mathrm{T}} \nu_{bU,k^{'}} }
 \\ &
= \medmath{ -\frac{\Delta_{t}}{c^2} \sum_{k^{},u^{}k^{'},u^{'}} \underset{bu^{},k^{}}{\operatorname{SNR}} \underset{bu^{'},k^{'}}{\operatorname{SNR}} } \medmath{   \bm{\Delta}_{bu^{},k^{}}  (k^{'}) {\bm{\Delta}_{bu^{'},k^{'}}^{\mathrm{T}}} \omega_{bU,k} \omega_{bU,k^{'}}}    \medmath{   \left(\sum_{u,k} \underset{bu,k}{\operatorname{SNR}} \omega_{bU,k}\right)^{\mathrm{-1}}} \\&+ \sum_{k^{},u^{}k^{'},u^{'}}  \medmath{  \underset{bu^{},k^{}}{\operatorname{SNR}} \underset{bu^{'},k^{'}}{\operatorname{SNR}}  }  \medmath{{ \nabla_{\bm{p}_{U,0}} \nu_{bU,k^{}} \bm{\Delta}_{bU,k^{'}}^{\mathrm{T}}   }  \frac{(f_{c}^{2}) (\alpha_{obu,k^{}}^{2}\alpha_{obu^{'},k^{'}}^{2})}{4}   \left(\sum_{u,k} \frac{\underset{b^{}u,k^{}}{\operatorname{SNR}}  \alpha_{obu,k}^{2}}{2}\right)^{-1}} 
\end{split}
\end{align}
\end{figure*}

\subsubsection{Proof of Lemma \ref{lemma:information_loss_FIM_3D_velocity_3D_velocity}}
\label{Appendix_lemma_information_loss_FIM_3D_velocity_3D_velocity}
Using the definition of the information loss terms in Definition \ref{definition_EFIM}, we get the first equality (\ref{equ_lemma:information_loss_FIM_3D_3D_veloctiy_1}) and after substituting the appropriate FIM for channel parameters, we get the second equality.

\begin{figure*}
\begin{align}
\begin{split}
\label{equ_lemma:information_loss_FIM_3D_3D_veloctiy_1}
&\medmath{\bm{G}_{{{y} }}(\bm{y}_{}| \bm{\eta} ;\bm{v}_{U,0},\bm{v}_{U,0})  }= \\   &\medmath{\sum_{b}\Bigg[
 \Bigg[\sum_{k^{},u^{}} \bm{F}_{{{y} }}({y}_{bu,k}| {\eta} ;{\delta}_{bU},{\delta}_{bU}) \Bigg]^{-1} \norm{ \sum_{k^{},u^{}}   \bm{F}_{{{y} }}({y}_{bu,k}| {\eta} ;{\delta}_{bU},{\tau}_{bu,k}) \nabla_{\bm{v}_{U,0}}^{\mathrm{T}} \tau_{bu,k^{}}  }^{2}} \\ &+  \medmath{\Bigg[ \sum_{k,u}  \bm{F}_{{{y} }}({y}_{bu,k}| {\eta} ;{\epsilon}_{bU},{\epsilon}_{bU})\Bigg]^{-1}\norm{\sum_{k^{},u^{}}   \bm{F}_{{{y} }}({y}_{bu,k}| {\eta} ;{\epsilon}_{bU},{\nu}_{bU,k}) \nabla_{\bm{v}_{U,0}}^{\mathrm{T}} \nu_{bU,k^{}} }^2 }
 \\ &\medmath{+\Bigg[
 \Bigg[\sum_{q,k^{},u^{}} \bm{F}_{{{y} }}({y}_{qu,k}| {\eta} ;{\delta}_{QU},{\delta}_{QU}) \Bigg]^{-1} \norm{ \sum_{q,k^{},u^{}}   \bm{F}_{{{y} }}({y}_{qu,k}| {\eta} ;{\delta}_{QU},{\tau}_{qu,k}) \nabla_{\bm{v}_{U,0}}^{\mathrm{T}} \tau_{qu,k^{}}  }^{2}} \\ & \medmath{+  \Bigg[ \sum_{q,k,u}  \bm{F}_{{{y} }}({y}_{qu,k}| {\eta} ;{\epsilon}_{QU},{\epsilon}_{QU})\Bigg]^{-1}\norm{\sum_{q,k^{},u^{}}   \bm{F}_{{{y} }}({y}_{qu,k}| {\eta} ;{\epsilon}_{QU},{\nu}_{qU,k}) \nabla_{\bm{v}_{U,0}}^{\mathrm{T}} \nu_{qU,k^{}} }^2}
 \\ &
 =\medmath{\sum_{b}   \norm{\sum_{k^{},u^{}} \underset{bu^{},k^{}}{\operatorname{SNR}}\bm{\Delta}_{bu^{},k^{}}^{\mathrm{T}} \frac{\Delta_{t}(k) \omega_{bU,k}}{c}}}^{2}    \medmath{   \left(\sum_{u,k} \underset{bu,k}{\operatorname{SNR}} \omega_{bU,k}\right)^{\mathrm{-1}}  }  + \medmath{   \norm{\sum_{q,u^{},k^{}} \underset{qu^{},k^{}}{\operatorname{SNR}}\bm{\Delta}_{qu^{},k^{}}^{\mathrm{T}} \frac{ \Delta_{t}(k)\omega_{qU,k}}{c} }}^{2}    \medmath{   \left(\sum_{q,u,k} \underset{qu,k}{\operatorname{SNR}} \omega_{qU,k}\right)^{\mathrm{-1}}  } \\ &+  
\medmath{\sum_{b}\norm{{\sum_{u^{},k^{} } \underset{bu^{},k^{}}{\operatorname{SNR}} \; \; \bm{\Delta}_{b^{}U^{},k^{}}^{\mathrm{T}}      }  \frac{(f_{c}^{}) (\alpha_{obu,k^{}}^{2})}{2} }^{2}  \left(\sum_{u,k} \frac{\underset{bu,k}{\operatorname{SNR}}  \alpha_{obu,k}^{2}}{2}\right)^{-1}}  + \medmath{\norm{{\sum_{q^{},u,k^{} } \underset{q^{},u,k^{}}{\operatorname{SNR}} \; \; \bm{\Delta}_{q^{}U^{},k^{}}^{\mathrm{T}}   }  \frac{(f_{c}^{}) (\alpha_{oqu,k^{}}^{2})}{2} }^{2}  \left(\sum_{q,u,k} \frac{\underset{qu,k}{\operatorname{SNR}}  \alpha_{oqu,k}^{2}}{2}\right)^{-1}}
\end{split}
\end{align}
\end{figure*}

\subsubsection{Proof of Lemma \ref{lemma:information_loss_FIM_3D_velocity_3D_orientation}}
\label{Appendix_lemma_information_loss_FIM_3D_velocity_3D_orientation}
Using the definition of the information loss terms in Definition \ref{definition_EFIM}, we get the first equality (\ref{equ_lemma:information_loss_FIM_3D_velocity_3D_orientation_1}) and after substituting the appropriate FIM for channel parameters, we get the second equality.

\begin{figure*}
\begin{align}
\begin{split}
\label{equ_lemma:information_loss_FIM_3D_velocity_3D_orientation_1}
&\medmath{\bm{G}_{{{y} }}(\bm{y}_{}| \bm{\eta} ;\bm{v}_{U,0},\bm{\Phi}_{U})  }= \\   &\medmath{\sum_{b}\Bigg[
 \Bigg[\sum_{k^{},u^{}} \bm{F}_{{{y} }}({y}_{bu,k}| {\eta} ;{\delta}_{bU},{\delta}_{bU}) \Bigg]^{-1} \sum_{k^{'},u^{'}
 } \sum_{k^{},u^{}}   \bm{F}_{{{y} }}({y}_{bu,k}| {\eta} ;{\delta}_{bU},{\tau}_{bu,k}) \bm{F}_{{{y} }}({y}_{bu,k}| {\eta} ;{\delta}_{bU},{\tau}_{bu^{'},k^{'}}) \nabla_{\bm{v}_{U,0}} \tau_{bu,k^{}} \nabla_{\bm{\Phi}_{U}}^{\mathrm{T}} \tau_{bu^{'},k^{'}}  } 
 \\ &\medmath{+\Bigg[
 \Bigg[\sum_{q,k^{},u^{}} \bm{F}_{{{y} }}({y}_{qu,k}| {\eta} ;{\delta}_{QU},{\delta}_{QU}) \Bigg]^{-1}  \sum_{q^{'}}\sum_{q,k^{},u^{}}   \bm{F}_{{{y} }}({y}_{qu,k}| {\eta} ;{\delta}_{QU},{\tau}_{qu,k}) \bm{F}_{{{y} }}({y}_{qu,k}| {\eta} ;{\delta}_{QU},{\tau}_{q^{'}u^{'},k^{'}}) \nabla_{\bm{v}_{U,0}} \tau_{qu,k^{}} \nabla_{\bm{\Phi}_{U}}^{\mathrm{T}} \tau_{q^{'}u^{'},k^{'}}  } \\&
= \medmath{\frac{\Delta_{t}}{c}\sum_{b,k^{},u^{}k^{'},u^{'}} \underset{bu^{},k^{}}{\operatorname{SNR}} \underset{bu^{'},k^{'}}{\operatorname{SNR}} } \medmath{ (k)  \bm{\Delta}_{bu^{},k^{}}  \nabla_{\bm{\Phi}_{U}}^{\mathrm{T}} \tau_{bu^{'},k^{'}}  \omega_{bU,k} \omega_{bU,k^{'}}}    \medmath{   \left(\sum_{u,k} \underset{bu,k}{\operatorname{SNR}} \omega_{bU,k}\right)^{\mathrm{-1}}} \\ &+  \medmath{\frac{\Delta_{t}}{c}\sum_{q,q^{'},u^{},k^{}u^{'},k^{'}} \underset{q^{},u,k^{}}{\operatorname{SNR}} \underset{q^{'},u^{'},k^{'}}{\operatorname{SNR}} } \medmath{ (k)   \bm{\Delta}_{q^{}u^{},k^{}}   \nabla_{\bm{\Phi}_{U}}^{\mathrm{T}} \tau_{q^{'}u^{'},k^{'}}} \medmath{ \omega_{qU,k} \omega_{q^{'}U,k^{'}}}    \medmath{   \left(\sum_{qu,k} \underset{qu,k}{\operatorname{SNR}} \omega_{qU,k}\right)^{\mathrm{-1}}} 
\end{split}
\end{align}
\end{figure*}

\subsubsection{Proof of Lemma \ref{lemma:information_loss_FIM_3D_velocity_3D_bth_position_uncertainty}}
\label{Appendix_lemma_information_loss_FIM_3D_velocity_3D_bth_position_uncertainty}
Using the definition of the information loss terms in Definition \ref{definition_EFIM}, we get the first equality (\ref{equ_lemma:information_loss_FIM_3D_velocity_3D_bth_position_uncertainty_1}) and after substituting the appropriate FIM for channel parameters, we get the second equality.

\begin{figure*}
\begin{align}
\begin{split}
\label{equ_lemma:information_loss_FIM_3D_velocity_3D_bth_position_uncertainty_1}
&\medmath{\bm{G}_{{{y} }}(\bm{y}_{}| \bm{\eta} ;\bm{v}_{U,0},\check{\bm{p}}_{b,0})  }= \\   &\medmath{\Bigg[
 \Bigg[\sum_{k^{},u^{}} \bm{F}_{{{y} }}({y}_{bu,k}| {\eta} ;{\delta}_{bU},{\delta}_{bU}) \Bigg]^{-1} \sum_{k^{'},u^{'}
 } \sum_{k^{},u^{}}   \bm{F}_{{{y} }}({y}_{bu,k}| {\eta} ;{\delta}_{bU},{\tau}_{bu,k}) \bm{F}_{{{y} }}({y}_{bu,k}| {\eta} ;{\delta}_{bU},{\tau}_{bu^{'},k^{'}}) \nabla_{\bm{v}_{U,0}} \tau_{bu,k^{}} \nabla_{\check{\bm{p}}_{b,0}}^{\mathrm{T}} \tau_{bu^{'},k^{'}}  } \\ &+  \medmath{\Bigg[ \sum_{k,u}  \bm{F}_{{{y} }}({y}_{bu,k}| {\eta} ;{\epsilon}_{bU},{\epsilon}_{bU})\Bigg]^{-1}\sum_{k^{'},u^{'}}\sum_{k^{},u^{}}   \bm{F}_{{{y} }}({y}_{bu,k}| {\eta} ;{\epsilon}_{bU},{\nu}_{bU,k}) \bm{F}_{{{y} }}({y}_{bu,k}| {\eta} ;{\epsilon}_{bU},{\nu}_{bU,k^{'}}) \nabla_{\bm{v}_{U,0}} \nu_{bU,k^{}} \nabla_{\check{\bm{p}}_{b,0}}^{\mathrm{T}} \nu_{bU,k^{'}} }
 \\ &
= \medmath{\frac{-\Delta_{t}}{c^2} \sum_{k^{},u^{}k^{'},u^{'}} (k) \underset{bu^{},k^{}}{\operatorname{SNR}} \underset{bu^{'},k^{'}}{\operatorname{SNR}} } \medmath{   \bm{\Delta}_{bu^{},k^{}}  {\bm{\Delta}_{bu^{'},k^{'}}^{\mathrm{T}}} \omega_{bU,k} \omega_{bU,k^{'}}}    \medmath{   \left(\sum_{u,k} \underset{bu,k}{\operatorname{SNR}} \omega_{bU,k}\right)^{\mathrm{-1}}} \\ &- \frac{1}{c}\medmath{\sum_{k^{},u^{}k^{'},u^{'}}    \underset{bu^{},k^{}}{\operatorname{SNR}} \underset{bu^{'},k^{'}}{\operatorname{SNR}}  }  \medmath{{ \bm{\Delta}_{b^{}U^{},k^{}} \nabla_{\check{\bm{p}}_{b,0}}^{\mathrm{T}} \nu_{bU,k^{'}}   }  \frac{(f_{c}^{2}) (\alpha_{obu,k^{}}^{2}\alpha_{obu^{'},k^{'}}^{2})}{4}   \left(\sum_{u,k} \frac{\underset{b^{}u,k^{}}{\operatorname{SNR}}  \alpha_{obu,k}^{2}}{2}\right)^{-1}}  
\end{split}
\end{align}
\end{figure*}

\subsubsection{Proof of Lemma \ref{lemma:information_loss_FIM_3D_velocity_3D_bth_velocity_uncertainty}}
\label{Appendix_lemma_information_loss_FIM_3D_velocity_3D_bth_velocity_uncertainty}
Using the definition of the information loss terms in Definition \ref{definition_EFIM}, we get the first equality (\ref{equ_lemma:information_loss_FIM_3D_velocity_3D_bth_velocity_uncertainty_1}), and after substituting the appropriate FIM for channel parameters, we get the second equality.

\begin{figure*}
\begin{align}
\begin{split}
\label{equ_lemma:information_loss_FIM_3D_velocity_3D_bth_velocity_uncertainty_1}
&\medmath{\bm{G}_{{{y} }}(\bm{y}_{}| \bm{\eta} ;\bm{v}_{U,0},\check{\bm{v}}_{b,0})  }= \\   &\medmath{\Bigg[
 \Bigg[\sum_{k^{},u^{}} \bm{F}_{{{y} }}({y}_{bu,k}| {\eta} ;{\delta}_{bU},{\delta}_{bU}) \Bigg]^{-1} \sum_{k^{'},u^{'}
 } \sum_{k^{},u^{}}   \bm{F}_{{{y} }}({y}_{bu,k}| {\eta} ;{\delta}_{bU},{\tau}_{bu,k}) \bm{F}_{{{y} }}({y}_{bu,k}| {\eta} ;{\delta}_{bU},{\tau}_{bu^{'},k^{'}}) \nabla_{\bm{v}_{U,0}} \tau_{bu,k^{}} \nabla_{\check{\bm{v}}_{b,0}}^{\mathrm{T}} \tau_{bu^{'},k^{'}}  } \\ &+  \medmath{\Bigg[ \sum_{k,u}  \bm{F}_{{{y} }}({y}_{bu,k}| {\eta} ;{\epsilon}_{bU},{\epsilon}_{bU})\Bigg]^{-1}\sum_{k^{'},u^{'}}\sum_{k^{},u^{}}   \bm{F}_{{{y} }}({y}_{bu,k}| {\eta} ;{\epsilon}_{bU},{\nu}_{bU,k}) \bm{F}_{{{y} }}({y}_{bu,k}| {\eta} ;{\epsilon}_{bU},{\nu}_{bU,k^{'}}) \nabla_{\bm{v}_{U,0}} \nu_{bU,k^{}} \nabla_{\check{\bm{v}}_{b,0}}^{\mathrm{T}} \nu_{bU,k^{'}} }
 \\ &
=  -\medmath{\frac{\Delta_{t}^{2}}{c^2} \sum_{k^{},u^{}k^{'},u^{'}} \underset{bu^{},k^{}}{\operatorname{SNR}} \underset{bu^{'},k^{'}}{\operatorname{SNR}} } \medmath{   \bm{\Delta}_{bu^{},k^{}} (k) (k^{'}) {\bm{\Delta}_{bu^{'},k^{'}}^{\mathrm{T}}} \omega_{bU,k} \omega_{bU,k^{'}}}    \medmath{   \left(\sum_{u,k} \underset{bu,k}{\operatorname{SNR}} \omega_{bU,k}\right)^{\mathrm{-1}}} \\ -& \frac{1}{c^2}\medmath{  \sum_{k^{},u^{}k^{'},u^{'}}}  \medmath{  \underset{bu^{},k^{}}{\operatorname{SNR}} \underset{bu^{'},k^{'}}{\operatorname{SNR}}  }  \medmath{{ \bm{\Delta}_{bU,k^{}} \bm{\Delta}_{bU,k^{'}}^{\mathrm{T}}   }  \frac{(f_{c}^{2}) (\alpha_{obu,k^{}}^{2}\alpha_{obu^{'},k^{'}}^{2})}{4}   \left(\sum_{u,k} \frac{\underset{b^{}u,k^{}}{\operatorname{SNR}}  \alpha_{obu,k}^{2}}{2}\right)^{-1}}  
\end{split}
\end{align}
\end{figure*}

\subsubsection{Proof of Lemma \ref{lemma:information_loss_FIM_3D_orientation_3D_orientation}}
\label{Appendix_lemma_information_loss_FIM_3D_orientation}
Using the definition of the information loss terms in Definition \ref{definition_EFIM}, we get the first equality (\ref{equ_lemma:information_loss_FIM_3D_orientation_1}) and after substituting the appropriate FIM for channel parameters, we get the second equality.

\begin{figure*}
\begin{align}
\begin{split}
\label{equ_lemma:information_loss_FIM_3D_orientation_1}
&\medmath{\bm{G}_{{{y} }}(\bm{y}_{}| \bm{\eta} ;\bm{\Phi}_{U},\bm{\Phi}_{U})  }= \\   &\medmath{\sum_{b}\Bigg[
 \Bigg[\sum_{k^{},u^{}} \bm{F}_{{{y} }}({y}_{bu,k}| {\eta} ;{\delta}_{bU},{\delta}_{bU}) \Bigg]^{-1} \norm{ \sum_{k^{},u^{}}   \bm{F}_{{{y} }}({y}_{bu,k}| {\eta} ;{\delta}_{bU},{\tau}_{bu,k}) \nabla_{\bm{\Phi}_{U}}^{\mathrm{T}} \tau_{bu,k^{}}  }^{2}}
 \\ &\medmath{+\Bigg[
 \Bigg[\sum_{q,k^{},u^{}} \bm{F}_{{{y} }}({y}_{qu,k}| {\eta} ;{\delta}_{QU},{\delta}_{QU}) \Bigg]^{-1} \norm{ \sum_{q,k^{},u^{}}   \bm{F}_{{{y} }}({y}_{qu,k}| {\eta} ;{\delta}_{QU},{\tau}_{qu,k}) \nabla_{\bm{\Phi}_{U}}^{\mathrm{T}} \tau_{qu,k^{}}  }^{2}} 
 \\ &
 =\sum_{b}   \norm{\sum_{k^{},u^{}} \underset{bu^{},k^{}}{\operatorname{SNR}} \; \omega_{bU,k} \nabla_{\bm{\Phi}_{U}}^{\mathrm{T}} \tau_{bu^{},k^{}}   }^{2}   \medmath{   \left(\sum_{u,k} \underset{bu,k}{\operatorname{SNR}} \; \omega_{bU,k}\right)^{\mathrm{-1}}  }  + 
\medmath{  \norm{\sum_{q^{},k^{},u^{}} \underset{qu^{},k^{}}{\operatorname{SNR}} \; \omega_{qU,k} \nabla_{\bm{\Phi}_{U}}^{\mathrm{T}} \tau_{qu^{},k^{}}   }}^{2}    \medmath{   \left(\sum_{q,u,k} \underset{qu,k}{\operatorname{SNR}} \; \omega_{qu,k}\right)^{\mathrm{-1}}  } 
\end{split}
\end{align}
\end{figure*}

\subsubsection{Proof of Lemma \ref{lemma:information_loss_3D_orientation_3D_bth_position_uncertainty}}
\label{Appendix_lemma_information_loss_FIM_3D_orientation_3D_bth_position_uncertainty}
Using the definition of the information loss terms in Definition \ref{definition_EFIM}, we get the first equality (\ref{equ_lemma:information_loss_FIM_3D_orientation_3D_bth_position_uncertainty_1}) and after substituting the appropriate FIM for channel parameters, we get the second equality.

\begin{figure*}
\begin{align}
\begin{split}
\label{equ_lemma:information_loss_FIM_3D_orientation_3D_bth_position_uncertainty_1}
&\medmath{\bm{G}_{{{y} }}(\bm{y}_{}| \bm{\eta} ;\bm{\Phi}_{U},\check{\bm{p}}_{b,0})  }= \\   &\medmath{\Bigg[
 \Bigg[\sum_{k^{},u^{}} \bm{F}_{{{y} }}({y}_{bu,k}| {\eta} ;{\delta}_{bU},{\delta}_{bU}) \Bigg]^{-1} \sum_{k^{'},u^{'}
 } \sum_{k^{},u^{}}   \bm{F}_{{{y} }}({y}_{bu,k}| {\eta} ;{\delta}_{bU},{\tau}_{bu,k}) \bm{F}_{{{y} }}({y}_{bu,k}| {\eta} ;{\delta}_{bU},{\tau}_{bu^{'},k^{'}}) \nabla_{\bm{\Phi}_{U}} \tau_{bu,k^{}} \nabla_{\check{\bm{p}}_{b,0}}^{\mathrm{T}} \tau_{bu^{'},k^{'}}  } \\ &
=  \medmath{\frac{-1}{c} \sum_{k^{},u^{}k^{'},u^{'}} \underset{bu^{},k^{}}{\operatorname{SNR}} \underset{bu^{'},k^{'}}{\operatorname{SNR}} } \medmath{   \nabla_{\bm{\Phi}_{U}} \tau_{bu^{},k^{}}  {\bm{\Delta}_{bu^{'},k^{'}}^{\mathrm{T}}} \omega_{bU,k} \omega_{bU,k^{'}}}    \medmath{   \left(\sum_{u,k} \underset{bu,k}{\operatorname{SNR}} \omega_{bU,k}\right)^{\mathrm{-1}}}
\end{split}
\end{align}
\end{figure*}

\subsubsection{Proof of Lemma \ref{lemma:information_loss_FIM_3D_orientation_3D_bth_velocity_uncertainty}}
\label{Appendix_lemma_information_loss_FIM_3D_orientation_3D_bth_velocity_uncertainty}
Using the definition of the information loss terms in Definition \ref{definition_EFIM}, we get the first equality (\ref{equ_lemma:information_loss_FIM_3D_orientation_3D_bth_velocity_uncertainty_1}) and after substituting the appropriate FIM for channel parameters, we get the second equality.

\begin{figure*}
\begin{align}
\begin{split}
\label{equ_lemma:information_loss_FIM_3D_orientation_3D_bth_velocity_uncertainty_1}
&\medmath{\bm{G}_{{{y} }}(\bm{y}_{}| \bm{\eta} ;\bm{\Phi}_{U},\check{\bm{v}}_{b,0})  }= \\   &\medmath{\Bigg[
 \Bigg[\sum_{k^{},u^{}} \bm{F}_{{{y} }}({y}_{bu,k}| {\eta} ;{\delta}_{bU},{\delta}_{bU}) \Bigg]^{-1} \sum_{k^{'},u^{'}
 } \sum_{k^{},u^{}}   \bm{F}_{{{y} }}({y}_{bu,k}| {\eta} ;{\delta}_{bU},{\tau}_{bu,k}) \bm{F}_{{{y} }}({y}_{bu,k}| {\eta} ;{\delta}_{bU},{\tau}_{bu^{'},k^{'}}) \nabla_{\bm{\Phi}_{U}} \tau_{bu,k^{}} \nabla_{\check{\bm{v}}_{b,0}}^{\mathrm{T}} \tau_{bu^{'},k^{'}}  } 
 \\ &
= \medmath{ -\frac{\Delta_{t}}{c} \sum_{k^{},u^{}k^{'},u^{'}} \underset{bu^{},k^{}}{\operatorname{SNR}} \underset{bu^{'},k^{'}}{\operatorname{SNR}} } \medmath{   \nabla_{\bm{\Phi}_{U}} \tau_{bu,k^{}}   (k^{'}) {\bm{\Delta}_{bu^{'},k^{'}}^{\mathrm{T}}} \omega_{bU,k} \omega_{bU,k^{'}}}    \medmath{   \left(\sum_{u,k} \underset{bu,k}{\operatorname{SNR}} \omega_{bU,k}\right)^{\mathrm{-1}}} 
\end{split}
\end{align}
\end{figure*}

\subsubsection{Proof of Lemma \ref{lemma:information_loss_FIM_3D_bth_position_3D_bth_position}}
\label{Appendix_lemma_information_loss_FIM_3D_bth_position_3D_bth_position}
Using the definition of the information loss terms in Definition \ref{definition_EFIM}, we get the first equality (\ref{equ_lemma:information_loss_FIM_3D_bth_position_3D_bth_position_1}) and after substituting the appropriate FIM for channel parameters, we get the second equality.

\begin{figure*}
\begin{align}
\begin{split}
\label{equ_lemma:information_loss_FIM_3D_bth_position_3D_bth_position_1}
&\medmath{\bm{G}_{{{y} }}(\bm{y}_{}| \bm{\eta} ;\check{\bm{p}}_{b,0},\check{\bm{p}}_{b,0})  }= \\   &\medmath{\Bigg[
 \Bigg[\sum_{k^{},u^{}} \bm{F}_{{{y} }}({y}_{bu,k}| {\eta} ;{\delta}_{bU},{\delta}_{bU}) \Bigg]^{-1} \sum_{k^{'},u^{'}
 } \sum_{k^{},u^{}}   \bm{F}_{{{y} }}({y}_{bu,k}| {\eta} ;{\delta}_{bU},{\tau}_{bu,k}) \bm{F}_{{{y} }}({y}_{bu,k}| {\eta} ;{\delta}_{bU},{\tau}_{bu^{'},k^{'}}) \nabla_{\check{\bm{p}}_{b,0}} \tau_{bu,k^{}} \nabla_{\check{\bm{p}}_{b,0}}^{\mathrm{T}} \tau_{bu^{'},k^{'}}  } \\ &+  \medmath{\Bigg[ \sum_{k,u}  \bm{F}_{{{y} }}({y}_{bu,k}| {\eta} ;{\epsilon}_{bU},{\epsilon}_{bU})\Bigg]^{-1}\sum_{k^{'},u^{'}}\sum_{k^{},u^{}}   \bm{F}_{{{y} }}({y}_{bu,k}| {\eta} ;{\epsilon}_{bU},{\nu}_{bU,k}) \bm{F}_{{{y} }}({y}_{bu,k}| {\eta} ;{\epsilon}_{bU},{\nu}_{bU,k^{'}}) \nabla_{\check{\bm{p}}_{b,0}} \nu_{bU,k^{}} \nabla_{\check{\bm{p}}_{b,0}}^{\mathrm{T}} \nu_{bU,k^{'}} }
 \\ &\medmath{+\Bigg[
 \Bigg[\sum_{q,k^{}} \bm{F}_{{{y} }}({y}_{bq,k}| {\eta} ;{\delta}_{bQ},{\delta}_{bQ}) \Bigg]^{-1}  \sum_{q^{'},k^{'}}\sum_{q,k^{}}   \bm{F}_{{{y} }}({y}_{bq,k}| {\eta} ;{\delta}_{bQ},{\tau}_{bq,k}) \bm{F}_{{{y} }}({y}_{bq,k}| {\eta} ;{\delta}_{bQ},{\tau}_{bq^{'},k^{'}}) \nabla_{\check{\bm{p}}_{b,0}} \tau_{bq,k^{}} \nabla_{\check{\bm{p}}_{b,0}}^{\mathrm{T}} \tau_{bq^{'},k^{'}}  } \\ & \medmath{+  \Bigg[ \sum_{q,k}  \bm{F}_{{{y} }}({y}_{bq,k}| {\eta} ;{\epsilon}_{bQ},{\epsilon}_{bQ})\Bigg]^{-1}\sum_{q^{'},k^{'}}\sum_{q,k^{}}   \bm{F}_{{{y} }}({y}_{bq,k}| {\eta} ;{\epsilon}_{bQ},{\nu}_{bq,k}) \bm{F}_{{{y} }}({y}_{bq,k}| {\eta} ;{\epsilon}_{bQ},{\nu}_{bq^{'},k^{'}}) \nabla_{\check{\bm{p}}_{b,0}} \nu_{bq,k^{}} \nabla_{\check{\bm{p}}_{b,0}}^{\mathrm{T}} \nu_{bq^{'},k^{'}}}
 \\ &
 = \medmath{  \norm{\sum_{k^{},u^{}} \underset{bu^{},k^{}}{\operatorname{SNR}}\bm{\Delta}_{bu^{},k^{}}^{\mathrm{T}} \frac{ \omega_{bU,k}}{c} }}^{2}    \medmath{   \left(\sum_{u,k} \underset{bu,k}{\operatorname{SNR}} \omega_{bU,k}\right)^{\mathrm{-1}}  }    +
\medmath{   \norm{\sum_{k^{},q^{}} \underset{bq^{},k^{}}{\operatorname{SNR}}\bm{\Delta}_{bq^{},k^{}}^{\mathrm{T}} \frac{ \omega_{bq,k}}{c} }}^{2}    \medmath{   \left(\sum_{q,k} \underset{bu,k}{\operatorname{SNR}} \omega_{bq,k}\right)^{\mathrm{-1}}  } \\&+  
\medmath{\norm{{\sum_{u^{},k^{} } \underset{bu^{},k^{}}{\operatorname{SNR}} \; \; \nabla_{\check{\bm{p}}_{b,0}}^{\mathrm{T}} \nu_{bU,k^{}}   }  \frac{(f_{c}^{}) (\alpha_{obu,k^{}}^{2})}{2} }^{2}  \left(\sum_{u,k} \frac{\underset{bu,k}{\operatorname{SNR}}  \alpha_{obu,k}^{2}}{2}\right)^{-1}}  + \medmath{\norm{{\sum_{q^{},k^{} } \underset{bq^{},k^{}}{\operatorname{SNR}} \; \; \nabla_{\check{\bm{p}}_{b,0}}^{\mathrm{T}} \nu_{bq,k^{}}   }  \frac{(f_{c}^{}) (\alpha_{obq,k^{}}^{2})}{2} }^{2}  \left(\sum_{q,k} \frac{\underset{bu,k}{\operatorname{SNR}}  \alpha_{obq,k}^{2}}{2}\right)^{-1}} 
\end{split}
\end{align}
\end{figure*}

\subsubsection{Proof of Lemma \ref{lemma:information_loss_FIM_3D_bth_position_3D_bth_velocity_uncertainty}}
\label{Appendix_lemma_information_loss_FIM_3D_bth_position_3D_bth_veloctiy}
Using the definition of the information loss terms in Definition \ref{definition_EFIM}, we get the first equality (\ref{equ_lemma:information_loss_FIM_3D_bth_position_3D_bth_veloctiy_1}) and after substituting the appropriate FIM for channel parameters, we get the second equality.

\begin{figure*}
\begin{align}
\begin{split}
\label{equ_lemma:information_loss_FIM_3D_bth_position_3D_bth_veloctiy_1}
&\medmath{\bm{G}_{{{y} }}(\bm{y}_{}| \bm{\eta} ;\check{\bm{p}}_{b,0},\check{\bm{v}}_{b,0})  }= \\   &\medmath{\Bigg[
 \Bigg[\sum_{k^{},u^{}} \bm{F}_{{{y} }}({y}_{bu,k}| {\eta} ;{\delta}_{bU},{\delta}_{bU}) \Bigg]^{-1} \sum_{k^{'},u^{'}
 } \sum_{k^{},u^{}}   \bm{F}_{{{y} }}({y}_{bu,k}| {\eta} ;{\delta}_{bU},{\tau}_{bu,k}) \bm{F}_{{{y} }}({y}_{bu,k}| {\eta} ;{\delta}_{bU},{\tau}_{bu^{'},k^{'}}) \nabla_{\check{\bm{p}}_{b,0}} \tau_{bu,k^{}} \nabla_{\check{\bm{v}}_{b,0}}^{\mathrm{T}} \tau_{bu^{'},k^{'}}  } \\ &+  \medmath{\Bigg[ \sum_{k,u}  \bm{F}_{{{y} }}({y}_{bu,k}| {\eta} ;{\epsilon}_{bU},{\epsilon}_{bU})\Bigg]^{-1}\sum_{k^{'},u^{'}}\sum_{k^{},u^{}}   \bm{F}_{{{y} }}({y}_{bu,k}| {\eta} ;{\epsilon}_{bU},{\nu}_{bU,k}) \bm{F}_{{{y} }}({y}_{bu,k}| {\eta} ;{\epsilon}_{bU},{\nu}_{bU,k^{'}}) \nabla_{\check{\bm{p}}_{b,0}} \nu_{bU,k^{}} \nabla_{\check{\bm{v}}_{b,0}}^{\mathrm{T}} \nu_{bU,k^{'}} }
 \\ &\medmath{+\Bigg[
 \Bigg[\sum_{q,k^{}} \bm{F}_{{{y} }}({y}_{bq,k}| {\eta} ;{\delta}_{bQ},{\delta}_{bQ}) \Bigg]^{-1}  \sum_{q^{'},k^{'}}\sum_{q,k^{}}   \bm{F}_{{{y} }}({y}_{bq,k}| {\eta} ;{\delta}_{bQ},{\tau}_{bq,k}) \bm{F}_{{{y} }}({y}_{bq,k}| {\eta} ;{\delta}_{bQ},{\tau}_{bq^{'},k^{'}}) \nabla_{\check{\bm{p}}_{b,0}} \tau_{bq,k^{}} \nabla_{\check{\bm{v}}_{b,0}}^{\mathrm{T}} \tau_{bq^{'},k^{'}}  } \\ & \medmath{+  \Bigg[ \sum_{q,k}  \bm{F}_{{{y} }}({y}_{bq,k}| {\eta} ;{\epsilon}_{bQ},{\epsilon}_{bQ})\Bigg]^{-1}\sum_{q^{'},k^{'}}\sum_{q,k^{}}   \bm{F}_{{{y} }}({y}_{bq,k}| {\eta} ;{\epsilon}_{bQ},{\nu}_{bq,k}) \bm{F}_{{{y} }}({y}_{bq,k}| {\eta} ;{\epsilon}_{bQ},{\nu}_{bq^{'},k^{'}}) \nabla_{\check{\bm{p}}_{b,0}} \nu_{bq,k^{}} \nabla_{\check{\bm{v}}_{b,0}}^{\mathrm{T}} \nu_{bq^{'},k^{'}}}
 \\ &
 = \medmath{ \frac{\Delta_{t}}{c^2} \sum_{k^{},u^{}k^{'},u^{'}} \underset{bu^{},k^{}}{\operatorname{SNR}} \underset{bu^{'},k^{'}}{\operatorname{SNR}} } \medmath{   \bm{\Delta}_{bu^{},k^{}}  (k^{'}) {\bm{\Delta}_{bu^{'},k^{'}}^{\mathrm{T}}} \omega_{bU,k} \omega_{bU,k^{'}}}    \medmath{   \left(\sum_{u,k} \underset{bu,k}{\operatorname{SNR}} \omega_{bU,k}\right)^{\mathrm{-1}}}  \\ & +  \medmath{ \frac{\Delta_{t}}{c^2} \sum_{k^{},q^{}k^{'},q^{'}} \underset{bq^{},k^{}}{\operatorname{SNR}} \underset{bq^{'},k^{'}}{\operatorname{SNR}} } \medmath{   \bm{\Delta}_{bq^{},k^{}}  (k^{'}) {\bm{\Delta}_{bq^{'},k^{'}}^{\mathrm{T}}} \omega_{bq,k} \omega_{bq^{'},k^{'}}}    \medmath{   \left(\sum_{q,k} \underset{bq,k}{\operatorname{SNR}} \omega_{bq,k}\right)^{\mathrm{-1}}} \\ &+ \frac{1}{c} \medmath{  \sum_{k^{},u^{}k^{'},u^{'}} \underset{bu^{},k^{}}{\operatorname{SNR}} \underset{bu^{'},k^{'}}{\operatorname{SNR}}  }  \medmath{{ \nabla_{\check{\bm{p}}_{b,0}} \nu_{bU,k^{}} \bm{\Delta}_{bU,k^{'}}^{\mathrm{T}}   }  \frac{(f_{c}^{2}) (\alpha_{obu,k^{}}^{2}\alpha_{obu^{'},k^{'}}^{2})}{4}   \left(\sum_{u,k} \frac{\underset{b^{}u,k^{}}{\operatorname{SNR}}  \alpha_{obu,k}^{2}}{2}\right)^{-1}} \\ & + \frac{1}{c}\medmath{  \sum_{k^{},q^{}k^{'},q^{'}} \underset{bq^{},k^{}}{\operatorname{SNR}} \underset{bq^{'},k^{'}}{\operatorname{SNR}}  }  \medmath{{ \nabla_{\check{\bm{p}}_{b,0}} \nu_{bq,k^{}} \bm{\Delta}_{bq^{'},k^{'}}^{\mathrm{T}}   }  \frac{(f_{c}^{2}) (\alpha_{obu,k^{}}^{2}\alpha_{obq^{'},k^{'}}^{2})}{4}   \left(\sum_{q,k} \frac{\underset{b^{}q,k^{}}{\operatorname{SNR}}  \alpha_{obq,k}^{2}}{2}\right)^{-1}} 
\end{split}
\end{align}
\end{figure*}

\subsubsection{Proof of Lemma \ref{lemma:information_loss_FIM_3D_bth_velocity_3D_bth_velocity}}
\label{Appendix_lemma_information_loss_FIM_3D_bth_velocity_3D_bth_velocity}
Using the definition of the information loss terms in Definition \ref{definition_EFIM}, we get the first equality (\ref{equ_lemma:information_loss_FIM_3D_bth_velocity_3D_bth_velocity_1}) and after substituting the appropriate FIM for channel parameters, we get the second equality.

\begin{figure*}
\begin{align}
\begin{split}
\label{equ_lemma:information_loss_FIM_3D_bth_velocity_3D_bth_velocity_1}
&\medmath{\bm{G}_{{{y} }}(\bm{y}_{}| \bm{\eta} ;\check{\bm{v}}_{b,0},\check{\bm{v}}_{b,0})  }= \\   &\medmath{\Bigg[
 \Bigg[\sum_{k^{},u^{}} \bm{F}_{{{y} }}({y}_{bu,k}| {\eta} ;{\delta}_{bU},{\delta}_{bU}) \Bigg]^{-1} \sum_{k^{'},u^{'}
 } \sum_{k^{},u^{}}   \bm{F}_{{{y} }}({y}_{bu,k}| {\eta} ;{\delta}_{bU},{\tau}_{bu,k}) \bm{F}_{{{y} }}({y}_{bu,k}| {\eta} ;{\delta}_{bU},{\tau}_{bu^{'},k^{'}}) \nabla_{\check{\bm{v}}_{b,0}} \tau_{bu,k^{}} \nabla_{\check{\bm{v}}_{b,0}}^{\mathrm{T}} \tau_{bu^{'},k^{'}}  } \\ &+  \medmath{\Bigg[ \sum_{k,u}  \bm{F}_{{{y} }}({y}_{bu,k}| {\eta} ;{\epsilon}_{bU},{\epsilon}_{bU})\Bigg]^{-1}\sum_{k^{'},u^{'}}\sum_{k^{},u^{}}   \bm{F}_{{{y} }}({y}_{bu,k}| {\eta} ;{\epsilon}_{bU},{\nu}_{bU,k}) \bm{F}_{{{y} }}({y}_{bu,k}| {\eta} ;{\epsilon}_{bU},{\nu}_{bU,k^{'}}) \nabla_{\check{\bm{v}}_{b,0}} \nu_{bU,k^{}} \nabla_{\check{\bm{v}}_{b,0}}^{\mathrm{T}} \nu_{bU,k^{'}} }
 \\ &\medmath{+\Bigg[
 \Bigg[\sum_{q,k^{}} \bm{F}_{{{y} }}({y}_{bq,k}| {\eta} ;{\delta}_{bQ},{\delta}_{bQ}) \Bigg]^{-1}  \sum_{q^{'},k^{'}}\sum_{q,k^{}}   \bm{F}_{{{y} }}({y}_{bq,k}| {\eta} ;{\delta}_{bQ},{\tau}_{bq,k}) \bm{F}_{{{y} }}({y}_{bq,k}| {\eta} ;{\delta}_{bQ},{\tau}_{bq^{'},k^{'}}) \nabla_{\check{\bm{v}}_{b,0}} \tau_{bq,k^{}} \nabla_{\check{\bm{v}}_{b,0}}^{\mathrm{T}} \tau_{bq^{'},k^{'}}  } \\ & \medmath{+  \Bigg[ \sum_{q,k}  \bm{F}_{{{y} }}({y}_{bq,k}| {\eta} ;{\epsilon}_{bQ},{\epsilon}_{bQ})\Bigg]^{-1}\sum_{q^{'},k^{'}}\sum_{q,k^{}}   \bm{F}_{{{y} }}({y}_{bq,k}| {\eta} ;{\epsilon}_{bQ},{\nu}_{bq,k}) \bm{F}_{{{y} }}({y}_{bq,k}| {\eta} ;{\epsilon}_{bQ},{\nu}_{bq^{'},k^{'}}) \nabla_{\check{\bm{v}}_{b,0}} \nu_{bq,k^{}} \nabla_{\check{\bm{v}}_{b,0}}^{\mathrm{T}} \nu_{bq^{'},k^{'}}}
 \\ &
 = \medmath{   \norm{\sum_{k^{},u^{}} \underset{bu^{},k^{}}{\operatorname{SNR}}\bm{\Delta}_{bu^{},k^{}}^{\mathrm{T}} \frac{ (k) \Delta_{t}\omega_{bU,k}}{c} }}^{2}    \medmath{   \left(\sum_{u,k} \underset{bu,k}{\operatorname{SNR}} \omega_{bU,k}\right)^{\mathrm{-1}}  }  +
 \medmath{   \norm{\sum_{k^{},q^{}} \underset{bq^{},k^{}}{\operatorname{SNR}}\bm{\Delta}_{bq^{},k^{}}^{\mathrm{T}} \frac{ (k) \Delta_{t}\omega_{bq,k}}{c} }}^{2}    \medmath{   \left(\sum_{q,k} \underset{bq,k}{\operatorname{SNR}} \omega_{bq,k}\right)^{\mathrm{-1}}  }\\& +  
\medmath{\norm{{\sum_{u^{},k^{} } \underset{bu^{},k^{}}{\operatorname{SNR}} \; \; \bm{\Delta}_{bU^{},k^{}}^{\mathrm{T}}   }  \frac{(f_{c}^{}) (\alpha_{obu,k^{}}^{2})}{2} }^{2}  \left(\sum_{u,k} \frac{\underset{bu,k}{\operatorname{SNR}}  \alpha_{obu,k}^{2}}{2}\right)^{-1}} 
+ \medmath{\norm{{\sum_{q^{},k^{} } \underset{bq^{},k^{}}{\operatorname{SNR}} \; \; \bm{\Delta}_{bq^{},k^{}}^{\mathrm{T}}   }  \frac{(f_{c}^{}) (\alpha_{obq,k^{}}^{2})}{2} }^{2}  \left(\sum_{q,k} \frac{\underset{bq,k}{\operatorname{SNR}}  \alpha_{obq,k}^{2}}{2}\right)^{-1}} 
\end{split}
\end{align}
\end{figure*}

{
\bibliographystyle{IEEEtran}
\bibliography{refs}
}
\end{document}